%% file: article-JHEP.tex
\documentclass[12pt]{article}
\pdfoutput=1

\usepackage[utf8]{inputenc}
\usepackage{microtype}
\usepackage{graphicx}
\usepackage{mathtools}
\usepackage[usenames,dvipsnames]{xcolor}
\usepackage{amssymb,amsmath,amsthm}
\usepackage{amsfonts}
\usepackage{shellesc}
\usepackage{stmaryrd}
\usepackage{xspace}
\usepackage{blkarray}
\usepackage{multicol}
\usepackage{bm}
\usepackage{url}
\usepackage{tikz}
\usepackage{pgfplots}
\usepackage{comment}
\usepackage{cite}
\usepackage{breqn}
\usepackage{textgreek}
\usepackage{xfrac}
\usepackage{enumitem}
\usepackage[cachedir={_minted-JHEP},finalizecache=false,frozencache=true]{minted}
\setminted[julia]{frame=lines,
rulecolor=\color{white!80!black},
fontsize=\small,
numbers=right,
numbersep=-5pt,
obeytabs=true,
tabsize=4}

\mathtoolsset{showonlyrefs}

\usepackage{caption}
\usepackage{subcaption}
\usepackage[export]{adjustbox}
\usepackage{jheppub}

\usepackage{hyperref}
\hypersetup{
    colorlinks=true,
    linkcolor=orange!80!black,
    filecolor=magenta,      
    urlcolor=RoyalBlue,
    citecolor=RoyalBlue,
}

\input{macros.tex}

\usepackage{chngcntr}
\counterwithout{equation}{section}
\counterwithout{proposition}{section}
\counterwithout{theorem}{section}
\counterwithout{lemma}{section}
\counterwithout{remark}{section}
\counterwithout{corollary}{section}
\counterwithout{definition}{section}
\counterwithout{claim}{section}
\counterwithout{figure}{section}
\counterwithout{table}{section}

\renewcommand{\Re}{\operatorname{Re}}
\renewcommand{\Im}{\operatorname{Im}}

\newcommand{\be}{\begin{equation}}
\newcommand{\ee}{\end{equation}}
\newcommand{\nn}{\nonumber}
\renewcommand{\I}{\mathcal{I}}
\renewcommand{\R}{\mathbb{R}}
\newcommand{\RR}{\mathrm{R}}
\renewcommand{\E}{\mathrm{E}}
\renewcommand{\d}{\mathrm{d}}
\newcommand{\U}{\mathcal{U}}
\renewcommand{\F}{\mathcal{F}}
\renewcommand{\V}{\mathcal{V}}
\renewcommand{\D}{\mathrm{D}}
\renewcommand{\N}{\mathrm{N}}
\renewcommand{\L}{\mathrm{L}}
\newcommand{\eps}{\varepsilon}
\renewcommand{\Z}{\mathbb{Z}}
\renewcommand{\m}{\mathsf{m}}
\renewcommand{\M}{\mathsf{M}}
\renewcommand{\C}{\mathbb{C}}
\renewcommand{\K}{\mathcal{K}}
\newcommand{\UU}{\mathbf{U}}
\newcommand{\FF}{\mathbf{F}}
\renewcommand{\Newt}{\mathrm{Newt}}
\newcommand{\conv}{\textup{Conv}}

\renewcommand{\v}[1]{\partial_v{#1}}
\newcommand{\gam}[1]{\partial_\gamma{#1}}

\allowdisplaybreaks
\graphicspath{{figures/}}
\title{\huge Landau Discriminants}

\author[1]{Sebastian Mizera,}\emailAdd{smizera@ias.edu}
\author[2]{Simon Telen}\emailAdd{simon.telen@mis.mpg.de}

\affiliation[1]{Institute for Advanced Study, Einstein Drive, Princeton, NJ 08540, USA}
\affiliation[2]{Max Planck Institute for Mathematics in the Sciences, Inselstra\ss e 22, 04103 Leipzig,\\ Germany}

\arxivnumber{2109.08036}

\abstract{%
	Scattering amplitudes in quantum field theories have intricate analytic properties as functions of the energies and momenta of the scattered particles. In perturbation theory, their singularities are governed by a set of nonlinear polynomial equations, known as \emph{Landau equations}, for each individual Feynman diagram. The singularity locus of the associated Feynman integral is made precise with the notion of the \emph{Landau discriminant}, which characterizes when the Landau equations admit a solution. In order to compute this discriminant, we present approaches from classical elimination theory, as well as a numerical algorithm based on homotopy continuation. These methods allow us to compute Landau discriminants of various Feynman diagrams up to $3$ loops, which were previously out of reach. For instance, the Landau discriminant of the envelope diagram is a reducible surface of degree $45$ in the three-dimensional space of kinematic invariants. We investigate geometric properties of the Landau discriminant, such as irreducibility, dimension and degree. In particular, we find simple examples in which the Landau discriminant has codimension greater than one. Furthermore, we describe a numerical procedure for determining which parts of the Landau discriminant lie in the physical regions. In order to study degenerate limits of Landau equations and bounds on the degree of the Landau discriminant, we introduce \emph{Landau polytopes} and study their facet structure. Finally, we provide an efficient numerical algorithm for the computation of the number of master integrals based on the connection to algebraic statistics. The algorithms used in this work are implemented in the open-source \texttt{Julia} package \texttt{Landau.jl} available at \url{https://mathrepo.mis.mpg.de/Landau/}.
}

\begin{document}

\maketitle


\setcounter{page}{2}
\setcounter{tocdepth}{4}

\section{Introduction}
\input{introduction}

\section{\label{sec:discriminant}Landau Analysis of Feynman Integrals}
\input{section2}

\section{\label{sec:computing}Computing Landau Discriminants}
\input{section3}

\section{\label{sec:polytopes}Landau Polytopes}
\input{section4}

\section{\label{sec:counting}Counting the Number of Master Integrals}
\input{section5}

\section{\label{sec:conclusion}Conclusion and Outlook}

In this work we introduced the Landau discriminant $\nabla_G$ of a Feynman diagram $G$. This formalises the notion of the singularity locus of Feynman integrals in the kinematic space, with a view towards explicit computations. In particular, we perform the Landau analysis from the point of view of nonlinear algebra. We proved that the Landau discriminant is an irreducible, strict subvariety of the kinematic space, and present examples where it has codimension $>1$. Our symbolic and numerical methods allow us to compute the dimension, degree and (with some simplifying assumptions on the kinematic parameters) the defining equation of the Landau discriminant for nontrivial examples that were previously out of reach. We studied combinatorial properties of convex polytopes related to Feynman diagrams and Landau equations, and provided bounds on the degree of $\nabla_G$ via ${\cal A}$-discriminants and mixed toric resultants. Finally, we showed how to use numerical nonlinear algebra to compute the number of master integrals of $G$. 

A number of open questions remain, especially about estimating the complexity of the discriminant for a general diagram $G$. For example, while the computation of the dimension and degree of $\nabla_G$ can be made efficient in the examples we studied in this paper (see Tab.~\ref{tab:dimdeg}), it would be interesting to know whether they can be determined (or at least estimated) purely from the combinatorics of $G$. Similarly, on physical grounds one expects that $\nabla_G$ simplifies in the limit when the Mandelstam invariants $s_{ij}$ become large at fixed masses $\M_i, \m_e$. Concretely, for $n_G=4$ one could study the asymptotes of the components of $\nabla_G$ when $|s|, |t| \gg 1$.

The Landau discriminant polynomial $\Delta_G$ is in many cases sparse with respect to its degree. For numerical interpolation, knowing the monomials (or other basis functions) of $\Delta_G$ a priori would lead to a significant improvement. Listing these monomials (or a superset) from the combinatorics of $G$ could be done by exploiting the connection with ${\cal A}$-discriminants and sparse resultants, see for instance \cite{sturmfels1994newton}.

In a number of results we made simplifying assumptions about genericity of the masses $\M_i, \m_e$. It would be interesting to extend them to massless cases, where one expects to see effects of infrared physics come into play. For instance, the degree of Landau discriminants can jump discontinuously and the facet description of the Landau polytopes $\LL_G$ could change drastically in such a degenerate limit. We leave such questions for future investigations.

\acknowledgments
We thank Nima Arkani-Hamed, Nick Early, Mathieu Giroux, Hofie Hannesdottir, Aaron Hillman, Chiara Meroni, Andrzej Pokraka, Marcus Spradlin, and Bernd Sturmfels for useful discussions, work on related projects, and comments on drafts of this paper. We are grateful to Sascha Timme for his help with \texttt{HomotopyContinuation.jl} and to Lukas K\"uhne for his help with the mathrepo website.
S.M. gratefully acknowledges the funding provided by Frank and Peggy Taplin, as well as the grant DE-SC0009988 from the U.S. Department of Energy.


\appendix
\section{\label{sec:appendix} Feynman Integrals in a Nutshell}
\input{appendix}

\pagebreak
\addcontentsline{toc}{section}{References}
\bibliographystyle{JHEP}
\bibliography{references}

\end{document}

%% file: macros.tex
\theoremstyle{plain}
\newtheorem{theorem}{Theorem}[section]

\newtheorem{remark}{Remark}[section]
\newtheorem{proposition}{Proposition}[section]
\newtheorem{corollary}{Corollary}[section]

\theoremstyle{definition}
\numberwithin{equation}{section}
\newtheorem{definition}{Definition}[section]

\newtheorem{examplex}{Example}
\newenvironment{example}
  {\pushQED{\qed}\examplex}
  {\popQED\endexamplex}
\newtheorem{assumption}{Assumption}
\newtheorem{conjecture}{Conjecture}
\newcommand{\N}{\mathbb{N}}
\newcommand{\Conv}{\textup{Conv}}

\newcommand{\m}{\mathfrak{m}}

\newcommand{\F}{\mathscr{F}}

\newcommand{\I}{\mathscr{I}}

\newcommand{\codim}{\textup{codim}}

\newcommand{\A}{\mathcal{A}}

\newcommand{\M}{\mathcal{M}}
\newcommand{\LLL}{\mathcal{E}}
\newcommand{\LL}{\mathbf{L}}
\newcommand{\PPP}{\mathbf{P}}

\newcommand{\K}{K}

\newcommand{\Z}{\mathbb{Z}}
\newcommand{\E}{\mathcal{E}}

\newcommand{\C}{\mathbb{C}}

\newcommand{\PP}{\mathbb{P}}

\newcommand{\V}{\mathcal{V}}

\newcommand{\D}{\delta}
\renewcommand{\d}{\partial}

\newcommand{\R}{\mathbb{R}}

\newcommand{\Vol}{\textup{Vol}}

\newcommand{\MV}{\textup{MV}}

\newcommand{\Newt}{\textup{Newt}}

%% file: introduction.tex
Feynman integrals are crucial for making theoretical predictions for high-precision particle physics experiments in the framework of perturbative quantum field theories. These integrals are extremely complicated functions of scattering energies and momenta of the particles involved.
Their explicit computation remains a challenging task spanning an enormous literature, see, e.g., \cite{Smirnov:2012gma} for a review. In fact, a lot of modern-day research is devoted to answering the simpler question: What are the singularities of a given Feynman integral and how complicated can they be?

Indeed, one of the biggest open questions in this topic has been the determination of the general analyticity properties of scattering amplitudes consistent with the underlying physical principles such a causality, locality, or unitarity. The importance of such investigations is emphasized by the recent applications in the bootstrap approaches \cite{Caron-Huot:2016owq,Caron-Huot:2019vjl,Guerrieri:2020bto,Guerrieri:2021tak}, bounds on low-energy effective field theories \cite{Adams:2006sv,Bellazzini:2020cot,Tolley:2020gtv,Caron-Huot:2020cmc}, all-multiplicity conjectures for singularities of the planar $\mathcal{N}=4$ super Yang--Mills amplitudes \cite{Prlina:2018ukf,Drummond:2019cxm,Arkani-Hamed:2019rds,Henke:2019hve,Mago:2020kmp}, or connections to cluster algebras \cite{Golden:2013xva,Chicherin:2020umh}, to name a few. A systematic study of such analyticity properties has been initiated in the 1960's in a program known as the \emph{S-matrix theory}; see, e.g., \cite{Eden:1966dnq,todorov2014analytic}. While a great deal of progress has been made in special cases---such as scattering of the lightest state in theories with a mass gap---the determination of the analytic structure of all but the very simplest Feynman diagrams remained too demanding computationally.

On the other hand, computational methods from \emph{nonlinear algebra} have seen significant advances in latest years, in particular with the development of robust numerical continuation methods \cite{timme2019mixed,telen2020robust} and fast, reliable implementations in packages such as \texttt{PHCpack} \cite{verschelde1999algorithm}, \texttt{Bertini} \cite{bates2013numerically}, and \texttt{HomotopyContinuation.jl} \cite{10.1007/978-3-319-96418-8_54}. These methods seem tailor-made to address the above questions. This work follows a natural direction in applying such recent nonlinear algebra techniques to the old problems in the S-matrix theory.

A connection between the S-matrix theory and nonlinear algebra was established in the work of Bjorken \cite{Bjorken:1959fd}, Landau \cite{Landau:1959fi}, and Nakanishi \cite{10.1143/PTP.22.128} who formulated a set of polynomial equations determining allowed positions of singularities of a given Feynman integral, nowadays known as the \emph{Landau equations}. In this way, the investigation of the analytic properties of Feynman integrals was transformed into an algebraic problem. Physically, Landau equations are the conditions for the worldline path integral of a given scattering process to localize on its classical saddle points \cite{Mizera:2021ujs,Mizera:2021fap}, in which the virtual particles become on-shell states.

Recent work on mathematical aspects of Landau equations includes \cite{Brown:2009ta,Bloch:2010gk,Abreu:2017ptx,Schultka:2019tfi,Collins:2020euz,Berghoff:2020bug,Muhlbauer:2020kut,HMSV}. For a more comprehensive summary of the literature see \cite[Sec.~II.C]{Mizera:2021fap}. We note that methods of computational algebraic geometry have been previously applied for integration-by-parts reduction of Feynman integrals, see \cite{Bendle:2020iim} for a review.

\paragraph*{Contributions.}\;
The goal of this paper is to investigate the above physical questions from an algebro-geometric and computational point of view. For each Feynman integral, we introduce the \emph{Landau discriminant} as a projective variety whose points are potential singularities of the integral. We prove its irreducibility and investigate its dimension and degree. We develop algorithmic tools for computing defining equations of the Landau discriminant, significantly advancing the state of the art. In order to showcase the effectiveness of these methods, we apply them to a gallery of examples, illustrated in Fig.~\ref{fig:diagrams}. We provide an implementation in the form of a \texttt{Julia} package \texttt{Landau.jl}, available at
\begin{center}
\url{https://mathrepo.mis.mpg.de/Landau/}.
\end{center}
The code makes use of \texttt{HomotopyContinuation.jl} \cite{10.1007/978-3-319-96418-8_54} (v2.6.0). Additionally, we study the combinatorics of polytopes arising from Feynman diagrams and Landau equations, motivated by the relation between Landau discriminants and the ${\cal A}$-discriminants from \cite{gelfand2008discriminants}. Finally, we present a numerical nonlinear algebra routine to compute the number of master integrals for any family of Feynman diagrams, exploiting the connection to maximum likelihood estimation in algebraic statistics. With this work, we aspire to pave the way for future research by both physicists and mathematicians in the study of the analytic structure of Feynman integrals, Landau discriminants, and related topics.

\paragraph*{Outline.}\;
This paper is organized as follows. We start by recalling elementary definitions and introducing the notation in Sec.~\ref{subsec:feynman}-\ref{sec:landaueq}. In addition, in App.~\ref{sec:appendix}, we include a basic introduction to Feynman integrals for readers with a mathematical background. It explains the Feynman rules and the transition from the loop-momentum integral to the worldline formalism. We show how this conversion gives rise to Schwinger parameters and Symanzik polynomials, which play a key role in this paper.

In Sec.~\ref{subsec:discriminant} we introduce the Landau discriminant $\nabla_G$, which describes the singularity locus of a given Feynman diagram $G$, see Def.~\ref{def:discriminant}.
In Thm.~\ref{thm1} we prove that $\nabla_G$ is an irreducible variety of codimension at least $1$ in the projectivized kinematic space of energies, momenta, and masses of $G$, denoted $\PP(\K_G)$. In Sec.~\ref{sec:ngon}-\ref{sec:banana} we compute Landau discriminants for the well-known examples of the $n$-gon and banana diagrams, see Fig.~\ref{fig:An} and \ref{fig:BE}.

\begin{figure}[!t]
     \centering
     \captionsetup{justification=centering}
     \begin{subfigure}[c]{0.3\textwidth}
         \centering
         \includegraphics[scale=0.8]{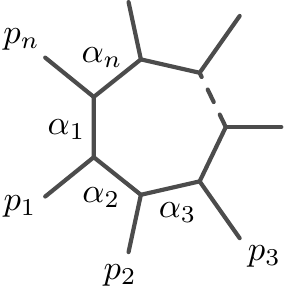}
         \caption{\label{fig:An}One-loop $n$-gon diagram, $G=\mathtt{A}_n$ (Sec.~\ref{sec:ngon})}
     \end{subfigure}
     \begin{subfigure}[c]{0.3\textwidth}
         \centering
         \includegraphics[scale=0.8]{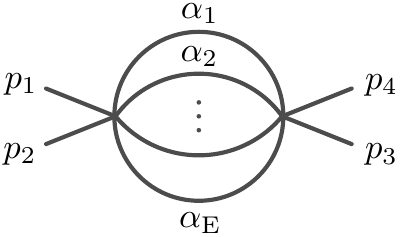}
         \caption{\label{fig:BE}Banana diagram with $\E$ edges, $G=\mathtt{B}_\E$ (Sec.~\ref{sec:banana}, \ref{sec:banana-polytopes})}
     \end{subfigure}
     \begin{subfigure}[c]{0.3\textwidth}
         \centering
         \includegraphics[scale=0.8]{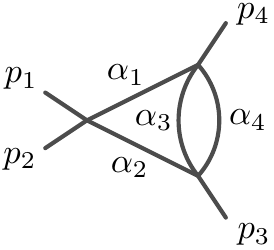}
         \caption{\label{fig:par}Parachute diagram,\\ $G=\mathtt{par}$ (Ex.~\ref{ex:UFLpar}, Thm.~\ref{thm2})}
    \end{subfigure}
     \\ \vspace{1.5em}
     \begin{subfigure}[c]{0.3\textwidth}
         \centering
         \includegraphics[scale=0.8]{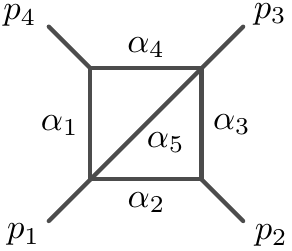}
         \caption{\label{fig:acn}Acnode diagram,\\ $G=\mathtt{acn}$ (Ex.~\ref{ex:acnode}, Rk.~\ref{ex:acnode2}, Thm.~\ref{thm2})}
     \end{subfigure}
    \begin{subfigure}[c]{0.3\textwidth}
         \centering
         \includegraphics[scale=0.8]{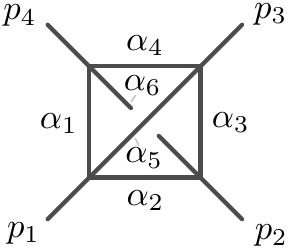}
         \caption{\label{fig:env}Envelope diagram,\\ $G=\mathtt{env}$ (Ex.~\ref{ex:env}, Sec.~\ref{sec:CN-analysis})}
     \end{subfigure}
    \begin{subfigure}[c]{0.3\textwidth}
         \centering
         \includegraphics[scale=0.8]{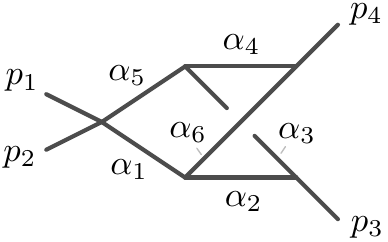}
         \caption{\label{fig:npltrb}Non-planar triangle-box diagram, $G=\mathtt{npltrb}$ (Thm.~\ref{thm2})}
     \end{subfigure}
     \\ \vspace{1.5em}
     \begin{subfigure}[c]{0.3\textwidth}
        \centering
        \includegraphics[scale=0.8]{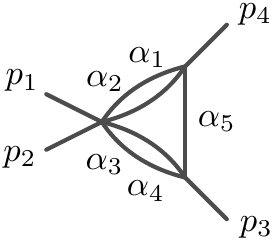}
         \caption{\label{fig:tdetri}Twice doubled-edge triangle diagram, $G=\mathtt{tdetri}$ (Thm.~\ref{thm2})}
     \end{subfigure}
     \begin{subfigure}[c]{0.3\textwidth}
        \centering
        \includegraphics[scale=0.8]{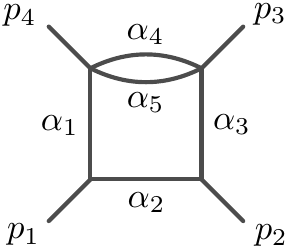}
        \caption{\label{fig:debox}Doubled-edge box diagram, $G=\mathtt{debox}$ (Thm.~\ref{thm2})}
     \end{subfigure}
     \begin{subfigure}[c]{0.3\textwidth}
        \centering
        \includegraphics[scale=0.8]{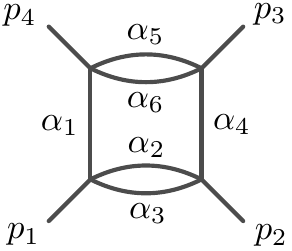}
        \caption{\label{fig:tdebox}Twice doubled-edge box diagram, $G=\mathtt{tdebox}$ (Thm.~\ref{thm2})}
     \end{subfigure}
     \\ \vspace{1.5em}
     \begin{subfigure}[c]{0.3\textwidth}
         \centering
         \includegraphics[scale=0.8]{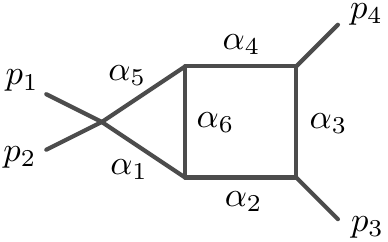}
         \caption{\label{fig:pltrb}Planar triangle-box diagram, $G=\mathtt{pltrb}$ (Sec.~\ref{subsec:generalparameters}, Thm.~\ref{thm2})}
     \end{subfigure}
     \begin{subfigure}[c]{0.3\textwidth}
         \centering
         \includegraphics[scale=0.8]{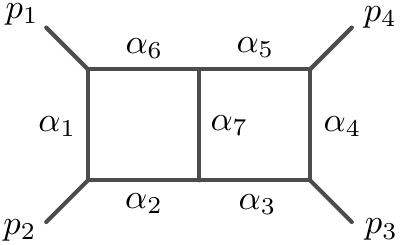}
         \caption{\label{fig:dbox}Double-box diagram,\\ $G=\mathtt{dbox}$ (Thm.~\ref{thm2})}
     \end{subfigure}
     \begin{subfigure}[c]{0.3\textwidth}
         \centering
         \includegraphics[scale=0.8]{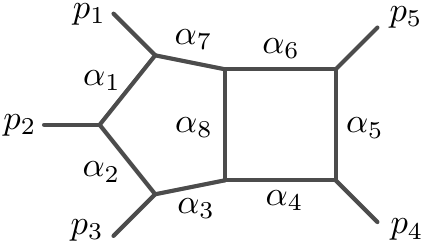}
         \caption{\label{fig:pbox}Penta-box diagram,\\ $G=\mathtt{pentb}$ (Ex.~\ref{ex:pentb})}
     \end{subfigure}
     \caption{\label{fig:diagrams}Summary of the Feynman diagrams considered in this paper.}
\end{figure}

In Sec.~\ref{sec:computing}, we introduce nonlinear algebra methods for computing Landau discriminants, constituting our main computational results. Sec.~\ref{subsec:elimination} gives a brief discussion on symbolic elimination methods using \texttt{Macaulay2} \cite{M2}, with the acnode diagram from Fig.~\ref{fig:acn} as a running example. In Sec.~\ref{subsec:sampling} we discuss a numerical approach based on homotopy continuation. It combines (pseudo-)witness sets of linear projections \cite{hauenstein2010witness}, irreducible decomposition via monodromy \cite{duff2019solving} and numerical interpolation. It is illustrated in Fig.~\ref{fig:intro}. 
To each point in the projectivized kinematic space $\PP(\K_G)$ we attach the space $X$ of admissible Schwinger parameters. The incidence variety $Y$ in the product $X \times \PP(\K_G)$ is defined by the Landau equations. The closure of its projection onto $\PP(\K_G)$ is the Landau discriminant $\nabla_G$ (both $Y$ and $\nabla_G$ are shown in blue in Fig.~\ref{fig:intro}). The numerical sampling algorithm proceeds by repeatedly intersecting $Y$ with the pullback of a generic hyperplane in $\PP(\K_G)$ (orange). The intersection points are computed using homotopy techniques, and the projection gives $\deg \nabla_G$ points on $\nabla_G$ (red). 
Defining equations for the Landau discriminant are obtained by interpolating between the red sampling points in $\PP(\K_G)$.

In the remainder of Sec.~\ref{sec:computing} we apply this technique to the computation of Landau discriminants for a range of diagrams illustrated in Fig.~\ref{fig:diagrams}. With the exception of $\mathtt{A}_n$ and $\mathtt{B}_\E$, all the computations are cutting-edge and go far beyond the solutions of Landau equations computed using previous methods \cite{Eden:1966dnq,todorov2014analytic}. These results are summarized in Thm.~\ref{thm2} and Ex.~\ref{ex:env}-\ref{ex:pentb}. The most complicated diagram we consider is the envelope diagram $\mathtt{env}$ from Fig.~\ref{fig:env}, whose Landau discriminant $\nabla_{\mathtt{env}}$ is a reducible surface of degree $45$ in the projectivized kinematic space $\PP^3$. Another large example is the penta-box diagram from Fig.~\ref{fig:pbox}, whose Landau discriminant is a degree $12$ $5$-fold in $\PP^6$. A summary of degrees and dimensions of Landau discriminants for all the diagrams from Fig.~\ref{fig:diagrams} is provided in Tab.~\ref{tab:dimdeg}. This part of the work is concluded with Sec.~\ref{sec:CN-analysis}, where we explain how to visualize the discriminant on kinematic subspaces and quantitatively determine which singularities are physically relevant.

\begin{figure}[!t]
	\centering
	\includegraphics[scale=1]{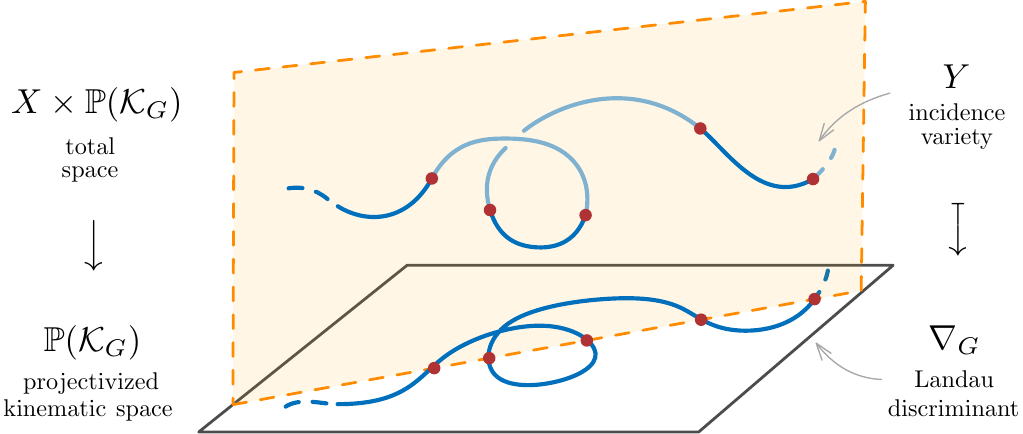}
	\caption{\label{fig:intro}Diagrammatic illustration of the numerical sampling algorithm.}
\end{figure}

In order to study degenerate solutions of Landau equations and bound the the degree of the Landau discriminant, we introduce Landau polytopes $\mathbf{L}_G$ in Sec.~\ref{sec:polytopes}. In particular, in Prop.~\ref{prop:L-boundary} we explain how their facet structure can be determined in terms of the combinatorics of the Feynman diagram $G$. As an example, in Sec.~\ref{sec:banana-polytopes} we explicitly work out the facets of $\mathbf{L}_{\mathtt{B}_\E}$ for the banana diagrams $\mathtt{B}_\E$. We conjecture that all facets of the Landau polytope are labelled by certain subdiagrams of $G$, see Conj.~\ref{conj:facets}. This conjecture is verified for all examples in Fig.~\ref{fig:diagrams}.

In Sec.~\ref{sec:counting}, we apply homotopy continuation techniques to a different problem related to Feynman integrals, namely that of counting the number of master integrals. After explaining the formulation of Feynman integrals in analytic regularization in terms of twisted cohomologies \cite{Mastrolia:2018uzb}, we formalize previous results from the literature as Thm.~\ref{thm:thm2} stating the connection between the signed Euler characteristic and the number of critical points of the potential $W_G$ associated to a given $G$. We provide a simple \texttt{Julia} routine, which proves a lower bound on the signed Euler characteristic via certified numerical computations.

Sec.~\ref{sec:conclusion} contains concluding remarks and a list of research directions.

%% file: section2.tex
In this section, after recalling the definition of Feynman integrals and Symanzik polynomials in Sec.~\ref{subsec:feynman}-\ref{sec:Symanzik}, we present the formulation of the \emph{Landau equations} related to a Feynman diagram in Sec.~\ref{sec:landaueq}. This sets the stage for the definition of the \emph{Landau discriminant} in Sec.~\ref{subsec:discriminant} and a discussion of its first properties. In Sec.~\ref{sec:ngon}-\ref{sec:banana} we present several examples. 

\subsection{Feynman Integrals} \label{subsec:feynman}
The motivation for this work comes from studying analyticity properties of \emph{Feynman integrals}. We briefly recall their definition here. For the reader who is not familiar with scattering amplitudes, we provide a more basic introduction in App.~\ref{sec:appendix}.
For our purposes, a Feynman diagram $G$ is a connected undirected graph with $\E_G$ internal edges, $n_G$ external legs (open edges), as well as $\L_G$ independent loops. Examples are shown in Fig.~\ref{fig:diagrams}.
Such diagrams encode interaction patterns in a scattering process of the external particles mediated by the internal ones.
Each external leg in $G$ corresponds to one of these particles and carries a \emph{momentum vector} $p_i \in \R^{1,\D-1}$ for $i=1,2,\ldots,n_G$.
Here $\R^{1,\D-1}$ represents \emph{$\D$-dimensional Minkowski momentum space} in the mostly-minus signature. Concretely, this means that the $p_i$ are vectors with $\D$ real entries, and the Minkowski pairing of $p = (p^{(0)}, p^{(1)}, \ldots, p^{(\D-1)})$, $q = (q^{(0)}, q^{(1)}, \ldots, q^{(\D-1)}) \in \R^{1,\D-1}$ is given by $p \cdot q = p^{(0)}q^{(0)} - p^{(1)} q^{(1)} - \cdots - p^{(\D-1)} q^{(\D-1)}$. We will use the standard notation $p^2 := p \, \cdot \, p$. Most physical interest lies in $\D=4$.
To each internal edge $e$ we associate a \emph{Schwinger parameter} $\alpha_e \in \C^\ast$, which needs to be integrated out, and a mass $m_e \in \R_+ \cup \{0\}$, indexed by $e=1,2,\ldots,\E_G$. 

In the \emph{worldline formalism} the \emph{Feynman integral} associated to a Feynman diagram $G$ is an integral over positive values of all the $\E_G$ Schwinger parameters, $\alpha_e \in \R_+$. Roughly speaking, it takes the following form:
\be\label{eq:Feynman0}
\int_{\R_+^{\E_G}}  \frac{\d^{\E_G}\alpha}{\U^{\D/2}_G} \N_G \exp \left[ \frac{i}{\hbar} \V_G \right].
\ee
The measure is simply $\d^{\E_G}\alpha := \prod_{e=1}^{\E_G} \d \alpha_e$. The exponent involves the function $\V_G$, which has the interpretation of the worldline action for $G$ after analytic continuation. It can be written as the ratio
\be\label{eq:V}
\V_G := \F_G / \U_G
\ee
of two polynomials, $\U_G$ and $\F_G$, called \emph{Symanzik polynomials}. Both $\U_G$ and $\F_G$ are homogeneous in the Schwinger parameters, and $\F_G$ involves the momentum vectors $p_i$ and the internal masses $m_e$. We will define these polynomials in terms of the combinatorics of the diagram $G$ in Sec.~\ref{sec:Symanzik}. They will play a central role in this work. The factor $\N_G$ is a polynomial in the Schwinger parameters (and typically other external data such as polarization vectors or color factors) that encodes the physics of the particle interactions. It will not be important for our purposes. Eq.~\eqref{eq:Feynman0} also features $i=\sqrt{-1}$ and the reduced Planck constant $\hbar>0$. 

For completeness, let us mention that the integrals of the type \eqref{eq:Feynman0} do not converge in general. Their more precise definition needs small deformations of the contour, which can be implemented by inserting an infinitesimal parameter $\epsilon>0$ (the Feynman $i\epsilon$ factor) in the exponent, ensuring exponential suppression of the integrand as each $\alpha_e \to \infty$:
\be\label{eq:Feynman}
\I_G^{\mathrm{reg}} := \int_{\R_+^{\E_G}}  \frac{\d^{\E_G}\alpha}{\U^{\D/2}_G} \N_G \RR^{\mathrm{reg}}_G \exp \left[ \frac{i}{\hbar} \left( \V_G + i\epsilon \textstyle\sum_{e=1}^{\E_G} \alpha_e \right) \right].
\ee
In Sec.~\ref{sec:counting} we will give an equivalent procedure with contour deformations.
Additionally, one needs to employ a regularization procedure, summarized in the factor $\RR_G^{\mathrm{reg}}$. Two popular choices are \emph{dimensional regularization} (dim), corresponding to a shift $\D \to \D - 2\eps$ for a parameter $\eps \in \C \setminus \Z$ (not related to $\epsilon$), and \emph{analytic regularization} (an), which introduces small parameters $\delta_e \in \C\setminus \Z$ for each edge $e$,
\be
\RR^{\mathrm{reg}}_G := \begin{dcases} \U_G^{\eps} \qquad&\text{reg}=\text{dim},\\
\textstyle\prod_{e=1}^{\E_G} \alpha_e^{\delta_e} \qquad&\text{reg}=\text{an}.
\end{dcases}
\ee
As a result, \eqref{eq:Feynman} becomes a meromorphic function of the regulators, which after summing over all Feynman diagrams contributing to a given scattering process are taken to zero.
A more precise definition of Feynman integrals in analytic regularization will be given later in Sec.~\ref{sec:counting}, and it will not play any role in the intervening sections.

In this paper, we will not be concerned with the important problem of \emph{evaluating} Feynman integrals. Rather, we are interested in investigating singularities of ${\cal I}_G^{{\rm reg}}$ as a function of the kinematic data. As we will see, this is a problem from \emph{nonlinear algebra} \cite{michalek2021invitation}.

\subsection{\label{sec:Symanzik}Symanzik Polynomials and the Kinematic Space}

A central role in the study of singularities of Feynman integrals \eqref{eq:Feynman} is played by the Symanzik polynomials $\U_G$ and $\F_G$. In this section, we present their definitions in terms of concepts from graph theory. This will establish Lorentz invariance of the Feynman integrals \eqref{eq:Feynman}, leading us to the definition of kinematic space. 

Let $G$ be a connected Feynman diagram. A \emph{spanning tree} in $G$ is a connected subset of $\E_G - \L_G$ internal edges that contains all vertices of $G$.
We write ${\cal T}_G$ for the set of all spanning trees of $G$. 

\begin{definition}[First Symanzik polynomial]
The \emph{first Symanzik polynomial} $\U_G$ is
\be
\U_G := \sum_{T \in {\cal T}_G} \prod_{e \notin T} \alpha_e, 
\ee
where the product runs over the $\L_G$ internal edges that were removed from $G$ to obtain the spanning tree $T$.
\end{definition}

A \emph{spanning $2$-tree} $T_1 \sqcup T_2$ in $G$ is a disjoint union (with respect to the sets of edges and vertices) of two trees $T_1$ and $T_2$ in $G$, containing all of its vertices. For a subset of external legs $S$, let ${\cal T}_{G,S}$ denote the set of all spanning $2$-trees $T_S \sqcup T_{\bar{S}}$ in $G$, such that $T_S$ contains the vertices attached to the external legs labeled by $S$, and no vertices attached to external legs labeled by the complementary set $\bar{S} = \{1,2,\ldots, n_G\} \setminus S$, and similarly for $T_{\bar{S}}$. 

\begin{definition}[Second Symanzik polynomial]
For $S \subset \{1, \ldots, n_G\}$, define
\be\label{eq:FGS}
\F_{G,S} := \sum_{T_S \sqcup T_{\bar{S}} \in {\cal T}_{G,S}} \prod_{e \notin T_S, T_{\bar{S}}} \alpha_e,
\ee
where the product is over all the $\L_G + 1$ edges that needed to be removed in order obtain the $2$-tree $T_S \sqcup T_{\bar{S}}$.
The second Symanzik polynomial $\F_G$ is 
\be\label{eq:FG}
\F_G := \sum_{\{S,\bar{S}\} \in {\cal P}_G} \!\!\!(\textstyle\sum_{i \in S} p_i)^2\, \F_{G,S} -  \left( \sum_{e=1}^{\E_G} m_e^2 \alpha_e\right) \U_G.
\ee
Here ${\cal P}_G$ denotes the set of all partitions of the $n_G$ external legs into two disjoint non-empty sets $S$ and $\bar{S}$. The number of terms in the first sum is $|{\cal P}_G| = 2^{n_G - 1} - 1$. Each $\F_{G,S}$ is weighted with the Minkowski norm $(\textstyle\sum_{i \in S} p_i)^2$ of the total external momentum flowing into $S$.
\end{definition}

The first (second) Symanzik polynomial is a homogeneous polynomial of degree $\L_G$ ($\L_G+1$) in the Schwinger parameters $\alpha_e$, at most linear (quadratic) in each individual $\alpha_e$. Only $\F_G$ depends on the external parameters $p_i, m_e$.

Using the momentum conservation constraint $\sum_{i = 1}^{n_G} p_i = 0$, one can show that the number of independent Lorentz invariants $(\sum_{i\in I}p_i)^2$ is equal to $n_G (n_G - 1)/2$ whenever $n_G \leq \D+1$. We can distinguish between the case $|I| = 1$, in which $p_i^2 =: M_i^2$ is the squared mass of the $i$-th external particle, and the cases with $2 \leq |I| \leq n_G {-}2$, in which case $(\sum_{i\in I}p_i)^2$ are called \emph{Mandelstam invariants}. There is a canonical choice of basis for such invariants, which we summarize in the following.

\begin{definition}[Kinematic space]
For any subset of external variables $I \subset \{1,2,\ldots,n_G\}$ such that $2 \leq |I| \leq n_G {-}2$, the \emph{Mandelstam invariant} $s_I$ is
\be
s_{I} := (\textstyle\sum_{i\in I}p_i)^2.
\ee
A canonical basis ${\cal B}_{n_G}$ of Mandelstam invariants is given by all such $I$ that consist of consecutive labels in the cyclic ordering $(1,2,\ldots,n_G)$. For the squared masses of external and internal particles we introduce respectively
\be
\M_i := M_i^2, \qquad \m_e := m_e^2.
\ee
The \emph{kinematic space} $\K_G$ for a given diagram $G$ is defined to be
\begin{align}
{\cal K}_{G} := \{ s_I \in \C \;|\; I \in {\cal B}_{n_G} \} &\times \{ \M_i \in \C \;|\; i=1,2,\ldots,n_G\}\times \{\m_e \in \C \;|\; e=1,2,\ldots,E_G\}.\nn
\end{align}
\end{definition}
Under the assumption $n_G \leq \D + 1$ (which we will make throughout, see Rk.~\ref{rem:nGD+1}), the kinematic space is an affine space of dimension $\dim_{\C} {\cal K}_G = \frac{n_G(n_G{-}1)}{2} + \E_G$.
Feynman integrals are multi-valued functions on $\K_G$. The goal of this paper is to determine varieties in $\K_G$ along which the Feynman integral \eqref{eq:Feynman} can develop singularities. We note that the notation $\M_{i}$, $\m_e$ for the squared masses is non-standard in the physics literature, but we employ it here to make $\F_G$ homogeneous in the kinematic invariants and to simplify the expressions given later in the text. The physically most interesting setup is to fix all the $\M_i$ and $\m_e$ to specific constant values and consider the behavior of \eqref{eq:Feynman} as the Mandelstam invariants $s_I$ are varied. Alternatively, one may interpret $\M_i$ as the norm of the total (off-shell) momentum attached to a given vertex, in which case it is interesting to also vary it.
\begin{remark} \label{rem:nGD+1}
When $n_G > \D+1$, there are additional constraints on the Mandelstam invariants coming from the fact that the $n_G$ momentum vectors $p_i$ are embedded in $\R^{1,\D-1}$. Concretely, these constraints are given by the vanishing of all the $(\D{+}1) \times (\D{+}1)$ minors of the Gram matrix with $(i,j)$-entry $p_i {\cdot} p_j$. Since all the examples in this work have $n_G \leq 5$ external legs, these constraints do not play any role in $\D=4$ dimensions.
\end{remark}

\begin{example}[Four-point scattering]
For diagrams $G$ with with $n_G=4$ (four-point scattering), the kinematic space is given by
\be
\K_{G} = \{ (s,t,\M_1,\M_2,\M_3,\M_4,\m_1,\m_2,\ldots,\m_{\E_G}) \in \C^{\E_G + 6} \},
\ee
where the two independent Mandelstam invariants are given by
\be
s := s_{12} = (p_1+p_2)^2, \qquad t := s_{23} = (p_2 + p_3)^2
\ee
in the conventional notation. For example, the Mandelstam invariant $u := s_{13}$ is not independent because using momentum conservation $\sum_{i=1}^{4} p_i = 0$ we can express it as
\be
u = (p_1 + p_3)^2 = \M_1 + \M_2 + \M_3 + \M_4 - s - t.
\ee
The kinematic space is therefore $(\E_G {+} 6)$-dimensional. In the above notation, the second Symanzik polynomial reads
\begin{align}
\F_{G} = &\; s (\F_{G,12} - \F_{G,13}) + t (\F_{G,23} - \F_{G,13})+{\textstyle\sum}_{i=1}^{4} \M_i (\F_{G,i} + \F_{G,13}) - \left( \textstyle\sum_{e=1}^{\E_G} \m_e \alpha_e \right) \U_G.\nn
\end{align}
\end{example}

We follow with two simple examples of diagrams with $n_G = 4$.

\begin{example}[Box diagram] \label{ex:box}
As the first example we consider the box diagram, $G=\texttt{A}_4$, as illustrated in Fig.~\ref{fig:An}. We have
\be
n_{\texttt{A}_4} = 4, \qquad \E_{\texttt{A}_4} = 4, \qquad \L_{\texttt{A}_4} = 1.
\ee
The first Symanzik polynomial is a sum over four spanning trees obtained by removing a single edge,
\be
\U_{\texttt{A}_4} = \alpha_1 + \alpha_2 + \alpha_3 + \alpha_4.
\ee
Similarly, there are six spanning $2$-trees obtained by removing two edges in all possible combinations, which gives the second Symanzik polynomial
\begin{align}
\F_{\texttt{A}_4} = s \alpha_1 \alpha_3 + t \alpha_2 \alpha_4 + \M_1 \alpha_1 \alpha_2 &+ \M_2 \alpha_2 \alpha_3 + \M_3 \alpha_3 \alpha_4 + \M_4 \alpha_4 \alpha_1 \\
& -(\m_1 \alpha_1 + \m_2 \alpha_2 + \m_3 \alpha_3 + \m_4 \alpha_4) \U_{\texttt{A}_4}.\nn
\end{align}
\end{example}

\begin{example}[Sunrise diagram]
Let us consider the sunrise diagram, $G=\texttt{B}_3$, illustrated in Fig.~\ref{fig:BE}. In this case we have
\be
n_{\texttt{B}_3} = 4,\qquad \E_{\texttt{B}_3} = 3,\qquad \L_{\texttt{B}_3} = 2.
\ee
The three spanning trees are obtained by removing two edges in all possible combinations, giving
\be
\U_{\texttt{B}_3} = \alpha_1 \alpha_2 + \alpha_2 \alpha_3 + \alpha_3 \alpha_1.
\ee
There is a unique spanning $2$-tree corresponding to the removal of all three edges,
\be
\F_{\texttt{B}_3} = s \alpha_1 \alpha_2 \alpha_3 - (\m_1 \alpha_1 + \m_2 \alpha_2 + \m_3 \alpha_3) \U_{\texttt{B}_3}.
\ee
In particular, the answer is $t$-independent since there is no way of separating the external legs $\{2,3\}$ from $\{1,4\}$.
\end{example}

\subsection{Saddle Points and Landau Equations} \label{sec:landaueq}

Regularization of Feynman integrals overcomes global divergence issues (such as ultraviolet or infrared divergences). However, the question whether the integral exists for a specific point in the kinematic space $\K_G$ still remains. We would like to determine the singular locus of \eqref{eq:Feynman} in $\K_G$. Physically, such singularities correspond to the classical limit, $\hbar\to0$, and are known as \emph{anomalous thresholds}; see, e.g., \cite{Eden:1966dnq}. To be more precise, by singularities we mean branch points or poles of the Feynman integral in the kinematic space. They are determined by the critical points of $\V_G$ and have the interpretation of intermediate particles becoming long-lived on-shell states. To set the background for the remainder of this paper, in this subsection we make the connection between singularities of \eqref{eq:Feynman} and saddle points of $\V_G$ more precise.

Recall that $\V_G$, as defined in \eqref{eq:V}, is a homogeneous rational function in the Schwinger parameters of degree $1$. In what follows, we assume that the parameters $\M_i, \m_e, s_I$ are fixed and we are interested in finding $\alpha^\ast$ such that 
\be\label{eq:saddle}
\frac{\partial \V_G}{\partial \alpha_e}(\alpha^\ast) = 0, \quad e = 1,2, \ldots, \E_G.
\ee
This is a system of $\E_G$ rational function equations in $\E_G$ homogeneous variables.
Homogeneity implies that the solutions $\alpha^\ast$
live naturally in a projective space: if $\alpha^\ast=(\alpha_1^\ast,\cdots,\alpha_{\E_G}^\ast)$ is a solution, then so is $\lambda \alpha^\ast=(\lambda\alpha_1^\ast,\cdots,\lambda\alpha_{\E_G}^\ast)$ for any $\lambda \in \C^\ast$. Let us first focus on critical points with nonzero coordinates, i.e.,
\be\label{eq:alpha-toric}
\alpha^* = (\alpha_1^\ast: \alpha_2^\ast : \cdots : \alpha_{\E_G}^\ast) \, \in \, \PP^{\E_G-1} \setminus V_{\PP^{\E_G-1}} (\alpha_1\cdots \alpha_e \U_G),
\ee
where $V_A(f)$ denotes the subvariety given by $\{f=0\}$ in $A$.
Here we also excluded $\U_G = 0$ and we will come back to this point shortly.
Note that for all these solutions we have $\V_G(\alpha^\ast) = 0$, since Euler's rule gives
\be
\V_G({\alpha}^\ast) = \sum_{e=1}^{\E_G} \alpha_e^\ast \frac{\partial \V_G}{\partial \alpha_e}({\alpha}^\ast) = 0.
\ee
Since \eqref{eq:saddle} is a system of $\E_G$ constraints on $(\E_G{-}1)$-dimensional space of Schwinger parameters $\alpha_e$, one expects that there are no solutions for generic kinematic parameters in $\K_G$. This is indeed the case, as we will show in Thm.~\ref{thm1}.

In order to see why solutions of \eqref{eq:saddle} lead to divergences of the Feynman integral, let us explicitly perform the integration over the projective scale $\lambda$. Without loss of generality we can assume that $\N_G$ is homogeneous with degree $m_G$. After the change of variables $\alpha_e \to \lambda \alpha_e$ and setting $\alpha_{\E_G} = 1$ we have
\begin{gather}
\U_G \to \lambda^{\L_G} \U_G, \qquad \V_G \to \lambda \V_G, \qquad \N_G \to \lambda^{m_G} \N_G,\\
\d^{\E_G}\alpha \to \lambda^{\E_G - 1} \d\lambda\, \d^{\E_G-1}\alpha_e 
\end{gather}
as well as
\be
\mathrm{R}_G^{\text{reg}} \to \lambda^{\delta} \mathrm{R}_G^{\text{reg}},\qquad \delta := \begin{dcases} \eps \L_G \qquad&\text{reg}=\text{dim},\\
\textstyle\sum_{e=1}^{\E_G} \delta_e \qquad&\text{reg}=\text{an}.
\end{dcases}
\ee
Introducing the \emph{degree of divergence} $d_G := m_G + \E_G - \L_G \D/2 + \delta$, the Feynman integral \eqref{eq:Feynman} becomes
\begin{align}
\I_{G}^{\text{reg}} &= \int_{\R_+^{\E_G-1}}  \frac{\d^{\E_G-1}\alpha}{\U^{\D/2}_G} \N_G \RR^{\mathrm{reg}}_G \int_{\R_+} \frac{\d\lambda}{\lambda^{1-d_G}} \exp \left[ \frac{i \lambda}{\hbar} \left( \V_G + i\epsilon \textstyle\sum_{e=1}^{\E_G} \alpha_e \right) \right]\nn\\
&= (i\hbar)^{-d_G} \Gamma(d_G) \int_{\R_+^{\E_G-1}}  \frac{\d^{\E_G-1}\alpha}{\U^{\D/2}_G} \frac{\N_G \RR^{\mathrm{reg}}_G}{(\V_G + i\epsilon \textstyle\sum_{e=1}^{\E_G} \alpha_e)^{d_G}}.\label{eq:IG}
\end{align}
Therefore, on a saddle point the $\epsilon \to 0^+$ limit gives a singularity if $d_G > 0$. Note that $\V_G = 0$ by itself does not imply a singularity because one can deform the integration contour to avoid it. An explicit deformation of this type is given later in \eqref{eq:Gamma-G}. See Ex.~\ref{ex:pinch} for an illustration in the case of the bubble diagram (Fig.~\ref{fig:bubble}).

The saddle point conditions \eqref{eq:saddle} are known as the \emph{leading Landau equations} \cite{Landau:1959fi}, which were first written in this form by Nakanishi in \cite{10.1143/PTP.22.128}. Traditionally, they have been associated with pinch singularities of the integration contour; see, e.g., \cite{Eden:1966dnq}. An explicit example is given in Ex.~\ref{ex:pinch}. The physical interpretation in terms of saddle points in the worldline formalism was given in \cite{Mizera:2021ujs,Mizera:2021fap}.

Singularities of regularized Feynman integrals can also come from non-toric saddle points, i.e., those for which there is a set $E$ of one or more edges such that $\alpha_e = 0$ for all $e\in E$. The equations \eqref{eq:saddle} where $\partial \F_G /\partial \alpha_e = 0$ is substituted for $\alpha_e = 0$ for all $e \in E$ are called \emph{subleading Landau equations}. They correspond to leading Landau equations for a diagram $G/E$ obtained from $G$ by contracting all the edges in $E$ \cite{Landau:1959fi}. It is also possible that $\alpha_e = 0$ and $\partial \F_G / \partial \alpha_e = 0$ simultaneously for a subset of edges, but we do not study such solutions in this work. We parenthetically remark that also solutions to $\U_G = 0$ are known to correspond to \emph{second-type} Landau singularities \cite{doi:10.1063/1.1724262}, which physically correspond to collinear divergences in the external kinematics.

\begin{example}\label{ex:subleading}
Consider the diagram $G = \mathtt{env}$ from Fig.~\ref{fig:env}. Its subleading Landau singularities can be determined as follows. Shrinking any of its $6$ internal edges leads to a twice doubled-edge triangle diagram of the same topology as $\mathtt{tdetri}$ from Fig.~\ref{fig:tdetri}. For each of them, we can further shrink the edge $5$ (in the notation of Fig.~\ref{fig:tdetri}), resulting in a banana integral with four edges, $\mathtt{B}_4$ from Fig.~\ref{fig:BE}. Shrinking of any other set of edges leads to diagrams that do not depend on any Mandelstam invariants, which we do not consider here.
\end{example}

Without loss of generality, we can focus only on the leading Landau equations and their toric solutions \eqref{eq:alpha-toric}. (Subleading Landau equations of $G$ are the same as the leading ones for all possible $G/E$.) Since this excludes solutions with a vanishing first Symanzik polynomial $\U_G$, the saddle point conditions are equivalent to
\be\label{eq:LE}
\frac{\partial \F_G(\alpha_e^\ast)}{\partial \alpha_e} = 0\quad\text{for}\quad e=1,2,\ldots,\E_G, \qquad \U_G(\alpha_e^\ast) \neq 0. 
\ee
We will refer to these equations (together with the inequation) as the \emph{Landau equations}, omitting the word ``leading'' for conciseness. We are interested in characterizing the kinematic parameters for which the Landau equations have solutions. Below, we will formalize this question in terms of the \emph{Landau discriminant}.

\subsection{\label{subsec:discriminant}Landau Discriminants}

Once we fix some kinematic data $\M_i, \m_e, s_I$, i.e., we fix a point $q \in \K_G$, the leading Landau singularities are points in the variety
\be \label{eq:defX}
X := \PP^{\E_G-1} \setminus V_{\PP^{\E_G-1}}(\alpha_1 \alpha_2 \cdots \alpha_{\E_G} \U_G).
\ee
They correspond to the points in $X$ where the projective hypersurface $\{\F_G = 0\} \subset \PP^{\E_G-1}$ is singular. This imposes $\E_G$ conditions on a $(\E_G{-}1)$-dimensional space. One could expect that for general parameters $q$, the solution set is empty. In this section, we show that this is indeed what happens. 

By homogeneity of $\F_G$ in the parameters, it is natural and convenient to work in the projectivized kinematic space
$\PP(\K_G) := \K_G/\sim$, 
where $q \sim \lambda q$ for all $\lambda \in \C^*$. 
First, let us define the \emph{incidence variety}
\begin{equation} \label{eq:incidence}
    Y := \left \{ (\alpha, q) \in X \times \PP(\K_G)  ~\bigg|~ \frac{\partial \F_G}{\partial \alpha_e}(\alpha; q) = 0,~ e = 1,2, \ldots, \E_{G} \right \}.
\end{equation}
Here the notation $\F_G(\alpha;q)$ makes the dependence of $\F_G$ on the kinematic parameters explicit. The variety $Y$ has two natural projection maps associated to it:
\be
\pi_X:\; Y \rightarrow X, \qquad \pi_{\PP(\K_G)}:\; Y \rightarrow \PP(\K_G).
\ee
These are given by
\be
\pi_X(\alpha,q) = \alpha, \qquad \pi_{\PP(\K_G)}(\alpha,q) = q.
\ee
The solutions to the Landau equations with kinematic data $q$ are the points in $(\pi_X \circ \pi_{\PP(\K_G)}^{-1})(q)$. We are interested in finding parameters $q \in \PP(\K_G)$ for which the Landau equations have solutions. That is, for which $(\pi_X \circ \pi_{\PP(\K_G)}^{-1})(q) \neq \emptyset$. 

\begin{definition}[Landau discriminant]\label{def:discriminant}
The \emph{Landau discriminant} $\nabla_G$ of a Feynman diagram $G$ is the subvariety of $\PP(\K_G)$ given by the Zariski closure
\be\label{eq:nablaG}
\nabla_G := \overline{\pi_{\PP(\K_G)}(Y)} \subset \PP(\K_G)
\ee
of $\pi_{\PP(\K_G)}(Y)$ in $\PP(\K_G)$. If $\nabla_G$ is a hypersurface, its defining polynomial $\Delta_G$ ($\nabla_G =: \{ \Delta_G = 0 \})$, which is unique up to scaling, is called the \emph{Landau discriminant polynomial}. If $\nabla_G$ has codimension greater than 1, we set $\Delta_G = 1$.
\end{definition}

\begin{theorem}\label{thm1}
For any Feynman diagram $G$, the Landau discriminant $\nabla_G$ is an irreducible, proper subvariety of $\PP(\K_G)$. 
\end{theorem}
\begin{proof}
First, we show that the fibers $\pi_X^{-1}(\alpha)$ are equidimensional linear spaces of dimension $\dim \K_G-\E_G$. Since the image of $\pi_X$ is closed in $X$ by \cite[Ch.~8, \S 5, Thm.~6]{cox2013ideals}, the proof of \cite[Thm.~1.26]{MR3100243} implies that $Y$ is irreducible of dimension $\dim \K_G - 2$. The fiber $\pi_X^{-1}(\alpha)$ consists of the points $q \in \PP(\K_G)$ such that 
\begin{equation} \label{eq:linlandau}
\frac{\partial \F_G(\alpha;q)}{\partial \alpha_1} = \cdots = \frac{\partial \F_G(\alpha;q)}{\partial \alpha_{\E_G}} = 0.
\end{equation}
These are linear equations in the coordinates of $q$, so they can be written in matrix format $M q = 0$. We show that $M$ has full rank for any $\alpha \in X$ by computing the maximal minor of $M$ corresponding to the columns indexed by the internal mass parameters $\m_1, \m_2, \ldots, \m_{E_G}$. This is the determinant of 
\be\begin{pmatrix}
\U_G(\alpha) \\ & \U_G(\alpha) \\ & & \ddots \\ 
& & & \U_G(\alpha)
\end{pmatrix} +  \begin{pmatrix}
 \frac{\partial \U_G}{\partial \alpha_1}(\alpha) \\ \frac{\partial \U_G}{\partial \alpha_2}(\alpha)  \\ \vdots \\ \frac{\partial \U_G}{\partial \alpha_{\E_G}}(\alpha)
\end{pmatrix} \begin{pmatrix}
\alpha_1 & \alpha_2 & \cdots \alpha_{\E_G}
\end{pmatrix}.
\ee
Using the identity $\det (A + u v^\top) = (1 + v^\top A^{-1} u) \det(A)$, this minor evaluates to 
\be
\left ( 1 + \sum_{e = 1}^{\E_G} \frac{\alpha_e}{\U_G(\alpha)} \frac{\partial \U_G}{\partial \alpha_e}(\alpha) \right) \U_G(\alpha)^{\E_G} = (\L_G+1) \U_G(\alpha)^{\E_G} \neq 0.
\ee
Here we used the fact that $\U_G$ is homogeneous of degree $\L_G$ in combination with Euler's rule. This shows that $M$ has rank $\E_G$.

Since $Y$ is irreducible of dimension $\dim \K_G - 2$, $\pi_{\PP(\K_G)}(Y)$ is irreducible of codimension at least 1, and so is its Zariski closure $\nabla_G$. 
\end{proof}
\begin{remark}[Dimension] \label{rem:dimension}
Thm.~\ref{thm1} does not make claims about the dimension of $\nabla_G$. In most of our examples, it is a hypersurface. However, as we will see in Sec.~\ref{subsec:generalparameters}, the Landau discriminant may have codimension $> 1$. In this case, the projection map $\pi_{\PP(\K_G)}: Y \rightarrow \PP(\K_G)$ has positive dimensional fibers. 
\end{remark}

\begin{remark}[Degree] \label{rem:degree}
Presently, there is no closed formula for the degree of $\nabla_G$ in terms of the combinatorics of the diagram $G$. Classical results from the theory of discriminants and resultants provide upper bounds (Prop.~\ref{prop:degAdisc} and \ref{prop:degmixedres}). We postpone these results to Sec.~\ref{sec:polytopes}, as they involve polytopes related to Symanzik polynomials and Landau equations. 
\end{remark}

\begin{remark}[Restriction] \label{rem:restriction}
In practice, we may only be interested in kinematic parameters inside a subvariety $\LLL$ of $\PP(\K_G)$. We will mostly consider restrictions to a linear subspace $\LLL  = \PP^q \subset \PP(\K_G)$. This allows, for instance, to set $\m_e = \m, \M_i = \M$, which is the physically meaningful case where external and internal masses are all equal. We replace $Y$ by $Y \cap (X \times \LLL)$. The closure of the projection of $Y$ to $\LLL$ is, with a slight abuse of terminology, also called the Landau discriminant. We will denote it by $\nabla_G(\LLL)$. The projection map will be denoted by $\pi_\LLL: Y \rightarrow \LLL$. We warn the reader that Thm.~\ref{thm1} only makes claims about the most general case $\LLL = \PP(\K_G)$. Once we restrict to a smaller $\LLL$, $Y$ and its projection may become reducible. Moreover, we will see that we may have $\nabla_G(\LLL) \subsetneq \nabla_G \cap \LLL$ (see, e.g., Ex.~\ref{ex:box2}). For $\LLL = \PP^q$, we will write $\Delta_G(\LLL)$ for the defining equation of the codimension 1 component(s) of $\nabla_G(\LLL)$. In case $\codim \nabla_G(\LLL) > 1$, we set $\Delta_G(\LLL) = 1$.
\end{remark}

To illustrate Thm.~\ref{thm1} and Rk.~\ref{rem:restriction}, we compute Landau discriminants in two well-known examples: the families of \emph{one-loop} and \emph{banana diagrams}.

\subsection{\label{sec:ngon}One-Loop Diagrams}
For the family of one-loop diagrams with $n$ external legs, $\mathtt{A}_n$, illustrated in Fig.~\ref{fig:An}, we have $\L_{\mathtt{A}_n}=1$, $\E_{\mathtt{A}_n} = n$ and the Symanzik polynomials are given by
\be
\U_{\mathtt{A}_n} = \sum_{i=1}^{n} \alpha_i, \qquad \F_{\mathtt{A}_n} = \frac{1}{2}\sum_{i,j=1}^{n} \mathbf{Y}_{ij} \alpha_i \alpha_j.
\ee 
Here the entries of the $n {\times} n$ symmetric matrix $\mathbf{Y}$ are given by
\be
\mathbf{Y}_{ij} := (\textstyle\sum_{k=i}^{j-1}p_k)^2 - \m_i - \m_j \text{ when } i < j, \qquad \text{ and } \mathbf{Y}_{ii} = -2\m_i.
\ee
When $i=j{-}1$ we have $\mathbf{Y}_{i,i+1} = \M_i - \m_i - \m_{i+1}$, and otherwise $(\textstyle\sum_{k=i}^{j-1}p_k)^2 = s_{i,i+1,\ldots,j-1}$ are Mandelstam invariants. We write $\mathbf{Y}(\LLL)$ for the restriction of $\mathbf{Y}$ to a subspace $\LLL \subset \PP(\K_{\mathtt{A}_n})$. The leading Landau equations \eqref{eq:LE} impose
\be
\mathbf{Y}  {\alpha} = 0, \qquad \alpha_1 \alpha_2 \cdots \alpha_n \left (\sum_{i=1}^{n} \alpha_i \right) \neq 0.
\ee
The following statement is an immediate consequence.
\begin{proposition}[One-loop diagrams]
The Landau discriminant $\nabla_{\mathtt{A}_n}(\LLL)$ is the Zariski closure of the subset of $\LLL \subset \PP(\K_{\mathtt{A}_n})$ defined by
\be\label{eq:conditiononeloop}
\det \mathbf{Y}(\LLL) = 0, \quad \textup{ker} \mathbf{Y}(\LLL) \not \subset \{ \alpha_1 \alpha_2 \cdots \alpha_n (\alpha_1 + \cdots + \alpha_n) = 0 \} .
\ee
\end{proposition}

\begin{example}[Bubble diagram]\label{ex:example4}
For $n=2$, we have the bubble diagram $\mathtt{A}_2$. Calling $\M_1 = \M_2 = s$ (which corresponds to attaching legs $1$ and $2$ on one end and $3$ and $4$ on the other end of the diagram, as a special case of Fig.~\ref{fig:BE} with $\E_G=2$), we find for generic internal masses $\m_1, \m_2$:
\be\label{eq:Delta-A2}
\Delta_{\mathtt{A}_2} = \det \left(\begin{array}{cc}
-2\m_1  & s-\m_1 - \m_2 \\
s-\m_1 - \m_2 & - 2\m_2
\end{array}\right) = 4 \m_1 \m_2 - (s - \m_1 - \m_2)^2.
\ee
This defines an irreducible curve in $\PP^2$ with coordinates $(s:\m_1:\m_2)$, dual to the rational curve in $(\PP^2)^\vee$ parametrized by $(\alpha_1\alpha_2:-\alpha_1(\alpha_1+\alpha_2):-\alpha_2(\alpha_1+\alpha_2))$. Restricting to the line $\LLL = \PP^1 = \{\m_1 = \m_2 = \m\} \subset \PP^2$ with coordinates $(s:\m)$, the determinant factors:
\be \label{eq:detYrestricted}
\det \mathbf{Y}(\LLL)= \det \left(\begin{array}{cc}
-2\m  & s-2\m \\
s-2\m & - 2\m
\end{array}\right) = -s (s-4\m).
\ee
The first component, $\{s=0\}$, is not in the discriminant: it corresponds to the null vector $ {\alpha} = (1,-1)$ whose entries sum to zero (instead, it is a second-type Landau singularity). The second factor, $\{s=4\m\}$, has $ {\alpha} = (1,1)$.
The discriminant polynomial in the equal-mass case is therefore
$\Delta_{\mathtt{A}_2}(\LLL) = s - 4\m$. In passing from general masses to equal masses, the degree of the discriminant drops by one: the variety of \eqref{eq:detYrestricted}, which is $\nabla_{\mathtt{A}_2} \cap \LLL$, strictly contains $\nabla_{\mathtt{A}_2}(\LLL) = \{(4:1)\}$.
\end{example}
\begin{remark}
We note that while $\Delta_{\mathtt{A}_2}$ given in \eqref{eq:Delta-A2} is irreducible, we can return to the variables $m_1$, $m_2$ (with $\m_i = m_i^2$), in terms of which $\Delta_{\mathtt{A}_2}$ for generic masses factors as
\be\label{eq:Delta-A2b}
\Delta_{\mathtt{A}_2} = -\left(s - (m_1 {+} m_2)^2\right)\left(s - (m_1 {-} m_2)^2\right).
\ee
The singularities associated to the two factors are known as the \emph{normal} and \emph{pseudo-normal} thresholds respectively. 
\end{remark}

\begin{example}[Box diagram]\label{ex:box2}
For the box diagram $\mathtt{A}_4$ with $n=4$ we have
\be
\Delta_{\mathtt{A}_4} = \det \left(
\begin{array}{cccc}
 -2 \m_1 & \M_1-\m_1-\m_2 & s-\m_1-\m_3
   & \M_4 -\m_1-\m_4 \\
 \M_1 -\m_1-\m_2 & -2 \m_2 &
   \M_2-\m_2-\m_3 & t-\m_2-\m_4 \\
 s-\m_1-\m_3 & \M_2-\m_2-\m_3 & -2 \m_3
   & \M_3 -\m_3-\m_4 \\
 \M_4-\m_1-\m_4 & t-\m_2-\m_4 &
   \M_3-\m_3-\m_4 & -2 \m_4 \\
\end{array}
\right)
\ee
and hence the discriminant polynomial is irreducible of degree $4$.
Restricting to the $\PP^3$ given by $\m_e = \m$ and $\M_i = \M$, we find
\begin{align}
\det \mathbf{Y} (\PP^3) &= \det \left(
\begin{array}{cccc}
 -2 \m & \M-2\m & s-2\m
   & \M -2\m \\
 \M -2\m & -2 \m &
   \M-2\m & t-2\m \\
 s-2\m & \M-2\m & -2 \m
   & \M - 2\m \\
 \M-2\m & t-2\m &
   \M-2\m & -2 \m \\
\end{array}
\right)\\
&= s\, t \left(st + 4 \m (4 \M-s-t)-4 \M^2\right).\nn
\end{align}
For $s=0$ and $t=0$ we find the null vectors to be respectively
\be
 {\alpha} = (1, 0, -1, 0) \qquad\text{and}\qquad  {\alpha} = (0,1,0,-1),
\ee
whose entries add up to zero in both cases, and some entries are zero. Therefore, the discriminant polynomial in the equal-mass case is
\be \label{eq:A4disc}
\Delta_{\mathtt{A}_4}(\PP^3) = st + 4 \m (4 \M-s-t)-4 \M^2.
\ee
This defines an irreducible quadratic surface strictly contained in $\Delta_{\mathtt{A}_4} \cap \PP^3$. One can check that the kernel condition in  \eqref{eq:conditiononeloop} is not satisfied for $s = \m = \M = 0$, yet \eqref{eq:A4disc} vanishes at this point in $\PP^3$. This is an example of a point in the kinematic space that is added to the discriminant by taking the closure of $\pi_{\PP^3}(Y)$. The same thing happens for $(s:t:\m:\M) = (1:0:0:0)$.
\end{example}

\subsection{\label{sec:banana}Banana Diagrams}

We start by presenting a new, explicit proof of the following well-known result. 
\begin{proposition}[Banana diagrams]
Substituting $\m_e = m_e^2$ in the Landau discriminant of the banana diagram with $\E$ internal edges ($\mathtt{B}_\E$ in Fig.~\ref{fig:BE}) gives
\be\label{eq:Delta-BE}
\Delta_{\mathtt{B}_E}(s, m_1^2, \ldots, m_{\E}^2) = \prod_{\{ \eta_e \} } \left( s - \Big( \sum_{e=1}^{\E} \eta_e m_e \Big)^2 \right),
\ee
where the product runs over all $2^{\E-1}$ projectively-inequivalent ways of assigning the signs $\eta_e \in \{\pm 1\}$ to each edge $e=1,2,\ldots,\E$.
\end{proposition}
\begin{proof}
The diagrams $\mathtt{B}_\E$ have $\L_{\mathtt{B}_\E}=\E{-}1$, $\E_{\mathtt{B}_\E} = \E$ and the Symanzik polynomials are given by
\be
\F_{\mathtt{B}_\E} = s \prod_{e=1}^{\E} \alpha_e - \U_{\mathtt{B}_\E} \sum_{e=1}^{\E} \m_e \alpha_e
\qquad\text{with}\qquad
\U_{\mathtt{B}_\E} = \sum_{e=1}^{\E} \prod_{\substack{e'=1\\ e' \neq e}}^{\E} \alpha_{e'}.
\ee
In order to simplify the notation let us introduce $\F_{\mathtt{B}_\E} = (\prod_{e=1}^{\E}\alpha_e) \widetilde{\F}$, where 
\be
\widetilde{\F} := s - \Big(\sum_{e=1}^{\E} \frac{1}{\alpha_e}\Big) \Big(\sum_{e=1}^{\E} \m_e \alpha_e\Big). 
\ee
Toric solutions have $\widetilde{\F}=0$. On the support of this constraint the leading Landau equations are
\be
\frac{\partial \F_{\mathtt{B}_\E}}{\partial \alpha_{e^\ast}} = \Big(\prod_{e=1}^{\E}\alpha_e\Big) \frac{\partial \widetilde{\F}}{\partial \alpha_{e^\ast}} = \frac{\prod_{e=1}^{\E}\alpha_e}{\alpha_{e^\ast}^2} \bigg[ \Big(\sum_{e=1}^{\E} \m_e \alpha_e\Big) - \m_{e^\ast} \alpha_{e^\ast}^2 \Big(\sum_{e=1}^{\E} \frac{1}{\alpha_e}\Big) \bigg] = 0,
\ee
together with the constraint $\U_{\mathtt{B}_\E} \neq 0$.
The term in the square brackets has to vanish for all $e^\ast$. We eliminate the term $\sum_e \m_e \alpha_e$ using $\widetilde{\F}=0$ (as well as $\U_{\mathtt{B}_\E}\neq 0$) to get
\be
\frac{s}{\left(\sum_{e=1}^{\E} \frac{1}{\alpha_e}\right)^2} = \m_{e^\ast} \alpha_{e^\ast}^2.
\ee
The term on the left-hand side is independent of the choice of $e^\ast$. It implies that for every pair of edges $e_1$ and $e_2$ the solution must satisfy
\be
\left(\frac{\alpha_{e_1}}{\alpha_{e_2}}\right)^{\!2} =  \frac{\m_{e_2}}{\m_{e_1}}.
\ee
Let us restore $\m_e = m_e^2$. The solutions of the Landau equations can be stated as
\be\label{eq:BE-alphas}
\left( \alpha_1 : \alpha_2 : \cdots : \alpha_\E \right) = \left( \frac{\eta_1}{m_1} : \frac{\eta_2}{m_2} : \cdots : \frac{\eta_\E}{m_\E}\right)
\ee
together with
\be
s = \Big( \sum_{e=1}^{\E} \eta_e m_e \Big)^2
\ee
for all $2^{\E-1}$ projectively inequivalent ways of assigning the signs $\eta_e \in \{\pm 1\}$ to each edge. The Landau discriminant polynomial is therefore given by \eqref{eq:Delta-BE}, up to a constant.
\end{proof}

\begin{example}[Bubble diagram]
Let us consider the case $\E=2$, which gives the same bubble diagram as the one studied in Ex.~\ref{ex:example4}. The two inequivalent ways of assigning signs are
\be
(\eta_1, \eta_2) \in \left\{ (1,1), \; (1,-1) \right\},
\ee
which according to \eqref{eq:Delta-BE} gives
\be
\Delta_{\mathtt{B}_2} = \left(s - (m_1 {+} m_2)^2\right)\left(s - (m_1 {-} m_2)^2\right),
\ee
in agreement with the previous result \eqref{eq:Delta-A2b}, $\Delta_{\mathtt{A}_2} = - \Delta_{\mathtt{B}_2}$.
\end{example}

%% file: section3.tex
In this section, we focus on the computation of the Landau discriminant $\nabla_G$ using methods from nonlinear algebra. We will focus on the case where $\nabla_G$ is a hypersurface in the projectivized kinematic space $\PP(\K_G)$, and our goal is to compute its defining equation $\Delta_G = 0$.
There are several ways to go about this. Depending on which tools are used, one strategy may be preferable over another. In Sec.~\ref{subsec:elimination}, we briefly discuss two approaches using classical elimination theory. In contrast, Sec.~\ref{subsec:sampling} describes numerical sampling methods using state-of-the-art technology from numerical nonlinear algebra. We use the latter approach to compute the Landau discriminant of the diagram $\mathtt{env}$ from Fig.~\ref{fig:env}, assuming equal internal masses ($\m_1 = \cdots = \m_{6} = \m$) and equal external masses ($\M_1 = \cdots = \M_4 = \M$). This is a reducible surface of degree $45$ in $\LLL = \PP^3$.

\subsection{\label{subsec:elimination}Symbolic Elimination Methods}

For a given $G$, the polynomials $\frac{\partial \F_G}{\partial \alpha_e}$ from \eqref{eq:LE} generate an ideal $I$ in the homogeneous coordinate ring
\be
\C[\PP(\K_G) \times \PP^{\E_G - 1}] = \C[s_I, \M_i, \m_e,\alpha_1, \ldots, \alpha_{\E_G}]
\ee
of $\PP(\K_G) \times \PP^{\E_G - 1}$. Its associated subvariety of $\PP(\K_G) \times \PP^{\E_G - 1}$ contains the Zariski closure $\overline{Y} \subset \PP(\K_G) \times \PP^{\E_G - 1}$ of the incidence variety $Y$ defined in \eqref{eq:incidence}. In order to eliminate spurious components, we must saturate\footnote{The \emph{saturation} of an ideal $I$ of a ring $R$ by a polynomial $g \in R$ is the ideal $(I:g^\infty) = \{ f \in R ~|~ g^k f \in I \text{ for some } k \in {\mathbb{N}} \}$. For more information and a geometric interpretation, see \cite[Ch.~4, \S 4]{cox2013ideals}.} the ideal $I$ by the polynomial $g = (\alpha_1 \cdots \alpha_{\E_G}) \U_G$: 
\be
V_{\PP(\K_G) \times \PP^{E_G - 1}} ( I : g^\infty ) = \overline{Y}.
\ee
In practice, it is more efficient to saturate by each of the factors of $g$ separately, i.e.
\be \label{eq:seqElim}
(I:g^\infty) = (((\cdots(I:\alpha_1^\infty): \cdots ):\alpha_{\E_G}^\infty):\U_G^\infty) .
\ee
The projection map $\pi_{\PP(\K_G)}$ extends naturally from $Y$ to $\overline{Y}$, and since $\PP^{\E_G}$ is complete, we have $\pi_{\PP(\K_G)}(\overline{Y}) = \overline{\pi_{\PP(\K_G)}(Y)} = \nabla_G$. Therefore, $\Delta_G$ is the generator of the elimination ideal $(I:g^\infty) \cap \C[s_I,\M_i,\m_e]$.

\begin{example}[Box diagram]
Consider the box diagram $\mathtt{A}_4$ from Ex.~\ref{ex:box} and \ref{ex:box2}. We define its Symanzik polynomials in the computer algebra software \texttt{Macaulay2} \cite{M2} as follows: 
\begin{minted}{julia}
R = QQ[a_1..a_4,M_1..M_4,m_1..m_4,s,t]

U = a_1+a_2+a_3+a_4;
F1 = a_1*a_2; F2 = a_2*a_3;
F3 = a_3*a_4; F4 = a_1*a_4;
F12 = a_1*a_3; F23 = a_2*a_4;

F = (s*F12 + t*F23 + M_1*F1 + M_2*F2 + M_3*F3 + M_4*F4
    - U*(m_1*a_1+m_2*a_2+m_3*a_3+m_4*a_4));
\end{minted}
The ideal $I$ is defined by 
\begin{minted}[firstnumber=last]{julia}
I = ideal apply(4, i->diff(a_(i+1), F))
\end{minted}
We saturate by the polynomial $g$. In this case, it turns out that $(I:g^\infty) = (I:\alpha_1^\infty)$. Geometrically, this means that in this case the spurious component of $V_{\PP(\K_{\mathtt{A}_4}) \times \PP^3}(I)$ contained in $\{g = 0\}$ is contained in $\{\alpha_1 = 0\}$. The \texttt{Macaulay2} command for saturating by $\alpha_1$ is
\begin{minted}[firstnumber=last]{julia}
J = saturate(I, a_1)
eliminate(J, {a_1,a_2,a_3,a_4})
\end{minted}
Here we also eliminated the Schwinger parameters, resulting in the Landau discriminant polynomial $\Delta_{\mathtt{A}_4}$, which is homogeneous of degree 4 and equal to the determinant in Ex.~\ref{ex:box2}.
\end{example}

An alternative way to take the non-vanishing of $g = (\alpha_1 \cdots \alpha_{\E_G}) \U_G$ into account is to work directly in the ring $\C[\PP(\K_G)] \otimes \C[X]$, where $\C[\PP(\K_G)] = \C[s_I,\M_i,\m_e]$ is the homogeneous coordinate ring of $\PP(\K_G)$ and $\C[X]$ is the coordinate ring of the affine variety $X = \PP^{\E_G - 1} \setminus V(g)$. We use the representation
\be
\C[X] = \C[\alpha_1, \ldots, \alpha_{\E_G-1}, y]/\langle 1 - y \tilde{g} \rangle,
\ee
where $\tilde{g}$ is obtained from setting $\alpha_{\E_G} = 1$ in $g$. Let $\tilde{I} \subset \C[s_I,\M_i,\m_e, \alpha_1, \ldots, \alpha_{\E_G-1}, y]$ be the ideal
\be
    \tilde{I} = \left \langle \frac{\partial \F_G}{\partial \alpha_e} \bigg|_{\alpha_{\E_G} = 1},\quad e = 1, \ldots, \E_G \right \rangle + \langle 1 - y \tilde{g} \rangle.
\ee
The elimination ideal $\tilde{I} \cap \C[s_I,\M_i,\m_e]= \langle \Delta_G \rangle$ is again generated by the Landau discriminant polynomial. 
\begin{example}
We apply the alternative technique for computing $\Delta_G$ to the box diagram $G=\mathtt{A}_4$. We add the variable $y$ to the ring $\mathtt{R}$ in Ex.~\ref{ex:box2} and execute 
\begin{minted}[firstnumber=10]{julia}
Itilde = sub(I, a_4=>1) + ideal(a_1*a_2*a_3*sub(U,a_4=>1)*y-1)
eliminate(Itilde, {a_1,a_2,a_3,y})
\end{minted}
This gives indeed the same answer as before. 
\end{example}

The symbolic elimination methods outlined here cannot deal with much larger examples, as the number of variables involved in the computation grows too big. However, as mentioned above (Rk.~\ref{rem:restriction}), it is often meaningful to make some simplifying assumptions on the kinematic parameters. For instance, the Feynman diagram $G = \mathtt{acn}$ with equal internal masses $\m_1 = \cdots = \m_5 = \m$ and equal external masses $\M_1 = \cdots = \M_4 = \M$ may be analyzed using these techniques. 

\begin{example}\label{ex:acnode}
The first Symanzik polynomial for the Feynman diagram $G=\mathtt{acn}$ from Fig.~\ref{fig:acn} is 
\be
\U_{\mathtt{acn}} = \alpha_5(\alpha_1+ \alpha_2 + \alpha_3+\alpha_4) + (\alpha_4 + \alpha_1)(\alpha_2 + \alpha_3).
\ee
with the assumptions $\m_e = \m$, $\M_i = \M$, the second Symanzik polynomial is 
\begin{align}
    \F_{\mathtt{acn}} &= s \alpha_1 \alpha_3 \alpha_5 + t \alpha_2 \alpha_4 \alpha_5 + \M~ [\alpha_1\alpha_2\alpha_5 + \alpha_2\alpha_3(\alpha_1 + \alpha_4 + \alpha_5) \\
    &\quad + \alpha_3\alpha_4\alpha_5 + \alpha_1 \alpha_4 (\alpha_2 + \alpha_3 + \alpha_5) ] - \m\, \U_{\mathtt{acn}}   ({\textstyle\sum}_{e=1}^5 \alpha_e  ) .\nn
\end{align} 
Setting $I = \langle \partial \F_{\mathtt{acn}}/\partial \alpha_e, e = 1, \ldots, 5 \rangle$, saturating by $g = (\alpha_1 \cdots \alpha_5) \cdot \U_{\mathtt{acn}}$ and eliminating $\alpha_1, \ldots, \alpha_5$ gives the Landau discriminant, which is a surface in $\PP^3$ with two irreducible components: $\Delta_{\mathtt{acn}}(\PP^3)$ factors as $\Delta_{\mathtt{acn}}(\PP^3) = \Delta_{\mathtt{acn},1} \cdot \Delta_{\mathtt{acn},2}$ with
\begin{dmath*}
\Delta_{\mathtt{acn},1} = 9 \M^4 \m^2 s^2-54 \M^3 \m^3 s^2+81 \M^2 \m^4 s^2+9 \M^2 \m^3 s^3-54 \M \m^4 s^3+9 \m^4 s^4 +16 \M^6 s t 
    -144 \M^5 \m s t+450 \M^4 \m^2 s t-540 \M^3 \m^3 s t+162 \M^2 \m^4 s t+12 \M^4 \m s^2 t-126 \M^3 \m^2 s^2 t 
    +297 \M^2 \m^3 s^2 t -162 \M \m^4 s^2 t+9 \M^2 \m^2 s^3 t+36 \m^4 s^3 t-10 \m^3 s^4 t+9 \M^4 \m^2 t^2-54 \M^3 \m^3 t^2 
    +81 \M^2 \m^4 t^2+12 \M^4 \m s t^2 -126 \M^3 \m^2 s t^2+297 \M^2 \m^3 s t^2-162 \M \m^4 s t^2-8 \M^4 s^2 t^2 
    +84 \M^3 \m s^2 t^2-189 \M^2 \m^2 s^2 t^2+54 \m^4 s^2 t^2-11 \M^2 \m s^3 t^2+42 \M \m^2 s^3 t^2-30 \m^3 s^3 t^2 
    +\m^2 s^4 t^2+9 \M^2 \m^3 t^3-54 \M \m^4 t^3+9 \M^2 \m^2 s t^3+36 \m^4 s t^3  -11 \M^2 \m s^2 t^3+42 \M \m^2 s^2 t^3 
    -30 \m^3 s^2 t^3+\M^2 s^3 t^3-4 \M \m s^3 t^3+2 \m^2 s^3 t^3+9 \m^4 t^4-10 \m^3 s t^4+\m^2 s^2 t^4,
\end{dmath*} 
as well as
\begin{align}
\Delta_{\mathtt{acn},2} = \m^2 \left(4 \M^2-5 \M (s+t)+(s+t)^2\right)-\m s t (s+t-5 \M)-\M^2 s t.
\end{align}
It turns out that the diagram $\mathtt{acn}$ provides an example for which the restriction $\U_G \neq 0$ in \eqref{eq:LE} shrinks the Landau discriminant significantly. If we saturate the ideal $I$ by the product $\alpha_1 \cdots \alpha_5$ (instead of by $g$) and then eliminate the $\alpha_e$'s, we obtain an extra factor $\Delta_{\mathtt{acn},0} = 4 \M - s - t$ in the generator of the elimination ideal. We now explain that this extra factor comes from spurious solutions to \eqref{eq:LE} satisfying $\U_{\mathtt{acn}} = 0$. Define the line $ L = \{ \alpha_1 = -\alpha_2 = \alpha_3 = -\alpha_4 \}$ in $\PP^4 \supset X$. Restricting the Landau equations to this line, they simplify to 
\begin{align}
      0 &= -\alpha_1^2 \M + \alpha_1\alpha_5 (-2\M - 4\m + s) - \m\alpha_5^2,\nn \\
      0 &=  -\alpha_1^2 \M + \alpha_1\alpha_5 (2\M - 4\m -t) - \m\alpha_5^2,\nn \\
      0 &= -\alpha_1^2 \M + \alpha_1\alpha_5 (-2\M - 4\m + s) - \m\alpha_5^2, \\
      0 &=  -\alpha_1^2 \M + \alpha_1\alpha_5 (2\M - 4\m -t) - \m\alpha_5^2,\nn \\
      0 &= -\alpha_1^2(4\M - s - t).\nn
      \end{align}
This shows that for general kinematic parameters satisfying $4 \M - s - t = 0$, the Landau equations have two solutions on $L \setminus V(\alpha_1 \cdots \alpha_5)$. The first Symanzik polynomial vanishes on both of these solutions, as one can easily check that it vanishes identically on $L$. 
\end{example}

\subsection{\label{subsec:sampling}Numerical Sampling Methods}

We now turn to numerical methods for computing the Landau discriminant. The punchline is that, although symbolic verification remains valuable when possible, methods from numerical nonlinear algebra can be used to compute $\Delta_G$ and some of its invariants in an efficient way. The basic strategy is to \emph{sample} the discriminant using numerical continuation techniques and then interpolate the sample points to obtain the defining equation $\Delta_G$. The methods outlined here are implemented in a \texttt{Julia} package \texttt{Landau.jl}, made available at \url{https://mathrepo.mis.mpg.de/Landau/}. The section will serve as a short tutorial on how to use some functions in this package. As a running example, we will use the acnode diagram $G = \texttt{acn}$ with equal internal and external masses. 

Our starting point is the system of polynomial equations 
\begin{equation} \label{eq:incidenceeq}
\frac{\partial\F_G}{\partial \alpha_1} \bigg|_{\alpha_{\E_G} = 1} = \cdots = \frac{\partial \F_G}{\partial \alpha_{\E_G}} \bigg|_{\alpha_{\E_G} = 1}  = 1 - y\tilde{g}= 0, 
\end{equation}
where $\tilde{g}$ is obtained as in the previous section, by substituting $\alpha_{\E_G} = 1$ in $g = (\alpha_1 \cdots \alpha_{\E_G}) \U_G$. These are the defining equations of the incidence variety $Y$ from \eqref{eq:incidence} embedded in $\C^{\E_G} \times  \PP(\K_G)$, where $\C^{\E_G}$ has coordinates $\alpha_1, \ldots, \alpha_{\E_G-1}, y$.
Using \texttt{Landau.jl}, the Landau equations are generated by specifying the edges and the \emph{node} labels. Here the \emph{nodes} are the vertices of the graph $G$ to which the external legs with the momenta $p_i$ are attached, in the order of appearance. In the notation of Fig.~\ref{fig:acn}, for $G = \mathtt{acn}$ we have
\begin{minted}{julia}
edges = [[4,1],[1,2],[2,3],[3,4],[1,3]]
nodes = [1,2,3,4]

LE, y, α, p, mm = affineLandauEquations(edges, nodes)
\end{minted} 

The Landau equations are stored in \texttt{LE}. The output elements \texttt{y}, \texttt{α} represent the auxiliary variable $y$ and the list of Schwinger parameters $\alpha_e$, while \texttt{p} and \texttt{mm} represent lists of the external momenta and internal masses, which will later be substituted by Mandelstam invariants. As in Rk.~\ref{rem:restriction}, we may restrict the kinematic parameters to a linear subspace $\LLL = \PP^q \subset \PP(\K_G)$ and replace $Y$ by $Y \cap (\C^{\E_G} \times \LLL) \subset  \C^{\E_G} \times \LLL$. The Landau discriminant is $\nabla_G(\LLL) = \overline{\pi_{\LLL}(Y)}$. We will focus on its components of dimension $q - 1$. Here is how to restrict to the linear subspace $\PP^q = \PP^3$ of equal external and internal masses. 
\begin{minted}[firstnumber=last]{julia}
LE, s, t, M, m = substitute4legs(LE, p, mm; equalM = true, equalm = true)
\end{minted} 
Let us consider the problem of estimating the degree of $\nabla_G(\LLL)$. We write $(z_0:z_1: \cdots:z_q)$ for the homogeneous coordinates on $\LLL$. Adding $q-1$ random linear equations 
\begin{equation} \label{eq:lineq}
a_{i,0} z_0 + \cdots + a_{i,q} z_q = 0, \quad a_{i,j} \in \C, ~ i = 1, \ldots, q-1 
\end{equation}
to \eqref{eq:incidenceeq} has the geometric interpretation of slicing $\pi_{\LLL}(Y)$ with a line $L$. The solutions to this larger set of  equations \eqref{eq:incidenceeq} and \eqref{eq:lineq} form the pre-image $\pi_{\LLL}^{-1}(\pi_{\LLL}(Y) \cap L) = \pi_{\LLL}^{-1}(L) \cap Y$. By genericity of the coefficients $a_{i,j}$, with probability one, we have $|\pi_{\LLL}(Y) \cap L| = | \nabla_G(\LLL) \cap L | = \deg \nabla_G(\LLL)$. This gives the following algorithm to compute $\deg \nabla_G(\LLL)$:
\begin{enumerate}[leftmargin=2.5em]
\item Compute $S = \pi_{\LLL}^{-1}(L) \cap Y$ by solving \eqref{eq:incidenceeq} $+$ \eqref{eq:lineq}. 
\item Count the number of distinct points in the projection $ \pi_{\LLL}(S)$. 
\end{enumerate}
Here, step 1 can be performed using software from numerical nonlinear algebra. We choose to use \emph{numerical homotopy continuation}. For all our computations, we use the \texttt{Julia} package \texttt{HomotopyContinuation.jl}. We dehomogenize by setting, for instance, $z_0 = 1$. In step 2, the points in the approximate solution set $S$ are considered equal or distinct according to some sensible heuristic, e.g. based on their relative distance $\lVert s - s' \rVert/\lVert s \rVert$. This is implemented in the function \texttt{degreeProjection} in \texttt{Landau.jl}: 
\begin{minted}[firstnumber=last]{julia}
vlist = [α[1:end-1];y]; plist = [M;m;s;t]
dproj = degreeProjection(LE, vlist, plist)
\end{minted} 
Here \texttt{vlist} is a vector of affine coordinates on $\C^{\E_G}$ and \texttt{plist} is a vector of coordinates on $\LLL = \PP^3$. As a by-product of this degree computation, we obtain $\deg \nabla_G(\LLL)$ points on the Landau discriminant $\nabla_G(\LLL)$. Essentially, in the language of numerical algebraic geometry, we have computed a \emph{(pseudo-)witness set} for $\nabla_G(\LLL)$ \cite{hauenstein2010witness}. Doing this for different lines $L$, we may sample $\nabla_G(\LLL)$ at will. Let $L ' \neq L$ be a new generic line in $\LLL$. From a homotopy continuation point of view, we prefer to use the previously computed points $S$ in order to obtain $S' = \pi_{\LLL}^{-1}(L') \cap Y$, rather than to start anew from scratch. We deform the line $L$ continuously into $L'$ by introducing a parameter into the equations and tracking the paths described by the points in $S$ along this deformation to end up in the points $S'$. A new set of $\deg \nabla_G(\LLL)$ points on $\nabla_G(\LLL)$ is obtained from $\pi_{\LLL}(S')$. This procedure is carried out repeatedly by the function \texttt{sampleProjection} in \texttt{Landau.jl}. Here is how to collect 500 samples:
\begin{minted}[firstnumber=last]{julia}
samp, R, A, b, H = sampleProjection(LE, vlist, plist; npoints = 500)
\end{minted} 
The output \texttt{H} contains the equations \eqref{eq:incidenceeq} $+$ \eqref{eq:lineq} used to find the first set of points (i.e., the \emph{starting solutions} for the homotopy). The line $L$ defined by \eqref{eq:lineq}, together with some random dehomogenization in $\LLL$, is given by \texttt{A*plist + b = 0}. The output \texttt{R} is a result returned by \texttt{HomotopyContinuation.jl} when solving \texttt{H}, and \texttt{samp} contains a list of at least 500 samples. The number 500 is chosen arbitrarily here. The number of sample points should be at least the dimension of the space of homogeneous polynomials of degree \texttt{dproj} in \#\texttt{plist} variables, minus one. This is because the remaining step is to interpolate these sample points by such a polynomial. In practice, for numerical reasons, we use more than the minimal amount of sample points. Our heuristic is to use at least $1.2 ~\times$ the minimal number of samples. Interpolation is done via
\begin{minted}[firstnumber=last]{julia}
disc, c, gap = interpolate_deg(samp, dproj, plist; homogeneous = true)
\end{minted} 
Here $\texttt{disc}$ is the Landau discriminant polynomial $\Delta_G(\LLL)$, \texttt{c} contains its coefficients and \texttt{gap} represents the ratio of the two smallest singular values of the coefficient matrix $A$ in the interpolation problem. The latter serves as a measure of trust in the computation: the size of the gap governs the sensitivity of the kernel of $A$ to perturbations in its entries \cite{stewart1991perturbation}. A large gap corresponds to a well-conditioned interpolation problem. In this example, we find \texttt{gap = 2.1438277266482883e10}. This means that the smallest singular value that was considered numerically nonzero is about $10^{10}$ times larger than the last singular value. A larger gap can sometimes be obtained by using more sample points. 

The coefficients of \texttt{disc} are floating point numbers, many of which are close to zero. We approximate these by rational numbers using 
\begin{minted}[firstnumber=last]{julia}
ratdisc = rat(disc)
\end{minted} 
The result is a reducible polynomial of degree 12, equal (up to a nonzero rational factor) to  $\Delta_{\mathtt{acn}} (\LLL)$ in Ex.~\ref{ex:acnode}.

\begin{remark}[Non-reduced incidence schemes]
For some Feynman diagrams $G$, the equations \eqref{eq:incidenceeq} do not define the vanishing ideal of the incidence variety $Y$. That is, the solution set in $\C^{\E_G} \times \LLL $ is $Y$, but the ideal generated by the $\E_G + 1$ polynomials is strictly smaller than the ideal of polynomials vanishing on $Y$. In such cases, \eqref{eq:incidenceeq} may define $Y$ with a certain \emph{multiplicity}. This happens, for instance, for $G = \mathtt{npltrb}$. The solutions to \eqref{eq:incidenceeq} $+$ \eqref{eq:lineq} are isolated points with multiplicity greater than one, which calls for more brute force sampling techniques based on homotopy \emph{end-games}. In \texttt{Landau.jl}, these methods are invoked by adding the option \texttt{findSingular = true} in \texttt{sampleProjection}. This makes the computations more time consuming, but it allows to deal with such singular components. 
\end{remark}

The approach outlined above works in general, under the assumption that $\nabla_G(\LLL) \subset \LLL$ has codimension 1. As the acnode example illustrates, reducing to $\LLL \subset \PP(\K_G)$ often leads to reducible Landau discriminants. In the rest of this section, we show how homotopy techniques provide a natural way of computing the irreducible factors of $\Delta_G(\LLL)$ separately. This leads to smaller interpolation problems that are numerically better behaved, and to faster computations. 

Let $\nabla_{G,i} \subset \nabla_{G}(\LLL) \subset \LLL$ be an irreducible component and let $Y_i = \pi_{\LLL}^{-1}(\nabla_{G,i})$. In general, $Y_i$ may consist of several irreducible components. Let $Y_{i,j}$ be any such component whose projection $\pi_{\LLL}(Y_{i,j})$ is dense in $\nabla_{G,i}$. The set of points $S_{i,j} = \pi_{\LLL}^{-1}(L) \cap Y_{i,j}$ projects to $\deg \nabla_{G,i}$ distinct points in $\nabla_{G,i}$. Moreover, the same is true for the points $S_{i,j}' = \pi_{\LLL}^{-1}(L') \cap Y_{i,j}$, obtained via continuation by continuously moving $L$ to $L'$. For reducible $Y$, the monodromy group of 
\be
\{(y,L) \in Y \times \textup{Gr}(2,q+1) ~|~ y \in \pi_{\LLL}^{-1}(L) \} \longrightarrow \textup{Gr}(2,q+1) \ee
acts non-transitively on a general fiber \cite{hauenstein2018numerical}. Therefore, the partitioning of $S$ into the groups $S_{i,j}$ can be realized using \emph{monodromy loops}. With this partitioning, we can sample the components $\nabla_{G,i}$ separately. 
The following command returns a list containing one representative for each of the solution groups $S_{i,j}$.
\begin{minted}[firstnumber=last]{julia}
reps = decompose(H, solutions(R), [A[:]; b])
\end{minted} 
In our acnode example, \texttt{reps} has two elements. We sample the component of $\nabla_G(\LLL)$ corresponding to the first representative by feeding a monodromy seed to the function \texttt{sampleProjection}. This consists of the first representative solution in \texttt{reps} and the parameter values \texttt{A}, \texttt{b} for the line $L$. 
\begin{minted}[firstnumber=last]{julia}
samp1, R1, A, b, H = sampleProjection(LE, vlist, plist; npoints = 500,
                               seedsol = reps[1], seedA = A, seedB = b)
\end{minted} 
An analogous syntax is used to estimate the degree of this component and to interpolate the samples:
\begin{minted}[firstnumber=last]{julia}
dproj1 = degreeProjection(LE, vlist, plist, seedsol = reps[1], 
                                          seedA = A, seedB = b)
disc1, N1, gap1 = interpolate_deg(samp1, dproj1, plist; homogeneous = true)                                     
\end{minted} 
The rationalization \texttt{rat(disc1)} gives $\Delta_{\mathtt{acn},1}$ from Ex.~\ref{ex:acnode} up to a nonzero rational factor. An analogous computation gives the second component $\Delta_{\mathtt{acn},2}$. As a check of correctness, we compute 
\begin{minted}[firstnumber=last]{julia}
expand(462//300*rat(disc) + rat(disc1)*rat(dics2))
\end{minted}
which gives \texttt{0}. The factor $462/300$ comes from the fact that \texttt{interpolate\_deg} returns a polynomial whose largest coefficient has modulus 1.
Notably, the singular value gaps \texttt{gap1} and \texttt{gap2} are $\approx 10^3$ times larger than \texttt{gap}.

\begin{remark}[Higher precision arithmetic] \label{rem:hp}
In the case where $\Delta_G(\LLL)$ or its factors have coefficients of strongly varying magnitude in the interpolation basis (here chosen as monomials), it might be necessary to use augmented precision in order to make reasonable rational approximations with the function \texttt{rat}. The package \texttt{Landau.jl} offers a higher precision version of \texttt{sampleProjection}, called \texttt{sampleProjection\_HP}, which computes points on $\Delta_G(\LLL)$ in \texttt{Julia}'s \texttt{BigFloat} format. These high precision coordinates are obtained by performing Newton iterations on the initial set of double precision solutions.
\end{remark}

\begin{remark}[Iterative methods for interpolation] \label{rem:iter}
Computing the coefficients of the interpolant through a collection of sample points requires the computation of a vector in the kernel of a complex matrix $M$. The standard way of doing this is via the singular value decomposition (SVD). However, if $M$ is too large, this may be infeasible. We observe that if the kernel of $M$ has dimension one, then it is spanned by the eigenvector of $M^H \cdot M$ corresponding to the eigenvalue 0 (here $\cdot^H$ is the Hermitian transpose). This eigenvector may be computed efficiently using iterative methods, such as the \texttt{eigs} function implemented in \texttt{Arpack.jl}. This may give less accurate results than for the SVD. However, it could give us an idea of which coefficients of the discriminant are zero. Using this information, the columns of $M$ corresponding to zero coefficients can be dropped, reducing the complexity of the kernel computation. 
\end{remark}

\subsection{Computational Results}

\subsubsection{Landau Discriminants in $\PP(\K_G)$} \label{subsec:generalparameters}
We use \texttt{Landau.jl} to compute dimension and degree of the Landau discriminants $\nabla_G$ corresponding to the diagrams in Fig.~\ref{fig:diagrams} (in the most general case, where $\LLL = \PP(\K_G)$). This can be done using the following three lines of code.
\begin{minted}{julia}
LE, y, α, p, m = affineLandauEquations(edges, nodes)
LE, s, t, M, m = substitute4legs(LE, p, m; equalM = false, equalm = false)
deg = degreeProjection(LE, [α[1:end-1];y], [s;t;M;m])
\end{minted}
For the only five leg diagram $G = \mathtt{pentb}$, we replace the last two lines by
\begin{minted}{julia}
LE, s12, s23, s34, s45, s51, M, m = substitute5legs(LE, p, m; equalM = false,
                                                              equalm = false)
deg = degreeProjection(LE, [α[1:end-1];y], [s12; s23; s34; s45; s51; M; m])
\end{minted}
The result is shown in the left half of Tab.~\ref{tab:dimdeg}.
\begin{table}[t]
\centering
\begin{tabular}{ll|ccc|ccc}
 & Diagram $G$       & $\textup{codim} \nabla_G$ & $\deg \nabla_G$ & time (sec) & $\nabla_G(\LLL)$ & $t_{\texttt{symb}}$ (sec) & $t_{\texttt{num}}$ (sec)     \\ \hline
 & $\mathtt{par}$    & 1                         & 6 & 0.176 & $[1,2]_1$  & 0.2 & 0.5               \\
 & $\mathtt{acn}$    & 1                         & 16 & 0.489   & $[4,8]_1$ & 175306.0 &  1.4              \\
 & $\mathtt{env}$    & 1                         & 114 & 13.1  & $[8,8,8,9,12]_1$ & $\times$ & 1226.1       \\
 & $\mathtt{npltrb}$ & 2                         & 10 & 37.2   & $[1,1]_1$ & 1.9 & 4.0             \\
 & $\mathtt{tdetri}$ & 1                         & 12 & 1.04  & $[2]_1, [1]_2$ & 8.1 & 1.2     \\
 & $\mathtt{debox}$  & 1                         & 8 & 0.366 & $[3]_1, [1]_2$ & 7.9 & 0.5     \\
 & $\mathtt{tdebox}$ & 1                         & 16 & 10.5 & $[2]_1, [1]_2$ & 1476.8 & 4.3     \\
 & $\mathtt{pltrb}$  & 2                         & 9 & 24.3  & $[1,1]_2$ & 0.6 & $\times$                         \\
 & $\mathtt{dbox}$   & 1                         & 12 & 8.64  & $[2,4]_1$ & 13634.2 & 4.5              \\
 & $\mathtt{pentb}$  & 1                         & 14 & 62.8  & $[12]_1$  & $\times$ & 815.9             
\end{tabular}
\caption{Dimension and degree of the Landau discriminants corresponding to the diagrams in Fig.~\ref{fig:diagrams}, with computation times. Here $\LLL$ is the linear subspace of $\PP(\K_G)$ where $\M_i = \M$ and $\m_e = \m$.}
\label{tab:dimdeg}
\end{table}
By Thm.~\ref{thm1}, each $\nabla_G \subset \PP(\K_G)$ is irreducible. For the diagrams $G = \mathtt{npltrb}, \mathtt{pltrb}$, the discriminant has codimension 2 in $\PP(\K_G)$. In order to compute its degree, we add the option \texttt{codimen = 2} in \texttt{degreeProjection}:
\begin{minted}[firstnumber=last]{julia}
deg = degreeProjection(ALE, [α[1:end-1];y], [s;t;M;m]; codimen = 2)
\end{minted}
For $G = \mathtt{pltrb}$. We checked symbolically that for equal external masses and generic external masses ($\LLL = \PP^8$ has coordinates $(s:t:\M:\m_1:\cdots:\m_6)$), $\nabla_G(\LLL)$ is defined by two equations:
\small
\begin{align*}
\m_1\m_4(\m_1-\m_2+\m_4)+((-\m_1+\m_2+\m_5)\m_2-(\m_1+\m_2)\m_4)\m_5 +\m_6 s^2 \qquad\qquad\;\;\;& \\
+(\m_1-\m_5)(\m_2-\m_4)\m_6 +((\m_1-\m_2)(\m_5-\m_4) -(\m_1 +\m_2 +\m_4 +\m_5 )\m_6 +\m_6^2 )s &= 0, \\
\M (\m_2 -\m_4)^2+(\M^2 -\M(\m_2 +2\m_3+\m_4) +(\m_2 -\m_3)(\m_4 - \m_3) )s+\m_3 s^2 &= 0.
\end{align*}
\normalsize
We find numerically that general fibers of the projection map $\pi_{\PP(\K_G)} : Y \rightarrow \nabla_G$ are curves of degree 5. For $G = \mathtt{npltrb}$, this degree is 8. 

\subsubsection{Equal-Mass Case} \label{subsubsec:equalmass}
We now consider the case where $\LLL = \PP^q \subset \PP(\K_G)$ is the $q$-dimensional subspace for which all external and internal masses are equal: $\M_i = \M, \m_e = \m$. In the case where $G$ has 4 legs, the space $\LLL = \PP^3$ has coordinates $(s:t:\M:\m)$. In case $n_G = 5$, the coordinates on $\LLL = \PP^6$ are $(s_{12}:s_{23}:s_{34}:s_{45}:s_{51}:\M:\m)$. As we have seen in previous examples, and as was pointed out in Rk.~\ref{rem:restriction}, the Landau discriminant $\nabla_G(\LLL)$ might be reducible. In fact, its irreducible components may have different dimensions. In Tab.~\ref{tab:dimdeg}, we encoded the components of $\nabla_G$ for all $G$ from Fig.~\ref{fig:diagrams} with their dimension and degree in the following way. A bracket $[d_1, \ldots, d_k]_c$ indicates that $\nabla_G$ has $k$ irreducible components of codimension $c$ with degrees $d_1, \ldots, d_k$. These numbers were obtained by using the methods in Sec.~\ref{subsec:elimination} and/or the methods in Sec.~\ref{subsec:sampling}. The columns $t_{\texttt{symb}}$ and $t_{\texttt{num}}$ report computation times for the symbolic and the numerical approach respectively. The symbolic method we opted for was the approach based on sequential saturation, as in \eqref{eq:seqElim}, since this gave the best results. All computations were performed using \texttt{HomotopyContinuation.jl} v2.6.0 on a 16 GB MacBook Pro with an Intel Core i7 processor working at 2.6 GHz. We warn the reader that the symbolic method (implemented in \texttt{Macaulay2}) computes the full elimination ideal. The column $t_{\texttt{num}}$ only comprises the time for computing the codimension 1 components of $\nabla_G(\LLL)$. In the rows of Tab.~\ref{tab:dimdeg} for which $t_{\texttt{symb}} = \times$, the \texttt{Macaulay2} computation did not finish within reasonable time. In all other cases, the discriminants $\nabla_G(\LLL)$ are computed \emph{exactly}. The same result was obtained using $\texttt{Landau.jl}$ for the codimension 1 components. We note that for the diagrams \texttt{tdetri}, \texttt{debox}, \texttt{tdebox}, the existence of the codimension two components and their degrees can be verified numerically by adapting the methods from Sec.~\ref{subsec:sampling}. We arrive at the following computational result. 

\begin{theorem}\label{thm2}
For $G \in \{ \mathtt{acn}, \mathtt{par},\mathtt{npltrb}, \mathtt{tdetri}, \mathtt{debox}, \mathtt{tdebox}, \mathtt{pltrb}, \mathtt{dbox} \}$, let $\LLL = \PP^3 \subset \PP(\K_G)$ be the subspace with coordinates $(s:t:\M:\m)$ for which $\M_i = \M$, $\m_e = \m$. We have that $\nabla_{\mathtt{acn}}(\LLL)$ is as in Ex.~\ref{ex:acnode}, and 
\begin{align*}
    \nabla_{\mathtt{par}}(\LLL) &= \{ (\M - \m)(\M^2 - 10\M\m + 9\m^2  + 4\m s) = 0 \} ,\\
    \nabla_{\mathtt{npltrb}}(\LLL) &= \{ \m(s - \M) = 0 \}, \\
    \nabla_{\mathtt{tdetri}}(\LLL) &= \{ s-4\M = s-4 \m = 0 \} \cup \{ 9 \m^2-10 \m \M+\m s+\M^2 = 0\}, \\
    \nabla_{\mathtt{debox}}(\LLL) &=  \{ t-\M = t-\m = 0 \} \\
    &\quad\;\cup \{ 36 \m^2 \M-9 \m^2 s-28 \m \M t+10 \m s t  
     +4 \m t^2+4\M^2 t-s t^2 = 0 \}, \\
    \nabla_{\mathtt{tdebox}}(\LLL) &= \{ t-4 \m = 40 \m+4 \M-11 t = 0\} \\
    &\quad\;\cup \{ 36 \m^2-40 \m \M+16 \m s +4 \m t+4 \M^2-s t = 0 \},\\
    \nabla_{\mathtt{pltrb}}(\LLL) &= \{ s - 3\m = s-\M = 0 \} \cup \{ s- 3 \m = s-3\M = 0 \},\\
    \nabla_{\mathtt{dbox}}(\LLL) &= \{ (4 \m \M-\m s-4 \m t+s t) \Big(144 \m^2 \M^2-72 \m^2 \M s-96 \m^2 \M t+9 \m^2 s^2 \\
&\quad\; +24 \m^2 s t +16 \m^2 t^2-96 \m \M^3+24 \m \M^2 s+16 \m \M^2 t+40 \m \M s t \\
&\quad\; -10 \m s^2 t-8 \m s t^2+16 \M^4-8 \M^2 s t+s^2 t^2\Big) = 0 \}.
\end{align*}
\end{theorem}

\begin{example}
The real section of $\nabla_{\mathtt{dbox}}(\LLL)$ is illustrated in Fig.~\ref{fig:nabla-dbox}. The figure shows that the limit $\lim_{\epsilon \rightarrow 0^+} \nabla_{\mathtt{dbox}}(\LLL) \cap \{ \m = \epsilon \}$ contains the curves $\{s = 0\}$, $\{t=0\}$ and $\{st - 4\M^2 = 0\}$ in the affine $(s,t)$-plane. However, one can check that $\alpha_1 \cdots \alpha_7 ~\U_{\mathtt{dbox}}$ is contained in the ideal generated by $\frac{\partial \F_{\mathtt{dbox}}}{\partial \alpha_e}\big|_{\m = 0}$, $e = 1, \ldots, 7$.
Therefore, the curves $\{\m = s = 0\}$, $\{\m = t = 0\}$ and $\{ \m = st - 4\M^2 = 0\}$ end up in $\nabla_{\mathtt{dbox}}(\LLL)$ by taking the closure in \eqref{eq:nablaG}.
\begin{figure}[!t]
    \centering
    \includegraphics[width=0.6\textwidth]{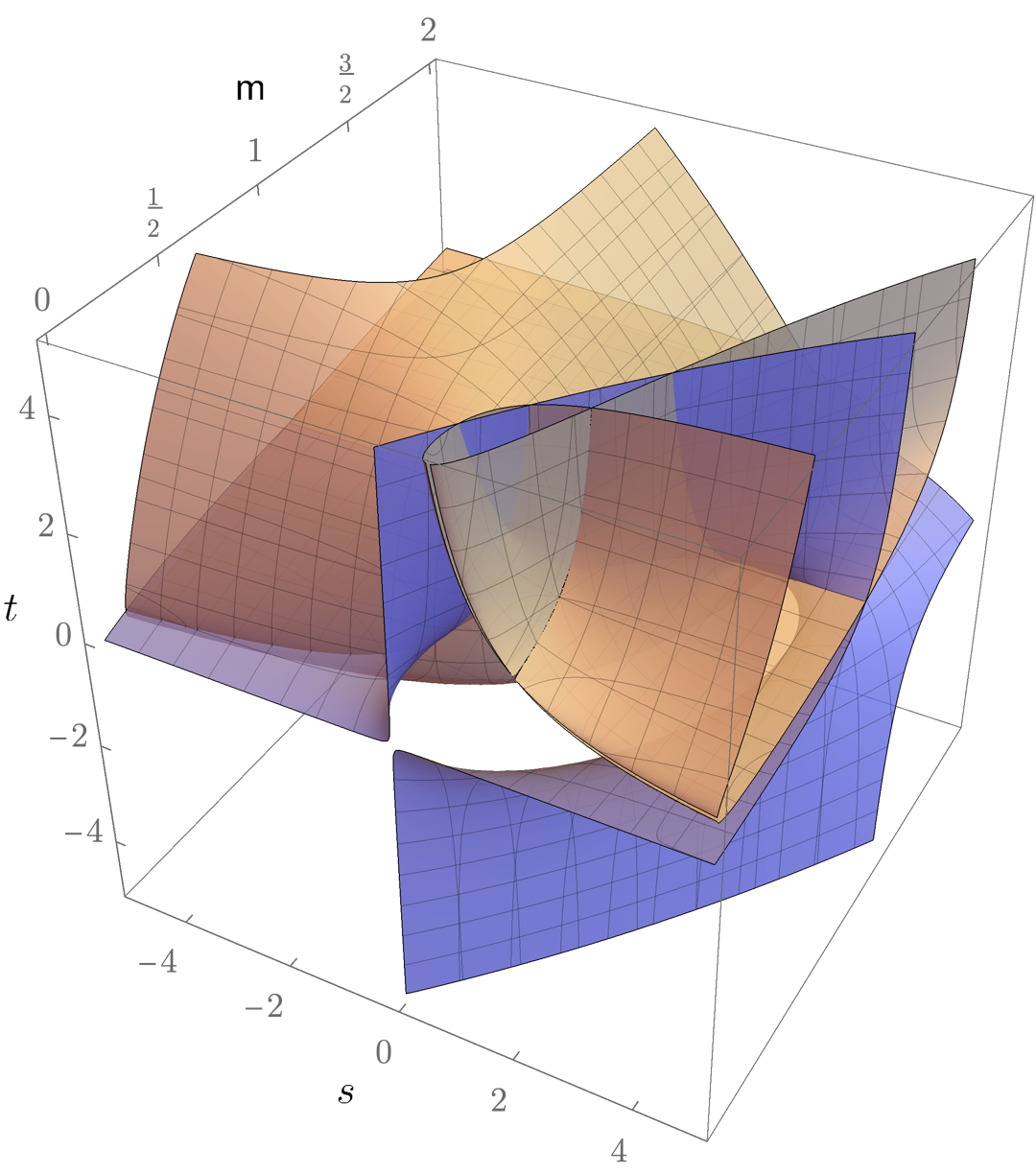}
    \caption{Real part of the Landau discriminant $\nabla_{\mathtt{dbox}}(\LLL)$ for the double-box diagram given in Thm.~\ref{thm2} with dehomogenization $\M=1$. The two components are illustrated in blue (degree $2$) and orange (degree $4$).
    }
    \label{fig:nabla-dbox}
\end{figure}
\end{example}

For the more complicated diagrams $G = \mathtt{env}$ and $G = \mathtt{pentb}$, the discriminant could not be computed symbolically. We discuss our results for these diagrams in the following two examples. 
\begin{example}[Envelope diagram $G = \mathtt{env}$] \label{ex:env}
In \cite{doi:10.1063/1.1664557}, the authors attempted to compute the Landau discriminant for $G = \texttt{env}$, but the results are limited to a numerical plot. Using our numerical sampling method, we find $5$ irreducible components
\be
\Delta_{\mathtt{env}}(\LLL) = \prod_{i=1}^{5} \Delta_{\mathtt{env},i},
\ee
where the first three are degree-$8$ with $83$ terms each:
\begin{dmath}[style={\small}]
\Delta_{\mathtt{env},1} = -16 s \M^7-432 \m^2 \M^6+20 s^2 \M^6+192 \m s \M^6+4 s t \M^6+1728 \m^3 \M^5-8 s^3 \M^5-240 \m
   s^2 \M^5-48 \m^2 s \M^5+216 \m^2 t \M^5-8 s^2 t \M^5-72 \m s t \M^5-2592 \m^4 \M^4+s^4 \M^4+96
   \m s^3 \M^4+492 \m^2 s^2 \M^4-27 \m^2 t^2 \M^4+s^2 t^2 \M^4+6 \m s t^2 \M^4-1280 \m^3 s
   \M^4-864 \m^3 t \M^4+2 s^3 t \M^4+134 \m s^2 t \M^4-84 \m^2 s t \M^4+1728 \m^5 \M^3-12 \m s^4
   \M^3-240 \m^2 s^3 \M^3-128 \m^3 s^2 \M^3+108 \m^3 t^2 \M^3-28 \m s^2 t^2 \M^3+48 \m^2 s t^2
   \M^3+2448 \m^4 s \M^3+1296 \m^4 t \M^3-40 \m s^3 t \M^3-408 \m^2 s^2 t \M^3+1232 \m^3 s t
   \M^3-432 \m^6 \M^2+30 \m^2 s^4 \M^2+224 \m^3 s^3 \M^2+2 \m s^2 t^3 \M^2-6 \m^2 s t^3 \M^2-468
   \m^4 s^2 \M^2-162 \m^4 t^2 \M^2+4 \m s^3 t^2 \M^2+136 \m^2 s^2 t^2 \M^2-468 \m^3 s t^2
   \M^2-1728 \m^5 s \M^2-864 \m^5 t \M^2+2 \m s^4 t \M^2+156 \m^2 s^3 t \M^2+156 \m^3 s^2 t
   \M^2-2052 \m^4 s t \M^2-28 \m^3 s^4 \M-72 \m^4 s^3 \M-20 \m^2 s^2 t^3 \M+76 \m^3 s t^3 \M+432
   \m^5 s^2 \M+108 \m^5 t^2 \M-32 \m^2 s^3 t^2 \M-48 \m^3 s^2 t^2 \M+576 \m^4 s t^2 \M+432 \m^6 s
   \M+216 \m^6 t \M-12 \m^2 s^4 t \M-136 \m^3 s^3 t \M+288 \m^4 s^2 t \M+1080 \m^5 s t \M+9 \m^4
   s^4+\m^2 s^2 t^4-4 \m^3 s t^4+2 \m^2 s^3 t^3+6 \m^3 s^2 t^3-54 \m^4 s t^3-108 \m^6 s^2-27
   \m^6 t^2+\m^2 s^4 t^2+20 \m^3 s^3 t^2-45 \m^4 s^2 t^2-162 \m^5 s t^2+10 \m^3 s^4 t+18 \m^4
   s^3 t-162 \m^5 s^2 t-108 \m^6 s t,
\end{dmath}
while the other two are obtained by relabelling:
\be
\Delta_{\mathtt{env},2} = \Delta_{\mathtt{env},1}|_{s \leftrightarrow t}, \qquad \Delta_{\mathtt{env},3} = \Delta_{\mathtt{env},1}|_{t \rightarrow u}.
\ee
where $u := 4\M {-} s {-} t$.
The remaining components are lengthy when written out in the variables $\M,\m,s,t$. However, noticing that they are permutation invariant with respect to the external legs, we express them in terms of the elementary symmetric functions
\be
\sigma_2 := st + tu + us, \qquad \sigma_3 := stu.
\ee
The degree-$9$ component is 
\be
\Delta_{\mathtt{env},4} = 4 \m^3 \sigma _3^2+64 \m^2 \M^3 (\m-\M)^4-\sigma _3 \left(27 \m^4-2 \m^2 \M^2+8 \m
   \M^3-\M^4\right) (\m-\M)^2.
\ee
Similarly, the last component is degree-$12$ with $81$ terms (versus $347$ before using permutation-invariant variables) and reads
\begin{dmath}[style={\small}]
\Delta_{\mathtt{env},5} = 4096 \M^{12}-65536 \m \M^{11}+475136 \m^2 \M^{10}-2048 \sigma _2 \M^{10}-2064384 \m^3 \M^9+24576 \m \sigma _2 \M^9+2048 \sigma _3 \M^9+5988352 \m^4 \M^8+256 \sigma _2^2 \M^8-147456 \m^2 \sigma _2 \M^8+14336 \m \sigma _3 \M^8-12222464 \m^5 \M^7-2048 \m \sigma _2^2 \M^7+606208 \m^3 \sigma _2 \M^7-313344 \m^2 \sigma _3 \M^7-512 \sigma _2 \sigma _3 \M^7+18006016 \m^6 \M^6+15360 \m^2 \sigma _2^2 \M^6+128 \sigma _3^2 \M^6-1847296 \m^4 \sigma _2 \M^6+1783808 \m^3 \sigma _3 \M^6-8704 \m \sigma _2 \sigma _3 \M^6-19300352 \m^7 \M^5-79872 \m^3 \sigma _2^2 \M^5+2560 \m \sigma _3^2 \M^5+4096000 \m^5 \sigma _2 \M^5-5031936 \m^4 \sigma _3 \M^5+82432 \m^2 \sigma _2 \sigma _3 \M^5+14946304 \m^8 \M^4-1024 \m^2 \sigma _2^3 \M^4+230912 \m^4 \sigma _2^2 \M^4+27136 \m^2 \sigma _3^2 \M^4+32 \sigma _2 \sigma _3^2 \M^4-6348800 \m^6 \sigma _2 \M^4+7350272 \m^5 \sigma _3 \M^4+1280 \m \sigma _2^2 \sigma _3 \M^4-390656 \m^3 \sigma _2 \sigma _3 \M^4-8159232 \m^9 \M^3+4096 \m^3 \sigma _2^3 \M^3-32 \sigma _3^3 \M^3-374784 \m^5 \sigma _2^2 \M^3-438784 \m^3 \sigma _3^2 \M^3-1408 \m \sigma _2 \sigma _3^2 \M^3+6602752 \m^7 \sigma _2 \M^3-4241408 \m^6 \sigma _3 \M^3-4096 \m^2 \sigma _2^2 \sigma _3 \M^3+1171968 \m^4 \sigma _2 \sigma _3 \M^3+2981888 \m^{10} \M^2-6144 \m^4 \sigma _2^3 \M^2+1184 \m \sigma _3^3 \M^2+343040 \m^6 \sigma _2^2 \M^2+1443456 \m^4 \sigma _3^2 \M^2-5440 \m^2 \sigma _2 \sigma _3^2 \M^2-4360192 \m^8 \sigma _2 \M^2-1185792 \m^7 \sigma _3 \M^2+43520 \m^3 \sigma _2^2 \sigma _3 \M^2-2141696 \m^5 \sigma _2 \sigma _3 \M^2-655360 \m^{11} \M+4096 \m^5 \sigma _2^3 \M-7072 \m^2 \sigma _3^3 \M-165888 \m^7 \sigma _2^2 \M-242688 \m^5 \sigma _3^2 \M+67200 \m^3 \sigma _2 \sigma _3^2 \M+1646592 \m^9 \sigma _2 \M+1425408 \m^8 \sigma _3 \M-129024 \m^4 \sigma _2^2 \sigma _3 \M+2366976 \m^6 \sigma _2 \sigma _3 \M+65536 \m^{12}+\sigma _3^4-1024 \m^6 \sigma _2^3+13024 \m^3 \sigma _3^3-48 \m \sigma _2 \sigma _3^3+33024 \m^8 \sigma _2^2-3433728 \m^6 \sigma _3^2+768 \m^2 \sigma _2^2 \sigma _3^2-149472 \m^4 \sigma _2 \sigma _3^2-270336 \m^{10} \sigma _2+458752 \m^9 \sigma _3-4096 \m^3 \sigma _2^3 \sigma _3+137472 \m^5 \sigma _2^2 \sigma _3-1276416 \m^7 \sigma _2 \sigma _3.
\end{dmath}
To obtain the correct rational coefficients of this component, we used higher precision as explained in Rk.~\ref{rem:hp}. The total degree of $\Delta_{\mathtt{env}}(\LLL)$ is $45$.
\end{example}

\begin{example}[The penta-box diagram $G = \mathtt{pentb}$] \label{ex:pentb}
 For $G = \mathtt{pentb}$, we consider the subspace $\LLL = \PP^6 \subset \PP(\K_{\mathtt{pentb}})$ with coordinates $(s_{12}:s_{23}:s_{34}:s_{45}:s_{51}:\M:\m)$ given by $\M_i = \M, \m_e = \m$. We used the iterative eigenvalue technique from Rk.~\ref{rem:iter} to compute the kernel vector of a $22280 \times 18564$ matrix. This gives a homogeneous polynomial of degree 12 in the 7 parameters with 2601 terms. The discriminant can be found at \url{https://mathrepo.mis.mpg.de/Landau/}.
\end{example}

\begin{remark}[Acnode diagram]\label{ex:acnode2}
The diagram $\mathtt{acn}$ was previously studied in \cite{doi:10.1063/1.1703752,doi:10.1063/1.1664557} with the specific assignment of masses (in the notation of Fig.~\ref{fig:acn}):
\be
\M_1 = \M_3 = M^2, \qquad \M_2 = \M_4 = m^2, \qquad \m_e = 1.
\ee
This defines a linear subspace $\LLL \subset \PP(\K_{\mathtt{acn}})$. Using either of the techniques presented in the above sections we find that the Landau discriminant has two irreducible components: $\Delta_{\mathtt{acn}}(\LLL) = \Delta_{\mathtt{acn},1} \cdot \Delta_{\mathtt{acn},2}$.
The first factor has $184$ terms and can be found at \url{https://mathrepo.mis.mpg.de/Landau/}.
While it has not appeared in the literature, Ref.~\cite[Eq. (4)]{doi:10.1063/1.1703752} provided its parametrization in the chart where $m=1$:
\begin{align}
s &= 5 + 4 \cos \phi + 2 \left(2 - \tfrac{1}{2}M^2 + \cos \theta + \cos \phi \right) \sin \phi / \sin \theta,\\
t &= 5 + 4 \cos \theta + 2 \left(2 - \tfrac{1}{2} M^2 + \cos \theta + \cos \phi \right) \sin \theta / \sin \phi,
\end{align}
with $\theta + \phi = \pi/3$. We checked that this indeed correctly parametrizes $\Delta_{\mathtt{acn},1}$. For a range of masses satisfying $M^2 \geq 4 + 2\sqrt{2}$, this curve develops cusps in the real $(s,t)$-space on the physical sheet known as \emph{acnodes} and \emph{crunodes} \cite{doi:10.1063/1.1703752}, which provided an explicit counterexample to the validity of the Mandelstam representation.

The remaining component of the Landau discriminant is given by
\begin{align}
\Delta_{\texttt{acn},2} = &\left( u - m^2(M^2 {-} 1)\right)\left((s{-}1)(t{-}1) + (M^2 {-}1)(1{+}2m^2 {-}M^2) \right)+ m^4(m^2 {-} 4)(M^2 {-} 1),
\end{align}
where $u := 2m^2 + 2M^2 - s- t$, in agreement with the result quoted in \cite[Eq. (9)]{doi:10.1063/1.1664557}. It was argued in Ref.~\cite{doi:10.1063/1.1664557} that this component never lies on the physical sheet.

One can check that the above results evaluated at $m=M$ match those of the Ex.~\ref{ex:acnode} evaluated at $\M = M^2$, $\m=1$.
\end{remark}

\subsection{\label{sec:CN-analysis}Coleman--Norton Analysis of the Envelope Diagram}

Recall that the Feynman integral \eqref{eq:Feynman} is in general a multi-valued function on the kinematic space $\K_G$.
The physically-relevant branch (consistent with causality) is defined by the $i\epsilon$ prescription in \eqref{eq:Feynman} within the \emph{physical regions} ${\cal P}_G$, given by a union of disconnected subsets of $\mathbb{RP}(\K_G)$ corresponding to the energies of the external momenta $p_i$ being real; see, e.g., \cite{PhysRev.117.1159}. For instance, when $n_G = 4$ the physical regions are given by
\be\label{eq:PG}
{\cal P}_G = \R \PP (\K_G) \cap \{ \det \big(p_i {\cdot} p_j \big)_{i,j=1,2,3} > 0\}.
\ee
This will be illustrated concretely in Fig~\ref{fig:ResRet}. The Landau equations give necessary but not sufficient conditions for singularities of Feynman integrals. That is, not all points on the Landau discriminant affect the numerical evaluation of \eqref{eq:Feynman}. The purpose of this subsection is to qualitatively identify the parts that do.

For a Feynman diagram $G$, let $\LLL \subset \PP(\K_G)$ be a nonempty subvariety and let $\nabla_G(\LLL) \subset \LLL$ be the associated Landau discriminant. 
It is physically meaningful to ask whether a point on $\nabla_G(\LLL)$ leads to singular points of the hypersurface $\{ \F_G = 0 \}$ on the projectivized integration domain
\be
\R \PP^{\E_G -1}_+ = \{ (\alpha_1: \cdots : \alpha_{\E_G}) \in \R\PP^{\E_G -1} ~|~ \alpha_e > 0 \text{ for all $e$ } \} \subset X.
\ee
This motivates the following definition, in which we use our previous notation $\pi_\LLL: Y \rightarrow \LLL$ for the projection of the incidence variety $Y \subset X \times \LLL$ to $\LLL$ (Rk.~\ref{rem:restriction}).
\begin{definition}[$\alpha$-positive point]
A point $q \in \nabla_G(\LLL)$ is called \emph{$\alpha$-positive} if $
\pi_{\LLL}^{-1}(q) \cap (\R \PP^{\E_G -1}_+ \times \{q \} )$ is nonempty.
\end{definition}
Note that the above definition applies to any $q \in \nabla_G(\LLL)$, but the physical interpretation is more subtle.
The significance of an $\alpha$-positive singularity in a physical region, $q \in \nabla_G(\LLL) \cap {\cal P}_G$, was explained by Coleman and Norton \cite{Coleman:1965xm}: it represents kinematics for which the internal particles of the Feynman diagram propagate along their classical trajectories, where each Schwinger parameter $\alpha_e$ is real and (in the massive case, $\m_e > 0$) proportional to the proper time elapsed between pairs of vertices of the Feynman diagram $G$. Since the value of a Feynman integral away from ${\cal P}_G$ needs to be defined by analytic continuation, it is in general much more difficult to determine if a given point on $\nabla_G(\LLL)$ is a singularity of the integral on the appropriate sheet.

\begin{example}
For the family of banana diagrams $\mathtt{B}_{\E}$, the explicit solution of the Landau discriminant was given in Sec.~\ref{sec:banana}. Using \eqref{eq:BE-alphas}, the only $\alpha$-positive component is given by $\eta_e = 1$ for all $e$, corresponding to
\be\label{eq:BE-s}
s = ({\textstyle\sum}_{e=1}^{\E}m_e)^2
\ee
known as the normal threshold. The remaining $2^{\E-1}-1$ components are singularities on sheets in the $s$-plane accessible by analytic continuation through the branch cut extending from \eqref{eq:BE-s} to $s=+\infty$. 
\end{example}

Below, we will briefly illustrate a qualitative Coleman--Norton analysis on the example of the envelope diagram, $G = \mathtt{env}$, with generic masses. We focus on its leading singularities only. We point out that the presented techniques can be applied to other diagrams in a straightforward way. 

\begin{proposition}
Consider the affine plane $\LLL = \C^2 \subset \PP(\K_{\mathtt{env}})$ obtained by setting all $\M_i, \m_e$ to fixed real values, with coordinates $(s,t)$. All the $\alpha$-positive points on the Landau discriminant $\nabla_{\mathtt{env}}(\LLL)$ are real.
\end{proposition}
\begin{proof}
We need to show that if the Landau equations (restricted to $\LLL$) have a positive solution, i.e., a solution $\alpha \in \R \PP_+^{5}$, we must have $\Im(s) = 0$ and $\Im(t) = 0$. Viewing the second Symanzik polynomial $\F_{\mathtt{env}}$ as a polynomial in $s,t$ with real coefficients parametrized by polynomials in the Schwinger parameters $\alpha_e$, its imaginary part equals
\be
\Im (\F_{\mathtt{env}}) = \Im (s)\, \alpha_1 \alpha_3 \left( \alpha_5 \alpha_6 - \alpha_2 \alpha_4 \right) + \Im (t)\, \alpha_2 \alpha_4 \left( \alpha_5 \alpha_6 - \alpha_1 \alpha_3 \right).
\ee
Suppose $(\alpha_1: \ldots: \alpha_6)$ is a positive solution to the Landau equations, so that it satisfies $\Im \frac{\partial \F_{\mathtt{env}}}{\partial \alpha_e} = 0$. One checks easily that this implies $\Im(s) = \Im(t) = 0$.
\end{proof}
The above result implies that, for finding $\alpha$-positive points with fixed masses, it is sufficient to investigate the real part $\R \nabla_{\mathtt{env}}(\LLL)$ of the Landau discriminant in $\R \LLL = \R^2$ with coordinates $(\Re(s), \Re(t))$. Our analysis will be qualitative, in the sense that we investigate the $\alpha$-positive points by means of a plot. 

The approach we present is numerical.
It bypasses the interpolation stage described in Sec.~\ref{subsec:sampling}. We argue that using these techniques, plotting $\nabla_G$ can be achieved for complicated diagrams for which it is not feasible to compute $\Delta_G$ symbolically. 

The strategy is to sample $\R \nabla_G(\LLL)$ (or its higher-dimensional analogues for $n_G > 4$) by intersecting it with a family of parallel lines in $\R \LLL$. Using sufficiently many of such lines, we will have sampled the discriminant densely enough for a detailed plot. In \texttt{Julia}, augmenting the Landau equations \texttt{LE} with such a pencil of lines can be done via 
\begin{minted}{julia}
@var e
line = randn()*s + randn()*t + e
LE_line = System([LE;line], parameters = [e])
\end{minted}
The parameter \texttt{e} of the family of lines represents the offset from the origin. We define an array \texttt{targetpars} of values of \texttt{e} for which we want to compute the intersections. Suppose we want to plot the Landau discriminant inside a bounded real box $B \subset \R \LLL = \R^2$. With the help of the basic auxiliary functions \texttt{filter_in_box}, which returns all points $y \in Y$ for which $\pi_{\LLL}(y) \in B$, and $\texttt{s_t_coordinates}$, which returns the projection $\pi_{\LLL}(y)$ of these points, the sampling can be done as follows.
\begin{minted}[firstnumber=last]{julia}
samples = []
for ee in targetpars
    R = solve(LE_line; target_parameters = [ee])
    sols = filter_in_box(solutions(R))
    samples_in_box = s_t_coordinates(sols)
    samples = push!(samples, samples_in_box...)
end
\end{minted}
We apply this to the Landau discriminant of the envelope diagram $G=\mathtt{env}$ evaluated at generic masses. As listed in Tab.~\ref{tab:dimdeg}, it is irreducible with degree $114$ and hence would be too impractical to compute symbolically. Instead, using the above algorithm we plot $\R \nabla_{\mathtt{env}}(\LLL)$, where $\LLL = \C^2$ is defined by the masses
\begin{align}\label{eq:env-masses1}
(\M_1, \M_2, \M_3, \M_4) &= (4,5,6,7),\\
\qquad (\m_1, \m_2, \m_3, \m_4, \m_5, \m_6) &= (\sfrac{1}{4}, \sfrac{1}{5}, \sfrac{1}{6}, \sfrac{1}{7}, \sfrac{10}{8}, \sfrac{10}{9}).
\end{align}
\begin{figure}[t]
    \centering
    \includegraphics[width=\textwidth]{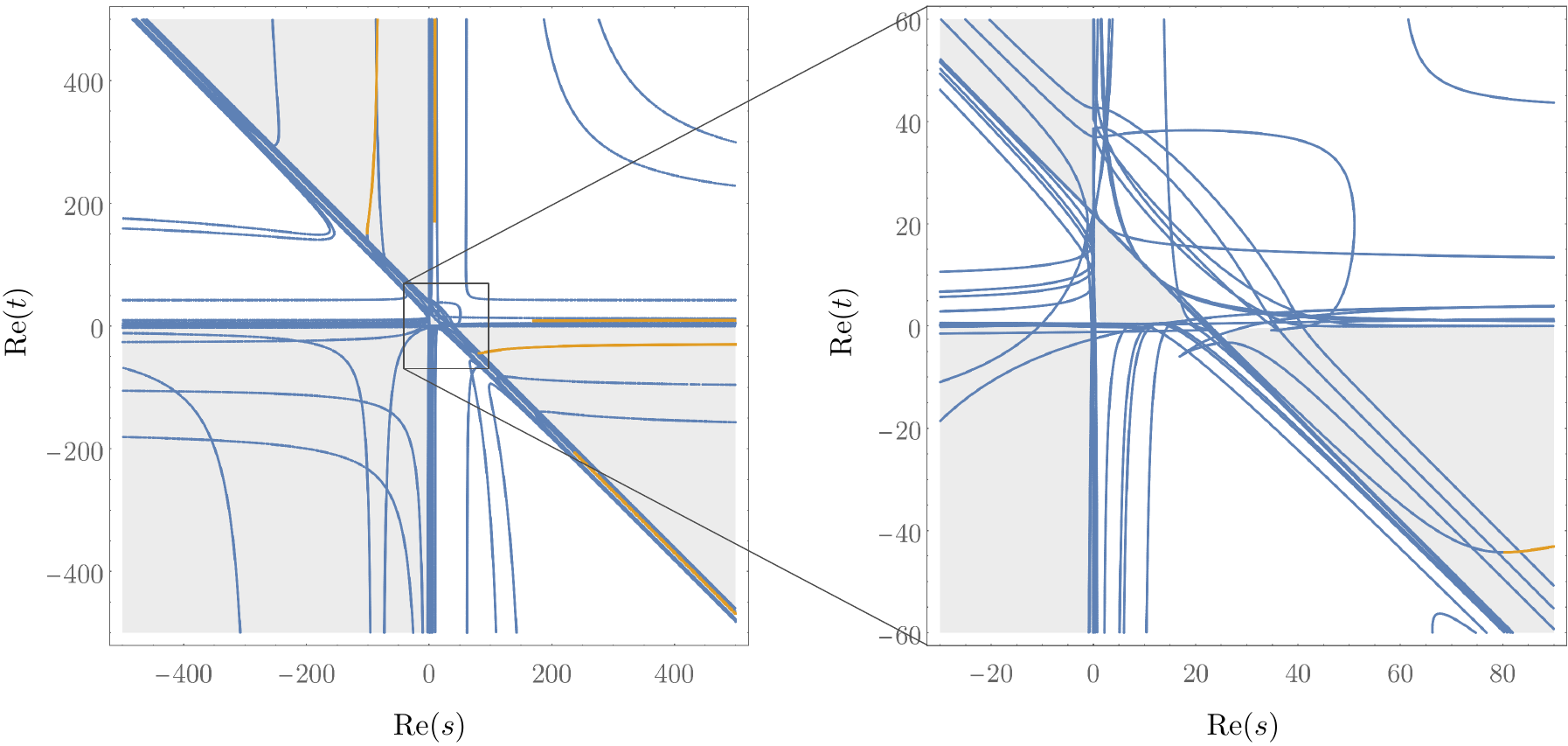}
    \caption{Leading Landau singularities for $G = \mathtt{env}$ with masses (\ref{eq:env-masses1}) at two different magnifications, featuring $\alpha$-positive (orange) and non-$\alpha$-positive (blue) curves. Physical regions ${\cal P}_G$ have $4$ disconnected components (gray). There are $3$ $\alpha$-positive curve segments intersecting ${\cal P}_G$, also known as physical-sheet singularities.}
    \label{fig:ResRet}
\end{figure}
The result is displayed in Fig.~\ref{fig:ResRet}. It features $5$ disconnected $\alpha$-positive parts (orange), among multiple other curve segments that do no satisfy the $\alpha$-positive criteria (blue). In this example, the physical regions ${\cal P}_G$ (gray) are $4$ disconnected regions carved out by the inequality in \eqref{eq:PG} translating to
\be
-s^2t-s t^2+22 s t-3 s+t-22 > 0,
\ee
for the choice of masses \eqref{eq:env-masses1}.
Within the confines of the plot, only $3$ out of the $\alpha$-positive curve segments intersect a physical region, including the one in the bottom-right corner of the left panel in Fig.~\ref{fig:ResRet}. In the physics language, those $3$ curve segments are said to lie on the \emph{physical sheet}. The right panel illustrates how one of the curves stops being $\alpha$-positive around $(s,t) \approx (80.4, -44.3)$, which is where it leaves the physical sheet. We found that at this point, the curve intersects one of the $\alpha$-positive branches of the subleading Landau singularity of $G$ obtained by shrinking the edge $1$.

One can easily include subleading Landau singularities of $\mathtt{env}$, as listed in Ex.~\ref{ex:subleading}, which are all straight lines (constant $s$, $t$, or $u = \sum_{i=1}^{4} \M_i {-} s {-} t$). They also feature many $\alpha$-positive components, but we do not include them here to avoid cluttering Fig.~\ref{fig:ResRet} further.

We note that anomalous thresholds have effects visible at the energy scales a couple of orders of magnitude larger than the mass scales involved in this problem.

%% file: section4.tex
This section is concerned with polytopes associated to Feynman diagrams and Landau equations. As explained in Sec.~\ref{subsec:motiv}, what links such polytopes to Landau discriminants is the theory of ${\cal A}$-discriminants and toric resultants. The faces of these polytopes inherit factorization properties of Symanzik polynomials. These factorization properties are listed in Sec.~\ref{sec:factorization}, and translated to the polytope setting in Sec.~\ref{subsec:faces}. In Sec.~\ref{sec:banana-polytopes}, we apply the results to the banana diagrams $\mathtt{B}_\E$ (see Fig.~\ref{fig:diagrams}). Finally, in Sec.~\ref{subsec:degree} we deduce bounds on the degree of the Landau discriminant from the theory of ${\cal A}$-discriminants and toric resultants.  

\subsection{Motivation} \label{subsec:motiv}

Convex lattice polytopes naturally arise in the context of toric geometry. They encode the stratification of normal, compact toric varieties into torus orbits. For background, see \cite{cox2011toric}. Such toric varieties are often considered as solution spaces for families of systems of equations with a fixed monomial support, see for instance \cite[Ch.~5]{telen2020thesis}. In particular, they provide a natural setting for studying \emph{discriminants} and \emph{resultants}.

\emph{Discriminants} are classical objects in algebraic geometry. The class of \emph{${\cal A}$-discriminants} was extensively studied in the pioneering work of Gelfand, Kapranov, and Zelevinsky \cite{gelfand2008discriminants}. The symbol ${\cal A}$ typically denotes a subset of lattice points in $\Z^n$, corresponding to the exponent vectors $a$ occurring in a Laurent polynomial $f = \sum_{a \in {\cal A}} c_a x^{a}$. The \emph{${\cal A}$-discriminant} is the Zariski closure of the set of all $(c_a) \in \C^{|{\cal A}|}$ for which the hypersurface $\{f = 0 \}$ has singularities in the algebraic torus $(\C^*)^n$. It naturally lives in the projective space $\PP(\C^{|{\cal A}|}) = \PP^{|{\cal A}|-1}$, and can be identified with the dual of the projective toric variety $X_{{\cal A}} \subset (\PP^{|{\cal A}|-1})^{\vee}$. Properties such as the degree of the ${\cal A}$-discriminant can be understood from the combinatorics of the convex polytope $\PPP = \textup{Conv}(\A) \subset \R^n$. The associated toric variety $X_{\PPP}$ is the normalization of $X_{{\cal A}}$. We would like to apply some standard results in the theory of ${\cal A}$-discriminants to our study of Landau discriminants. The role of the polynomial $f$ will be played by the second Symanzik polynomial $\F_G$ (with $n = \E_G$). Its monomials in the Schwinger parameters $\alpha_1, \ldots, \alpha_{\E_G}$ give the set ${\cal A}$, and its coefficients are parametrized linearly by the kinematic space $\PP(\K_G) \subset \PP^{|{\cal A}|-1}$. This justifies the study of the corresponding \emph{Symanzik polytope}, i.e.,~the \emph{Newton polytope} $\FF_G$ of $\F_G$, and also the Newton polytope of $\U_G$ will make a natural appearance. We warn the reader that, apart from restricting to the kinematic space or even smaller subspaces $\LLL$ of $\PP^{|{\cal A}|-1}$, an important difference with ${\cal A}$-discriminant analysis is that we discard singularities of $\{\F_G = 0\}$ inside the locus $\{\U_G = 0\}$. Related recent applications of $\cal A$-discriminants to holonomic systems for Feynman integrals include \cite{delaCruz:2019skx,Klausen:2019hrg,Feng:2019bdx,Tellander:2021xdz}.

\emph{Resultants} are closely related to discriminants. They encode parameter values for which overdetermined systems of equations admit solutions. A well-understood class of resultants is that of \emph{mixed $({\cal A}_0, \ldots, {\cal A}_{n})$-resultants}. The set $\A_i \subset \Z^n$ encodes monomials occurring in a Laurent polynomial $f_i = \sum_{\alpha \in \A_i} c_{i,\alpha} x^\alpha$. The $({\cal A}_0, \ldots, {\cal A}_{n})$-resultant polynomial $\textup{Res}_{{\cal A}_\bullet}$ is the unique polynomial (up to scaling) vanishing on all points $((c_{0,\alpha}), \ldots, (c_{n,\alpha})) \in \C^{|{\cal A}_0|} \times \cdots \times \C^{|{\cal A}_{n}|}$ for which $f_0 = \cdots = f_n = 0$ has a solution in $(\C^*)^n$, under the assumption that this set has codimension 1. Let $\PPP_i = \Conv({\cal A}_i) \subset \R^n$ and let $\LL = \PPP_0 \oplus \cdots \oplus \PPP_n$ be the Minkowski sum (defined below \eqref{eq:FUDelta}). One can show that each of the $f_i$ defines a Cartier divisor $V_{X_\LL}(f_i)$ on the normal toric variety $X_\LL$, and that $((c_{0,\alpha}), \ldots, (c_{n,\alpha})) \in V(\textup{Res}_{{\cal A}_\bullet})$ if and only if $V_{X_\LL}(f_0) \cap \cdots \cap V_{X_\LL}(f_n) \neq \emptyset$. This means that the resultant characterizes precisely when $f_0 = \cdots = f_n = 0$ has a solution \emph{in the toric variety $X_{\LL}$}. The construction automatically takes into account solutions on the boundary $X_{\LL} \setminus (\C^*)^n$ of the torus in this toric compactification. Recent efforts in numerical algebraic geometry have lead to methods for detecting and computing such solutions on the boundary \cite{duff2020polyhedral,bender2020toric,telen2020numerical}. Here, we will take the first steps in studying the polytopes $\LL$ associated to this compactification in the setting of Landau discriminants. We set $n = \E_G - 1$ and the polynomials $f_i$ are $\frac{\partial \F_G}{\partial \alpha_{i+1}}$ for $i = 0, \ldots, \E_G-1$ (here the essential number of variables is indeed $\E_G-1$, after dehomogenizing $\alpha_{\E_G} = 1$). The \emph{Landau polytope} $\LL_G$ associated to $G$ will be defined as the Minkowski sum of the polytopes $\Newt(f_i)$. Again, for our purposes, the coefficients $((c_{0,\alpha}), \ldots, (c_{n,\alpha}))$ are parametrized linearly by the kinematic space, and we ignore solutions in the hypersurface given by $\{ \U_G = 0 \}$.

In the above discussion, we have presented two, generally \emph{different} toric compactifications of the variety $X$ from \eqref{eq:defX}: $X \subset X_{\FF_G}$ and $X \subset X_{\LL_G}$. The former comes from interpreting the Landau discriminant analysis from an ${\cal A}$-discriminant point of view, while the latter comes from computing toric resultants. The boundaries $X_{\FF_G} \setminus X$ and $X_{\LL_G} \setminus X$ consist of the usual toric exceptional divisors, together with the hypersurface defined by $\U_G$ in the torus. The families of polytopes coming from these constructions are interesting combinatorial objecs in their own right. We will see in the next subsections, for instance, that their faces have factorization properties stemming from similar properties of the Symanzik polynomials. 

\subsection{\label{sec:factorization}Factorization Properties of Symanzik Polynomials}

Let $G$ be a Feynman diagram and let $\gamma \subset G$ be a connected subdiagram. We denote by $\U_G\big|_{\alpha_\gamma \rightarrow \epsilon \alpha_\gamma}$ and $\F_G\big|_{\alpha_\gamma \rightarrow \epsilon \alpha_\gamma}$ the result of replacing $\alpha_e$ by $\epsilon \alpha_e$ for all edges $e \in \gamma$ in $\U_G$ and $\F_G$ respectively. The symbol $\L_\gamma$ denotes the number of loops in $\gamma$ and the contraction $G/\gamma$ is obtained from $G$ by shrinking all the edges and vertices in $\gamma$ to a point. We start by recalling the following well-known result.
\begin{proposition}\label{prop:prop3}
For any connected subdiagram $\gamma \subset G$, we have
\begin{gather}
\U_G\big|_{\alpha_\gamma \rightarrow \epsilon \alpha_\gamma} = \epsilon^{\L_\gamma}\, \U_\gamma\, \U_{G/\gamma} + {\cal O}(\epsilon^{\L_\gamma + 1}),\label{eq:UGe}\\
\F_{G}\big|_{\alpha_\gamma \rightarrow \epsilon \alpha_\gamma} = \epsilon^{\L_\gamma}\, \U_{\gamma}\, \F_{G/\gamma} + {\cal O}(\epsilon^{\L_\gamma + 1}).\label{eq:FGe}
\end{gather}
\end{proposition}
\begin{proof}
The proof is standard; see, e.g., \cite[Thm.~5.1]{10.1143/PTPS.18.1} or \cite[Prop.~2.2, Thm.~2.7]{Brown:2015fyf}. We include a sketch for completeness. In the case of \eqref{eq:UGe} one needs to use the fact that each spanning tree $T$ in $G$ induces a spanning $k$-tree $T \cap \gamma$ in $\gamma$, where $k \in \{1,2,\ldots,\L_\gamma + 1\}$. At the leading order in $\epsilon \to 0$ only the spanning trees of $\gamma$ ($k=1$) contribute. In such cases also $T \cap (G/\gamma)$ must be a spanning tree in $G/\gamma$, from which \eqref{eq:UGe} follows. The derivation of \eqref{eq:FGe} uses the same arguments, except for the fact that a spanning $2$-tree in $G$, which induces a spanning $1$-tree in $\gamma$, must be a spanning $2$-tree in $G/\gamma$.
\end{proof}

While in the case of \eqref{eq:UGe} the leading order coefficient $\U_\gamma \U_{G/\gamma}$ is always non-zero on $\PP^{\E_G -1}_{>0}$, for \eqref{eq:FGe} it might happen that $\F_{G/\gamma} = 0$ identically. For simplicity, we will from now on assume that all the internal edges in $G$ are massive, i.e., $\m_e \neq 0$, which guarantees that $\F_{G/\gamma} \neq 0$ and makes the discussion less cluttered. 

\begin{remark}
Simple examples illustrate that factorization akin to \eqref{eq:FGe} does not persist at subleading orders in $\epsilon$ once the massive condition is relaxed \cite{AHHM}, though some factorization at the order $\epsilon^{\L_\gamma + 1}$ can happen under certain more restrictive kinematic conditions \cite[Thm.~2.7]{Brown:2015fyf}.
\end{remark}

In the following subsections we will also need the following exact version of Prop.~\ref{prop:prop3} in the special case where $\gamma$ is a single edge.
\begin{proposition}
For any edge $e$ of $G$, the Symanzik polynomials satisfy
\begin{gather}
\U_G = \U_{G/e} + \alpha_e \U_{G\setminus e},\label{eq:prop1-1}\\
\F_G = \F_{G/e} + \alpha_e (\F_{G \setminus e} - \m_e \U_G),\label{eq:prop1-2}
\end{gather}
\end{proposition}
\begin{proof}
The proof is standard, see, e.g., \cite[Lem.~1.9]{Brown:2015fyf}. We include a sketch for completeness. Because ${\cal U}_G$ is at most linear in each $\alpha_e$, we have
\be
\U_G = \U_G \big|_{\alpha_e = 0} + \alpha_e \frac{\partial \U_G}{\partial \alpha_e} 
= \sum_{\substack{T \in {\cal T}_G\\ T \ni e}} \prod_{e' \notin T} \alpha_{e'} \;+\; \alpha_e\!\! \sum_{\substack{T \in {\cal T}_G\\ T \not\ni e}} \prod_{\substack{e' \notin T\\ e' \neq e}} \alpha_{e'}.
\ee
Since the first sum is over only the spanning trees that contain $e$, it can be identified with $\U_{G/e}$. Similarly, the second sum contains precisely those spanning trees that do not contain the edge $e$, and hence it equals $\U_{G\setminus e}$, from which \eqref{eq:prop1-1} follows.

Using similar manipulations and definition \eqref{eq:FGS}, we see that for any subset $S \subset \{ 1, \ldots, n_G \}$ we have
\be
\F_{G,S} = \F_{G/e,S} + \alpha_e \F_{G\setminus e,S},
\ee
which together with \eqref{eq:prop1-1} allows us to write
\small
\be
\F_G = \sum_{\{S,\bar{S}\} \in {\cal P}_G} \!\!\!(\textstyle\sum_{i \in S} p_i)^2 (\F_{G/e,S} +\alpha_e \F_{G\setminus e,S})
- \big( \m_e \alpha_e +
\sum_{e' \neq e} \m_{e'} \alpha_{e'}\big) (\U_{G/e} + \alpha_e \U_{G\setminus e}),
\ee
\normalsize
according to \eqref{eq:FGe}.
We then recognize the polynomials $\F_{G/e}$ and $\F_{G\setminus e}$, together with the remaining terms proportional to $\m_e$,
which gives \eqref{eq:prop1-2}.
\end{proof}

The above proposition will be important in understanding the facet structure of Landau polytopes. As a side note, let us point out that it can be used to reformulate the Landau equations \eqref{eq:LE} in terms of Symanzik polynomials of the diagram $G\setminus e$ for any edge $e$.
Using \eqref{eq:prop1-1} and \eqref{eq:prop1-2} together with the fact that $\F_{G/e}$, $\U_{G/e}$, and $\U_{G\setminus e}$ are independent of $\alpha_e$, we find
\be
\frac{\partial \F_G}{\partial \alpha_e} = \F_{G \setminus e} - \m_e (\alpha_e\U_{G\setminus e} + \U_{G})  \quad \text{for $e=1,2,\ldots, \E_G$.}\label{eq:prop2-3}
\ee

\subsection{Facets of Symanzik and Landau Polytopes} \label{subsec:faces}

We can construct different types of polytopes based on the Symanzik polynomials. For a polynomial $P = \sum_{e \in \mathbb{N}^{\E_G}} c_e \alpha^e \in \C[\alpha_1, \ldots, \alpha_{\E_G}]$, let $\Newt(P) = \Conv(e \in \mathbb{N}^{\E_G} ~|~ c_e \neq 0) \subset \R^{\E_G}$ be its \emph{Newton polytope}.

\begin{definition}[Symanzik polytopes]\label{def:UF}
For a given diagram $G$ we define the \emph{Symanzik polytopes}
\be\label{eq:UF-polytopes}
\UU_G := \Newt(\U_G) \subset \R^{\E_G}, \qquad \FF_G := \Newt(\F_G) \subset \R^{\E_G},
\ee
where $\F_G$ is viewed as a polynomial in the Schwinger parameters after plugging in generic kinematic parameters. 
\end{definition}
\begin{remark}\label{rk:UF}
For the kinematic parameters to be ``generic'' in Definition \ref{def:UF}, it is necessary that all the internal and external masses are non-zero: $\m_e \neq 0$, $\M_i \neq 0$. In combination with generic Mandelstam invariants, this ensures that all possible monomials occurring in $\F_G$ have nonzero coefficients. One can study $\Newt(\F_G)$ for more degenerate kinematics, but this is beyond the scope of this paper.
\end{remark}
The polytope $\Newt(\U_G\F_G)$ first appeared in the work of Schultka \cite{Schultka:2018nrs} as a realization of the iterated blow-up procedure for the integration region of Feynman integrals introduced by Bloch, Esnault, Kreimer \cite{Bloch:2005bh}, and Brown \cite{Brown:2015fyf} (see also \cite{Pak:2010pt,Ananthanarayan:2018tog,Semenova:2018cwy,Panzer:2019yxl,Borinsky:2020rqs,Tellander:2021xdz} for related work). A systematic study of the polytopes \eqref{eq:UF-polytopes} was undertaken in \cite{AHHM}, where a specific Minkowski sum/difference of $\UU_G$ and $\FF_G$ was needed to study ultraviolet and infrared properties of Feynman integrals.

As a stepping stone to more complicated combinatorics, we first informally review some of the properties of $\UU_G$ and $\FF_G$ \cite{Schultka:2018nrs,AHHM}. 

For $d \in \R$, let $H_d$ be the hyperplane $H_{d} = \{ x \in \R^{\E_G} ~|~ x_1 + \cdots + x_{\E_G} = d \} \subset \R^{\E_G}$. In general, by homogeneity, the polytopes $\UU_G$ and $\FF_G$ are contained in the hyperplanes $H_{\L_G}$ and $H_{\L_G + 1}$ respectively. 

A connected graph $G$ is called \emph{one-vertex irreducible} (1VI) if it cannot be edge-disconnected by removing a single vertex. Note that a 1VI diagram that does not have a loop is necessarily a single edge. 
Note that all diagrams in Fig.~\ref{fig:diagrams} are 1VI. For a 1VI diagram $G$, we have
\be
\dim \UU_G = \dim \FF_G = \E_G - 1.
\ee
\vspace{-2em}
\begin{assumption}
In the rest of this section, we assume that $G$ is 1VI. 
\end{assumption}
We now investigate the facets (i.e., the $(\E_G {-}2)$-dimensional faces) of $\U_G$ and $\F_G$. Let $\langle \cdot , \cdot \rangle$ denote the usual pairing between $(\R^n)^\vee$ and $\R^n$. For a polytope $\mathbf{P} \subset \R^n$ and for each vector $v \in (\R^n)^\vee$, we denote the corresponding face of $\mathbf{P}$ by
\be
\v{\mathbf{P}} := \{ x \in \mathbf{P} ~|~ \langle v, x \rangle = \min_{y \in \mathbf{P}}  \langle v, y \rangle \}.
\ee
Note that $\partial_0 \mathbf{P} = \mathbf{P}$. As a consequence of Prop.~\ref{prop:prop3}, together with the fact that $\UU_G$ and $\FF_G$ are \emph{generalized permutohedra} \cite{Schultka:2018nrs}, all facets of $\UU_G \subset \R^{\E_G}$ are of the form $\UU_\gamma \times \UU_{G/\gamma} \subset \R^{\E_\gamma} \times \R^{\E_{G/\gamma}}  = \R^{\E_G}$, where $\gamma$ and $G/\gamma$ are both 1VI. For $\gamma$ such that $\gamma$ and $G/\gamma$ are both 1VI, we denote the corresponding facet of $\UU_G$ by
\be
\gam{\UU_G} = \partial_{w_{\gamma}} \UU_G = \UU_\gamma \times \UU_{G/\gamma},
\ee
where $w_{\gamma} = \sum_{e \in \gamma} w_e$ and $w_e$ is the $e$-th standard basis vector of $\R^{\E_G}$.

Similarly, for the polytope $\FF_G$, the facets are $\partial_\gamma \FF_G = \UU_\gamma \times \FF_{G/\gamma}$ for each 1VI subdiagram $\gamma$. Here the 1VI condition on $G/\gamma$ is not needed since it is already guaranteed by requiring that $G$ is 1VI.
In particular, this means that $\FF_G$ usually has more facets than $\UU_G$. Examples will be given in later sections. As explained in Rk.~\ref{rk:UF}, the cases where some of the masses are zero require more careful analysis, see \cite{Brown:2015fyf,AHHM}. For an explicit facet presentation of $\UU_G$ and $\FF_G$ under some generic kinematics conditions see \cite[Thm.~4.15]{Schultka:2018nrs} and \cite[Thm.~32]{Borinsky:2020rqs}. 

Since all monomials of $\F_{G}$ in \eqref{eq:FG} appear also in $(\sum_{e=1}^{\E_G} \alpha_e) \U_G$, we have
\be\label{eq:FUDelta}
\FF_G = \UU_G \oplus \mathbf{\Delta}_{\E_G -1},
\ee
where $\mathbf{\Delta}_{\E_G -1} = \Conv(w_1, \ldots, w_{\E_G})$ is the $(\E_G {-} 1)$-dimensional simplex and $\oplus$ denotes Minkowski addition. Recall that for two polytopes $\PPP_1, \PPP_2 \subset \R^n$, $\PPP_1 \oplus \PPP_2= \{ p_1 + p_2 ~|~ p_1 \in \PPP_1, p_2 \in \PPP_2 \} \subset \R^n$.

Next to the Symanzik polytopes $\UU_G$ and $\FF_G$, we introduce the class of \emph{Landau polytopes}, obtained directly from the critical point equations in  \eqref{eq:LE}.
\begin{definition}[Landau polytopes]
For a graph $G$ and an edge $e \in G$ we define
\be
\LL_{G,e} := \mathrm{Newt} \left( \frac{\partial \F_G}{\partial \alpha_e} \right),
\ee
where $\F_G$ is viewed as a polynomial in the Schwinger parameters after plugging in generic kinematic parameters.
The \emph{Landau polytope} $\LL_G$ of $G$ is given by
\be
\LL_G := \mathrm{Newt} \left( \prod_{e=1}^{\E_G} \frac{\partial \F_G}{\partial \alpha_{e}} \right) = \bigoplus_{e=1}^{\E_G} \LL_{G,e}.
\ee
\end{definition}
One motivation to consider these polytopes is to obtain a general bound on the degree of the Landau discriminant polynomial $\Delta_G$, see Prop.~\ref{prop:degmixedres}. 

Our next goal is to use Symanzik polytopes of subdiagrams $G' \subset G$ to describe the faces of the Landau polytope $\LL_G$. We start by considering the boundary structure of the constituent polytopes $\LL_{G,e}$.

\begin{proposition}\label{prop4}
For a subdiagram $\gamma \subset G$, let $\gam{\LL_{G,e}} := \partial_{w_\gamma} \LL_{G,e}$ be the corresponding face of $\LL_{G,e}$, with $w_\gamma = \sum_{e \in \gamma} w_e$. If either $\gamma$ and $\gamma \setminus e$ are both 1VI or $\gamma = e$, we have that $\partial_\gamma \LL_{G,e}$ is a facet given by 
\be\label{eq:prop4}
\partial_\gamma \LL_{G,e} = \begin{dcases}
	\UU_e \times \FF_{G\setminus e} &\qquad\textup{if}\qquad \gamma = e,\\
	\UU_{\gamma\setminus e} \times \FF_{G/\gamma} &\qquad\textup{if}\qquad e \in \gamma,\; \L_\gamma >0, \\
	\UU_{\gamma} \times \FF_{(G / \gamma)\setminus e} &\qquad\textup{if}\qquad e \in G/\gamma,\; \L_\gamma >0.
\end{dcases}
\ee
\end{proposition}
\begin{proof}
Using \eqref{eq:FGe} we have
\be\label{eq:partial-FG}
\frac{\partial \F_G}{\partial \alpha_e}\bigg|_{\alpha_\gamma \rightarrow \epsilon \alpha_\gamma} \!\! = \epsilon^{\L_\gamma} \left( \frac{\partial \U_\gamma}{\partial \alpha_e} \F_{G/\gamma} + \U_{\gamma} \frac{\partial \F_{G/\gamma}}{\partial \alpha_e}\right) + {\cal O}(\epsilon^{\L_\gamma + 1}).
\ee
Observe that, since $e$ is contained in either $\gamma$ or $G/\gamma$, at most one of the two terms in the leading coefficient standing with $\epsilon^{\L_\gamma}$ can be non-zero. We claim that the only one situation in which \emph{both} leading terms in \eqref{eq:partial-FG} vanish is when $\gamma$ is a tree ($\L_\gamma = 0$) containing $e$. 

If $e \in \gamma$, the first term vanishes if and only if $\gamma$ is a tree (for which $\partial \U_{\gamma} / \partial \alpha_e = 0$). 

If $e \in G / \gamma$, the second term is identically zero only when $G/\gamma$ contains a massless tadpole ($\m_e = 0$ and no external momentum flows through $e$, for which $\partial \F_{G/\gamma} / \partial \alpha_e =0$). This is excluded by our assumption of generic kinematics.

Let us investigate what happens when $e \in \gamma$ and $\L_\gamma = 0$. Then by \eqref{eq:prop1-2} we have
\be\label{eq:Fe2}
\frac{\partial \F_G}{\partial \alpha_e}\bigg|_{\alpha_e \rightarrow \epsilon \alpha_e} \!\! = \epsilon (\F_{G\setminus e} - \m_e \U_{G/e}) + {\cal O}(\epsilon^2).
\ee
If $\gamma$ has more than one edge, one can apply Prop.~\ref{prop:prop3} to conclude that $\gam{\LL_{G,e}}$ is a face of codimension $>1$. Hence it suffices to consider the case where $\gamma = e$.

The three cases corresponding to one or two zero terms in the leading coefficient of \eqref{eq:partial-FG} can be summarized as follows:
\be\label{eq:dFdz}
\frac{\partial \F_G}{\partial \alpha_e}\bigg|_{\alpha_\gamma \rightarrow \epsilon \alpha_\gamma} \!\! = \begin{dcases}
	\epsilon^{L_\gamma+1} \, \U_e \left(\F_{G\setminus e} - \m_e \U_{G/e}\right) + \ldots  &\textup{if } \gamma = e,\\
\epsilon^{L_\gamma} \,\U_{\gamma\setminus e}\, \F_{G/\gamma} + \ldots &\textup{if } e \in \gamma,\; L_\gamma >0, \\
\epsilon^{L_\gamma+1}\, \U_{\gamma} \left(\F_{(G/\gamma)\setminus e} - \m_e \U_{(G/ \gamma)/e} \right) + \ldots &\textup{if } e \in G/\gamma,\; L_\gamma >0.
\end{dcases}
\ee
In the first line we have inserted $\U_e = 1$ to make the structure more apparent.
In the second line we used \eqref{eq:prop1-1} with $G \to \gamma$, and in the third we used \eqref{eq:Fe2} with $G \to G/\gamma$.

We now investigate the Newton polytope of the leading coefficient polynomials in \eqref{eq:dFdz}. The first contains two terms: proportional to $\F_{G\setminus e}$ and $\U_{G/e}$. However, all the monomials in $\U_{G/e}$ are already contained in $\F_{G\setminus e}$ because every spanning tree in $G/e$ is also a $2$-tree in $G\setminus e$. Hence, for generic kinematics,
\be
\mathrm{Newt}(\F_{G\setminus e} - \m_e \U_{G/e}) = \mathrm{Newt}(\F_{G\setminus e}) = \FF_{G\setminus e}.
\ee
This gives the first line in \eqref{eq:prop4}. Recall that $\UU_e$ is a point.
A similar discussion applies to the third line of \eqref{eq:dFdz}, where the Newton polytope of the term in the parenthesis becomes $\FF_{(G/\gamma)\setminus e}$. Together with the 1VI condition on $\gamma$, it gives the required facet in \eqref{eq:prop4}. The second line in \eqref{eq:dFdz} gives the corresponding facet in \eqref{eq:prop4} provided that $\gamma \setminus e$ is 1VI.
\end{proof}

Boundaries of $\LL_{G,e}$ are therefore given by all ways in which we can combine the operations of shrinking $\gamma$ and removing $e$. Finally, we arrive at a similar result for the full Landau polytope.

\begin{proposition}\label{prop:L-boundary}
For a connected 1VI diagram $G$ and a subdiagram $\gamma \subset G$, let $\gam{\LL_{G}} := \partial_{w_\gamma} \LL_{G}$ be the corresponding face of $\LL_{G}$, with $w_\gamma = \sum_{e \in \gamma} w_e$. The polytope $\LL_{G}$ has a facet given by
\be\label{eq:prop5-1}
\partial_{\gamma} \LL_G = \UU_e \times \Bigg( \FF_{G\setminus e} \oplus \bigoplus_{\substack{e' \in G\\ e' \neq e}} \FF_{(G/ e)\setminus e'}\Bigg)
\ee
if $\gamma = e$ is a single edge, and a facet given by
\be\label{eq:prop5-2}
\partial_\gamma \LL_G = \left ( \E_{G/\gamma} \cdot \UU_{\gamma} \oplus \bigoplus_{e' \in \gamma} \UU_{\gamma\setminus e'} \right) \times \left (\E_\gamma \cdot \FF_{G/\gamma} \oplus \bigoplus_{e'\notin \gamma} \FF_{(G / \gamma)\setminus e'} \right)
\ee
if $\gamma$ is a 1VI subdiagram with $\L_\gamma > 0$. Here $c \cdot \mathbf{P}$ denotes a dilation of the polytope $\mathbf{P}$ by a constant $c>0$, and $\E_\gamma, \E_{G/\gamma} = \E_G - \E_\gamma$ are the number of edges in $\gamma$ and $G/\gamma$ respectively.
\end{proposition}
\begin{proof}
The face $\gam{\LL_G}$ is the Minkowski sum of the faces $\gam{\LL_{G,e'}}$, given in Prop.~\ref{prop4}, for every $e' \in G$.

When $\gamma=e$ is a single edge, there are two cases: either $e' = e$ giving $\UU_e \times \FF_{G\setminus e}$ from the first line of \eqref{eq:prop4}, or $e' \neq e$ which falls under the case $e' \notin e$ in the third line of \eqref{eq:prop4}. The latter gives $\UU_{e} \times \FF_{(G/e)\setminus e'}$ for every $e' \in G$, $e' \neq e$, showing \eqref{eq:prop5-1} since $\UU_e$ is a point.

Similarly, when $\gamma$ is not a tree, there are two cases depending on whether the specific $e'$ is or is not contained in $\gamma$. According to the second line in \eqref{eq:prop4} the former gives $\UU_{\gamma\setminus e'} \times \FF_{G/\gamma}$, while the latter $\UU_\gamma \times \FF_{(G/\gamma)\setminus e'}$. Summing over all possibilities for $e' \in G$ leaves us with the Newton polytope
\begin{align}
\partial_\gamma \LL_G &= \Newt \left(\prod_{e' \in \gamma} \left(\U_{\gamma\setminus e'}\, \F_{G/\gamma}\right)  \prod_{e' \notin \gamma} \left(\U_{\gamma}\, \F_{(G/\gamma)\setminus e'}\right) \right)\nn\\
&= \Newt \left( \U_\gamma^{\E_{G/\gamma}} \prod_{e' \in \gamma} \U_{\gamma\setminus e'} \right) \times \Newt\left(\F_{G/\gamma}^{\E_\gamma} \prod_{e' \notin \gamma} \F_{(G/\gamma)\setminus e'}\right),
\end{align}
where in the second line we have used `$\times$' since the sets of Schwinger parameters are disjoint. This leads to \eqref{eq:prop5-2}.
\end{proof}
Prop.~\ref{prop:L-boundary} identifies facets of the Landau polytope $\LL_G$ coming from 1VI subdiagrams. We expect that all facets arise in this way. 
\begin{conjecture} \label{conj:facets}
Let $G$ be a connected, 1VI diagram. There are $1:1$ correspondences
\be
\left \{ \textup{1VI subdiagrams $\gamma \subset G$} \right \} \quad \overset{1:1}{\longleftrightarrow} \quad \left \{ \textup{facets $\gam{\FF_G}$} \right \} \quad \overset{1:1}{\longleftrightarrow} \quad \left \{ \textup{facets $\gam{\LL_G}$} \right \}.
\ee
\end{conjecture}
We have verified Conj.~\ref{conj:facets} for the diagrams in Fig.~\ref{fig:diagrams} with the help of \texttt{Polymake} \cite{gawrilow2000polymake}. For the one-loop diagrams $\mathtt{A}_\E$, $\F_{\mathtt{A}_\E}$ is a dense quadratic polynomial (see Sec.~\ref{sec:ngon}) and its derivatives are linear forms. In this case, $\LL_{\mathtt{A}_\E}$ and $\FF_{\mathtt{A}_\E}$ are dilations of the standard simplex, and Conj.~\ref{conj:facets} trivially holds. For the banana diagrams $\mathtt{B_E}$, we prove the conjecture in the next section. Note that Conj.~\ref{conj:facets} implies that the normal fans of $\FF_G$ and $\LL_G$ coincide at the level of rays. We will see (Ex.~\ref{ex:banana}) that this is not true for cones of higher dimension. 
\begin{example}\label{ex:UFLpar}
Let us consider the diagram $G = \mathtt{par}$, labelled according to Fig.~\ref{fig:par}. The Symanzik polynomials read
\begin{align}
\U_{\mathtt{par}} &= (\alpha_1+\alpha_2)(\alpha_3+\alpha_4) + \alpha_3 \alpha_4,\\
\F_{\mathtt{par}} &= s \alpha_1 \alpha_2(\alpha_3 + \alpha_4) + (\M_3 \alpha_2 + \M_4 \alpha_1) \alpha_3 \alpha_4 - ({\textstyle\sum}_{e=1}^{4} \m_e \alpha_e) \U_{\mathtt{par}}.
\end{align}
The polytopes are given by
\be
\UU_{\mathtt{par}} = \Newt(\U_{\mathtt{par}}),\qquad  \FF_{\mathtt{par}} = \UU_{\mathtt{par}} + \mathbf{\Delta}_3,\qquad
\LL_{\mathtt{par}} = \bigoplus_{e=1}^{4} \Newt\left(\frac{\partial \F_{\mathtt{par}}}{\partial \alpha_e} \right),
\ee
where we used \eqref{eq:FUDelta}. They are displayed in Fig.~\ref{fig:UFLpar} (after dehomogenization by setting $\alpha_4 = 1$).
\begin{figure}[t]
    \centering
    \includegraphics[width=\textwidth]{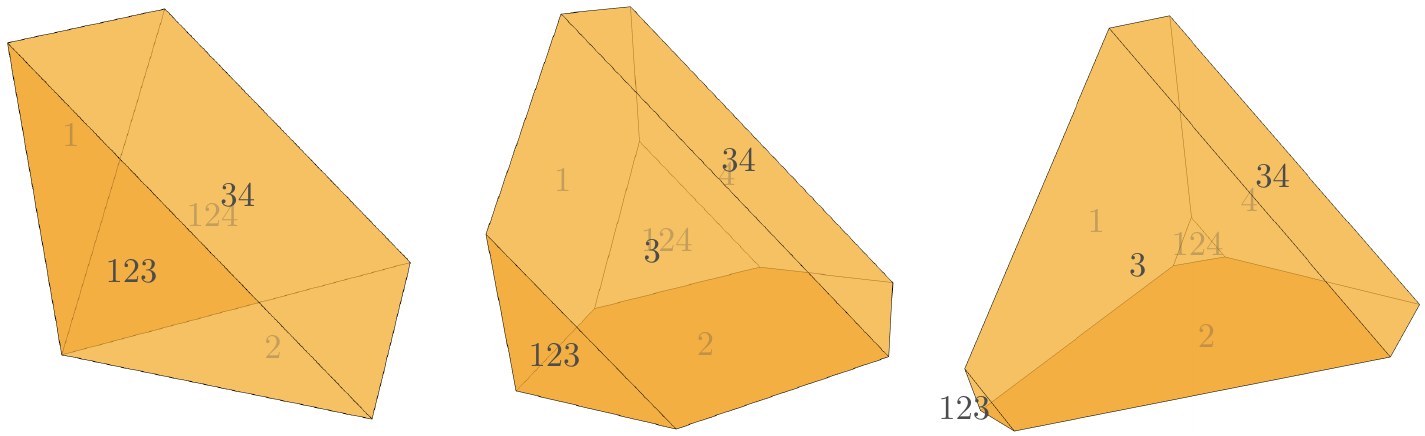}
    \caption{From left to right: polytopes $\UU_{\mathtt{par}}$, $\FF_{\mathtt{par}}$, $\LL_{\mathtt{par}}$. Facets are labelled by the edges belonging to the corresponding subdiagram $\gamma \subset \mathtt{par}$.}
    \label{fig:UFLpar}
\end{figure}
The set of facets is labelled by the 1VI subdiagrams
\be
\gamma \in \{ 1, 2, 3, 4, 34, 123, 124\},
\ee
where we represented each $\gamma$ by listing the edges that belong to it, e.g., the facet $\gamma=34$ corresponds to the subdiagram consisting of edges 3 and 4 (and their vertices). The polytope $\UU_{\mathtt{par}}$ is a pyramid with $5$ facets because for $\gamma = 123, 124$, $\mathtt{par}/\gamma$ is not 1VI. The polytopes $\FF_{\mathtt{par}}$ and $\LL_{\mathtt{par}}$ turn out to be combinatorially isomorphic, with $7$ facets each. Their $f$-vector is $(10, 15, 7)$.
\end{example}

\subsection{\label{sec:banana-polytopes}Example: Banana Diagrams}

As an explicit example of an infinite family of polytopes, in this subsection we focus on banana diagrams $\mathtt{B}_{\E}$ for any $\E \geq 2$, which we already considered in Sec.~\ref{sec:banana}. Previous work includes \cite{Klemm:2019dbm}. Recall that the Symanzik polynomials are given by
\be
\U_{\mathtt{B}_\E} = \sum_{j = 1}^\E \prod_{\substack{e=1\\ e \neq j}}^{\E} \alpha_e, \qquad \qquad \F_{\mathtt{B}_\E} = s \prod_{e=1}^{\E} \alpha_e -  \left ( \sum_{e=1}^{\E} \m_e \alpha_e \right ) \U_{\mathtt{B}_\E} .
\ee
Since the monomial $\prod_e \alpha_e$ occurs in $\left( \sum_e \m_e \alpha_e \right) \U_{\mathtt{B}_\E}$, we have 
\begin{align}
\FF_{\mathtt{B}_\E} &= \Newt \left ( \sum_e \m_e \alpha_e \right ) \oplus \UU_{\mathtt{B}_\E} \nn\\
&= \conv(w_1, w_2 \ldots, w_\E) \oplus \conv(u_1, u_2, \ldots, u_e), 
\end{align}
where $w_1, w_2, \ldots, w_\E$ are the standard basis vectors of $\R^\E$ and $u_i = \sum_{e \neq i} w_e$. Recall that $\mathbf{\Delta}_{\E-1} := \conv(w_1, w_2, \ldots, w_\E)$ is the standard simplex. The facet representations of $\UU_{\mathtt{B}_\E}$ and $\mathbf{\Delta}_{\E-1}$ are
\begin{align} 
\mathbf{\Delta}_{\E-1} &= \left \{ x \in \R^{\E} ~|~ x_e \geq 0,\quad e = 1, \ldots, \E  \;\text{ and }\; \sum_e x_e = 1 \right \},\\
\UU_{\mathtt{B}_\E} &= \left \{ x \in \R^{\E} ~|~ x_e \leq 1,\quad e = 1, \ldots, \E   \;\text{ and }\; \sum_e x_e = \E-1 \right \}.
\end{align}
Both these polytopes are simplices, so that $k$-dimensional faces are obtained by replacing $\E-k$ defining inequalities by equalities.

 A useful observation is that for a subset $\gamma \subset \{1,2, \ldots, \E \}$ of cardinality $k$ and $v = \sum_{j \in \gamma } c_j w_{j}, c_j > 0$, we have that $\v{\UU_{\mathtt{B}_\E}}$ is the codimension-$k$ face given by $\conv(w_j , j \notin \gamma)$. Moreover, all codimension-$k$ faces arise in this way. Similarly, for $v = - \sum_{j \in \gamma } c_j w_{j}, c_j > 0$, we have that $\v{\mathbf{\Delta}_{\E-1}}$ is the codimension-$k$ face given by $\conv(u_j , j \notin \gamma)$.

\begin{proposition} \label{prop:facetsF}
The facets of $\FF_{\mathtt{B}_\E}$ are in one to one correspondence with the faces of the standard simplex $\mathbf{\Delta}_{\E-1}$. In particular, the facets $\v{\FF_{\mathtt{B}_\E}}$ of $\FF_{\mathtt{B}_\E}$ are given by $v = \sum_{j \in \gamma} w_j$, for all proper, nonempty subsets $\gamma \subset \{1,2, \ldots, \E\}$. 
\end{proposition}

\begin{proof}

Since $\FF_G$ is a generalized permutohedron \cite{Schultka:2018nrs}, we know that for each facet $\v{\FF_G}$, $v$ is proportional to $\sum_{j \in \gamma} w_j$. 
Conversely, if $v = \sum_{j \in \gamma} w_j$ for a nonempty, proper subset $\gamma \subset \{1, \ldots, E\}$, we need to check that $\v{\FF_{\mathtt{B}_\E}}$ is a facet of $\FF_{\mathtt{B}_\E}$. Note that $\v{\mathbf{\Delta}_{\E-1}} = \Conv(w_j, j \notin \gamma)$ and $\v{\UU_{\mathtt{B}_\E}} = \Conv(u_j, j \in \gamma)$. Translating $\v{\UU_{\mathtt{B}_\E}} = \Conv(u_j, j \in \gamma)$ by $-\sum_e w_e$, we see that the sum of these faces of dimension $|\gamma|-1$ and $\E - |\gamma|-1$ is indeed a facet. 
\end{proof}

We now relate the facet-structure of $\FF_{\mathtt{B}_\E}$ to that of $\LL_{\mathtt{B}_\E}$. Recall that
\be
\LL_{\mathtt{B}_\E} = \bigoplus_{e = 1}^\E \LL_{{\mathtt{B}_\E},e} = \bigoplus_{e = 1}^\E \Newt \left (\frac{\partial \F_{\mathtt{B}_\E}}{\partial \alpha_e} \right ).
\ee
Since $ \frac{\partial \F_{\mathtt{B}_\E}}{\partial \alpha_e} = s  \prod_{i \neq e} \alpha_i  - \m_e \U_{\mathtt{B}_\E} - \left ( \sum_{i = 1}^\E \m_i \alpha_i \right )  \frac{\partial \U_{\mathtt{B}_\E}}{\partial \alpha_e}$ and
\be
\Newt( s \prod_{i \neq e} \alpha_i ) \subset \Newt( \U_{\mathtt{B}_\E}) \subset \Newt \left ( \left ( \sum_{i = 1}^\E \m_i \alpha_i \right )  \frac{\partial \U_{\mathtt{B}_\E}}{\partial \alpha_e}  \right ),
\ee
we conclude that 
\be
\Newt \left( \alpha_e \frac{\partial \F_{\mathtt{B}_\E}}{\partial \alpha_e} \right ) = \mathbf{\Delta}_{\E-1} \oplus \UU_{\mathtt{B}_\E,e},
\ee
where
\be
\UU_{\mathtt{B}_\E,e} := \conv( u_i, i \neq e) = \UU_{\mathtt{B}_\E} \cap \{ x_e = 1 \}.
\ee
As a consequence, the polytope $\LL_{\mathtt{B}_\E}$ is, up to translation by $(1,1,\ldots, 1)$, given by $\E \cdot \mathbf{\Delta}_{\E-1} \oplus \bigoplus_{e = 1}^\E \UU_{\mathtt{B}_\E,e}$.

\begin{proposition} \label{prop:facetsL}
The facets of $\LL_{\mathtt{B}_\E}$ are in one-to-one correspondence with the facets of $\FF_{\mathtt{B}_\E}$. In particular, the facets $\v{\LL_{\mathtt{B}_\E}}$ of $\LL_{\mathtt{B}_\E}$ are given by $v = \sum_{j \in \gamma} w_j$, for all proper, nonempty subsets $\gamma \subset \{1,2, \ldots, \E\}$. 
\end{proposition}
\begin{proof}

Since $\LL_{\mathtt{B}_\E} = \E \cdot \mathbf{\Delta}_{\E-1} \oplus \bigoplus_{e = 1}^\E \UU_{\mathtt{B}_\E,e}$ is a Minkowski sum of generalized permutohedra, it is a generalized permutohedron itself. Therefore, if $\v{\LL_\E}$ is a facet, $v$ is proportional to $\sum_{j \in \gamma} w_j$ for some nonempty strict subset $\gamma \subset \{1, 2, \ldots, \E\}$.
Conversely, we need to show that the vector $v = \sum_{j \in \gamma} w_j$ defines a facet. Suppose $\gamma$ has cardinality $k$. Then $v$ defines a $(k{-}1)$-dimensional face of $\bigoplus_{e=1}^\E \UU_{\mathtt{B}_\E,e}$ contained in $\E \cdot \Newt \left ( \sum_{i \in \gamma} \prod_{j \neq i} \alpha_j \right )$. The translation of the latter by $(-\E, \ldots, -\E)$ is contained in the span of $\{u_i, i \in \gamma \}$. The face $\E \cdot \v{\mathbf{\Delta}_{\E-1}}$ is contained in the span of the remaining standard basis vectors $\{u_i, i \notin \gamma \}$, so that
\begin{align}
\dim \left ( \E \cdot \v{\mathbf{\Delta}_{\E-1}}  \oplus  \bigoplus_{e=1}^\E \v{\UU_{\mathtt{B}_\E,e}}  \right )
&= \dim (\E \cdot \v{\mathbf{\Delta}_{\E-1}}) + \dim \left( \bigoplus_{e=1}^\E \v{\UU_{\mathtt{B}_\E,e}} \right)\nn \\
&= k - 1+ \E-1-k = \E-2.
\end{align}
This concludes the proof.
%
\end{proof}

\begin{corollary}
Conj.~\ref{conj:facets} holds for $G = \mathtt{B}_\E$, $\E \geq 2$.
\end{corollary}
\begin{example} \label{ex:banana} 
We illustrate Prop.~\ref{prop:facetsF} and \ref{prop:facetsL} for the banana diagram with $\E = 4$ edges. In this case, $\UU_{\mathtt{B}_4} = \mathbf{\Delta}_3$, $\FF_{\mathtt{B}_4}$ and $\LL_{\mathtt{B}_4}$ are $3$-dimensional polytopes in $\R^4$. Their projections onto $(\alpha_2, \alpha_3, \alpha_4)$-space are shown (not to scale) in Fig.~\ref{fig:polytopes}, which we generated using \texttt{Polymake.jl} \cite{gawrilow2000polymake,kaluba2020polymake}.
\begin{figure}
\centering
\includegraphics[width=\textwidth]{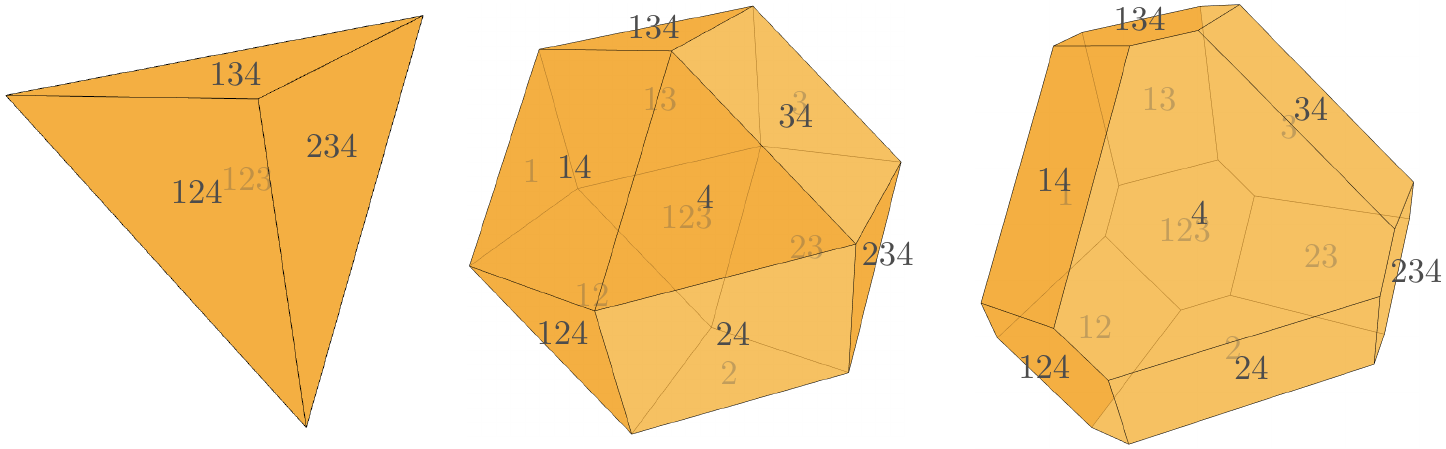}
\caption{From left to right: polytopes $\UU_{\mathtt{B}_4}$, $\FF_{\mathtt{B}_4}$, $\LL_{\mathtt{B}_4}$. Facets are labelled by the edges belonging to the corresponding subdiagram $\gamma \subset \mathtt{B}_4$.}
\label{fig:polytopes}
\end{figure}
Their facets are labelled by the $14$ 1VI subdiagrams
\be
\gamma \in \{1,2,3,4,12,13,14,23,24,34,123,124,134,234\},
\ee
with the same notation as in Ex.~\ref{ex:UFLpar}. $\UU_{\mathtt{B}_4}$ is a simplex and hence has only $4$ facets labelled by $\gamma = 123,124,134,234$, which are the only subdiagrams for which $\mathtt{B}_4/\gamma$ is 1VI. The polytopes $\FF_{\mathtt{B}_4}$ and $\LL_{\mathtt{B}_4}$ have $14$ facets each, corresponding to the $2^4 - 2$ choices of $\gamma$ in Prop.~\ref{prop:facetsF} and \ref{prop:facetsL}. Note that, although the normal fan for both polytopes coincide at the level of rays, the lower-dimensional face-structure is quite different. The $f$-vector of $\FF_{\mathtt{B}_4}$ is $(12,24,14)$, while that of $\LL_{\mathtt{B}_4}$ is $(24,36,14)$. 
\end{example}

\subsection{Bounds on the Degree of the Landau Discriminant} \label{subsec:degree}

In some cases, a bound on the degree of the Landau discriminant polynomial $\Delta_G$ can be obtained from the volume of $\FF_G$ and its faces. The statement requires some additional notation. For an $(\E_G{-}1)$-dimensional lattice polytope $\PPP \subset H_d \subset \R^{\E_G}$, let $\textup{Vol}(\PPP)$ be the volume of $\PPP$, viewed as a lattice polytope in the $(\E_G {-}1)$-dimensional quotient lattice $M_{\E_G - 1} = \Z^{\E_G}/(1,\ldots,1)\Z$. The volume is normalized such that $\textup{Vol}(\mathbf{\Delta}_{\E_G -1}) = 1$. The lattice points of each $D$-dimensional face $\mathbf{\Gamma} \subset \PPP$ generate a $D$-dimensional affine sublattice $M_{\mathbf{\Gamma}}$ of $M_{\E_G-1}$, in which $\mathbf{\Gamma}$ has volume $\textup{Vol}_{\mathbf{\Gamma}}(\mathbf{\Gamma})$. The normalization of the volume in the sublattice is again given by $\textup{Vol}_{\mathbf{\Gamma}}(\mathbf{\Delta}_D) = 1$. Finally, recall that a $D$-dimensional polytope is called \emph{simple} if all of its vertices are contained in exactly $D$ edges.

\begin{proposition} \label{prop:degAdisc}
If $\FF_G$ is simple, the degree of $\Delta_G$ is bounded by 
\be \label{eq:degAdisc}
\deg(\Delta_G) \leq \sum_{\mathbf{\Gamma} \subseteq \FF_G} (-1)^{\textup{codim}(\mathbf{\Gamma})} \,  (\textup{dim}\, \mathbf{\Gamma} + 1) \, \textup{Vol}_{\mathbf{\Gamma}}(\mathbf{\Gamma}),
\ee
where the sum is over all faces $\mathbf{\Gamma} \subset \FF_G$ together with $\FF_G$ viewed as a face of itself.
\end{proposition}
\begin{proof}
By \cite[Cor.~3.10]{Postnikov06faces}, every simple generalized permutohedron is Delzant. Therefore, the projective toric variety $X_{\FF_G}$ corresponding to $\FF_G$ is smooth and \cite[Ch.~9, Thm.~2.8]{gelfand2008discriminants} applies for its ${\cal A}$-discriminant. The proposition follows from the observation that $\nabla_G$ is a subvariety of a linear section of this ${\cal A}$-discriminant.
\end{proof}

\begin{remark}
The inequality in Prop.~\ref{prop:degAdisc} is expected to be strict in many cases, as the condition $\U_G \neq 0$ in \eqref{eq:LE} and the fact that $\FF_G$ is not ``generic" with respect to its monomial support may cause the degree of $\Delta_G$ to be strictly smaller than the bound in Prop.~\ref{prop:degAdisc}. 
\end{remark}

\begin{example}
For $G = \mathtt{A_4}$, we have $\FF_G = 2 \mathbf{\Delta}_3$ (see Sec.~\ref{sec:ngon}). Hence $\FF_G$ is simple, and the bound \eqref{eq:degAdisc} evaluates to 
\be
\deg(\Delta_{\mathtt{A_4}}) \leq \mathbf{1} \cdot (-1)^0 \cdot 4 \cdot 8 + \mathbf{4} \cdot (-1)^1 \cdot 3 \cdot 4 + \mathbf{6} \cdot (-1)^2 \cdot 2 \cdot 2 + \mathbf{4} \cdot (-1)^3 \cdot 1 \cdot 1 = 4.
\ee
Here the numbers in bold correspond to the number of codimension 0, 1, 2, 3 faces of $\FF_G$. We have seen in Ex.~\ref{ex:box2} that this bound is tight. Another example for which $\FF_G$ is simple is $G = \mathtt{par}$ (Ex.~\ref{ex:UFLpar}). The bound \eqref{eq:degAdisc} gives $6 \leq 24$.
\end{example}

Recall that for positive, real parameters $\lambda_1, \ldots, \lambda_{\E_G -1}$, the \emph{mixed volume} of $\E_G - 1$ polytopes $\LL_i \subset H_{d_i}, i = 1, \ldots, \E_G-1$, denoted $\MV(\LL_1, \ldots, \LL_{\E_G-1})$, is the coefficient standing with $\lambda_1 \lambda_2 \cdots \lambda_{\E_G - 1}$ in the homogeneous polynomial $\frac{1}{(\E_G-1)!}\Vol(\lambda_1 \cdot \LL_1 \oplus \cdots \oplus \lambda_{\E_G -1} \cdot \LL_{\E_G -1})$.
\begin{proposition} \label{prop:degmixedres}
The degree of $\Delta_G$ satisfies
\be \label{eq:degmixedres}
\deg(\Delta_G) \leq \sum_{e = 1}^{\E_G} \textup{MV}(\LL_{G,1}, \ldots, \widehat{\LL_{G,e}}, \ldots, \LL_{G,\E_G} ),
\ee
where $\widehat{\LL_{G,e}}$ indicates that the $e$-th polytope is omitted in the list.
\end{proposition}
\begin{proof}
For $e = 1, \ldots, \E_G$, let $\A_e \subset \Z^{\E_G}/(1,\ldots,1)\Z$ be the set of lattice points in $\LL_{G,e}$. Observe that $\nabla_G$ is a subvariety of a linear section of the mixed $({\cal A}_1, \ldots, {\cal A}_{\E_G})$-resultant \cite[Ch.~8, \S 1]{gelfand2008discriminants}. By \cite[Ch.~8, Prop.~1.6]{gelfand2008discriminants}, the proposition follows.
\end{proof}

\begin{remark} 
The degree bound from Prop.~\ref{prop:degAdisc} has the disadvantage that it only holds in the case where $\FF_G$ is simple. The bound from Prop.~\ref{prop:degmixedres} applies more generally. If $\FF_G$ is simple, the bound in \eqref{eq:degAdisc} is bounded from above by the bound from Prop.~\ref{prop:degmixedres}. We point out that both degree bounds trivially hold for $\Delta_G(\LLL)$, where $\LLL$ is any projective subspace of $\PP(\K_G)$ (e.g., the equal mass case).
\end{remark}

\begin{example}
For $G = \mathtt{par}$, the bounds evaluate to $6 \leq 24 \leq 8+8+6+6 = 28$. For $G = \mathtt{A_4}$, both bounds are tight: \eqref{eq:degAdisc} and \eqref{eq:degmixedres} give $4 \leq 4 \leq 4$. In the case of $\mathtt{A_4}$ and $\LLL = \PP^3$ as in Ex.~\ref{ex:box2}, we find $\deg(\Delta_{\mathtt{A}_4}(\LLL)) = 2 \leq 4 \leq 4$.
\end{example}

%% file: section5.tex
In this section we change gears and demonstrate how to apply homotopy continuation methods to the problem of counting the number of independent multi-loop Feynman integrals in analytic regularization. The approach rests on a theorem by Huh \cite{huh2013maximum} which identifies this dimension as the number of solutions to a system of rational critical point equations. A similar technique has been recently applied to tree-level scattering amplitudes in \cite{Sturmfels:2020mpv,agostini2021likelihood} (see also \cite{Liu:2018brz} for previous work).

\subsection{\label{sec:twisted}Feynman Integrals and Twisted Cohomology}

We first introduce a ``potential function'' $W_G$ associated to a given Feynman diagram $G$, constructed out of the Symanzik polynomials
\be\label{eq:WG}
W_G := (d_G - \D/2) \log\, \U_G - d_G \log \F_G + \sum_{e=1}^{\E_G} \delta_e \log \alpha_e,
\ee
where as before $\delta_e \in \C \setminus \Z$ are the analytic regulators and $d_G = m_G + \E_G - \L_G \D/2 + \delta$ is the degree of divergence with $\delta = \sum_{e=1}^{\E_G} \delta_e$.
Monodromies of $W_G$ define a line bundle ${\cal L}_G$ on the space of Schwinger parameters
\be\label{eq:MG}
X_G := (\C^{\ast})^{\E_G-1} \setminus \{ \U_G \F_G = 0\},
\ee
where we have fixed the projective gauge by setting $\alpha_{\E_G} = 1$.

We can now define the $k$-th locally-finite homology with coefficients in the line bundle ${\cal L}_G$ on $X_G$ (see, e.g., \cite{AomotoKita}), $H_k^{\text{lf}}(X_G, {\cal L}_G)$, as well as the twisted cohomology, $H^k(X_G, \boldsymbol{\nabla}_G)$, with the integrable connection $\boldsymbol{\nabla}_G := \d{+}\d W_G\wedge$.
A theorem of Aomoto, applied to the above case, states that the only non-trivial twisted (co)homology is in the middle dimension, $k = \dim_{\C}X_G = \E_G {-} 1$, if ${\cal L}_G$ is generic enough \cite[Thm. 1]{10.2969/jmsj/02720248}. That is, $d_G$ and the $\delta_e$'s do not satisfy certain linear relations. We refer the reader to \cite{AomotoKita} for more details on such constructions in the context of the theory of hypergeometric functions.

Feynman integrals in analytic regularization over some point in the physical region ${\cal P}_G \subset \R\PP(\K_G)$ are then defined as pairings
\be
H_{\E_G - 1}^{\text{lf}}(X_G, {\cal L}_G) \times H^{\E_G - 1}(X_G, \boldsymbol{\nabla}_G) \to \C,
\ee
given by
\be
( [\Gamma_G \otimes e^{W_G}], [\varphi] ) \mapsto \int_{\Gamma_G} \frac{\prod_{e=1}^{\E_G} \alpha_e^{\delta_e} }{\U_G^{\D/2 - d_G} \F_G^{d_G}} \varphi.
\ee
Here we choose the homology class, which has the same effect as the $i\epsilon$ factor in \eqref{eq:IG}, given by
\be\label{eq:Gamma-G}
\Gamma_G := \Big\{ \left( \alpha_e \exp \left(i\epsilon \frac{\partial \V_G}{\partial \alpha_e}\right) \right )_{ e=1,2,\ldots,\E_G {-} 1} \;\Big |\; \alpha_e \in \R_+ \Big\} \subset X_G
\ee
for sufficiently small $\epsilon$. In this section we are interested in Feynman integrals that \emph{do} converge, for values of kinematic invariants where Landau equations are not satisfied, in which case the above prescription is equivalent to the Feynman $i\epsilon$ prescription \cite{Mizera:2021fap}.
A twisted cohomology class $[\varphi]$ can be represented by a holomorphic $(\E_G {-} 1)$-form on $X_G$. That is, a meromorphic form on $\C^{\E_G-1}$ with poles along the divisor $\{\alpha_1 \cdots \alpha_{\E_G-1}\U_G = 0 \}$. Such forms are given by polynomials in $\alpha_e, 1/\alpha_e$ and $1/\U_G$ times the measure $\d^{\E_G - 1}\alpha$. Different choices of $[\varphi]$ give Feynman integrals within the same \emph{family} defined by the potential $W_G$.

Since $\R_+^{\E_G - 1}$ remains fixed, the problem of counting the number of independent Feynman integrals (known as the \emph{master integrals}) within a given family amounts to computing the dimension of the cohomology $H^{\E_G -1}(X_G, \boldsymbol{\nabla}_G)$. Because it is the only non-trivial cohomology and the line bundle ${\cal L}_G$ is flat, the topological Euler characteristic of $X_G$ equals \cite[Thm. 3]{10.2969/jmsj/02720248}
\begin{align}
\chi(X_G) &= \sum_{k} (-1)^k \dim H^{k} (X_G, \boldsymbol{\nabla}_G)\nn\\ &= (-1)^{\E_G - 1} \dim H^{\E_G - 1}(X_G, \boldsymbol{\nabla}_G).\label{eq:dimH}
\end{align}
Thus, the number we are looking for is the signed Euler characteristic. In practical applications it might be difficult to compute $\chi(X_G)$ directly for large enough $G$.

A more efficient route stems the connection to Morse theory on $X_G$ associated to the function $\Re(W_G)$. Since $W_G$ is holomorphic, the critical points of $\Re(W_G)$ are the same as those of $W_G$. Let us define the critical locus of $W_G$ to be
\be
\mathrm{Crit}(W_G) := \{ (\alpha_1, \alpha_2, \ldots, \alpha_{\E_G - 1}) \in X_G \;|\; \d W_G = 0 \}.
\ee
For suitably generic values of the parameters $\delta_1,\ldots, \delta_{\E_G}, d_G$, the above set is finite and all the critical points are isolated and non-degenerate.
In particular, we have the following result.
\begin{theorem}\label{thm:thm2}
The signed Euler characteristic $(-1)^{\E_G - 1} \chi(X_G)$ is equal to the number of critical points of $W_G$ for generic complex parameters $\delta_1, \delta_2, \ldots, \delta_{\E_G}, d_G$.
\end{theorem}
\begin{proof}
The statement of the theorem will follow from \cite[Thm. 1]{huh2013maximum} after we show that $X_G$ is a smooth very affine variety, i.e., a smooth closed subvariety of a torus. To see this, we embed $X_G \rightarrow (\C^*)^{\E_G}$ via 
\be
(\alpha_1, \alpha_2, \ldots, \alpha_{\E_G-1}) \mapsto \Big(\alpha_1, \alpha_2, \ldots, \alpha_{\E_G - 1}, \frac{1}{\U_G \F_G } \Big).
\ee
Hence $X_G$ is given by $y~ \U_G \F_G - 1 = 0$, where we use coordinates $(\alpha_1, \alpha_2, \ldots, \alpha_{\E_G-1},\allowbreak y)$ on $(\C^*)^{\E_G}$. At a singular point, we would have $\frac{\partial}{\partial y} (y~ \U_G \F_G - 1) = \U_G \F_G = 0$.
\end{proof}

For an alternative derivation using Morse theory, see \cite[Thm. 3]{10.2969/jmsj/02720248}.
Combining this result with \eqref{eq:dimH} we find that the number of critical points $W_G$ counts the number of independent Feynman integrals in analytic regularization.
The connection between counting master integrals, twisted cohomology, and the number of critical points was first explained in \cite{Mastrolia:2018uzb,Frellesvig:2019uqt} in the Baikov representation as well as \cite{Mizera:2019vvs} in the parametric representation. Related observations were previously made in the context of relative cohomology groups \cite{Lee:2013hzt} and D-modules \cite{Bitoun:2017nre}. In this section we gave a mild reformulation in terms of the potential $W_G$ from \eqref{eq:WG}, in order to stress more the connection to the combinatorics behind the Symanzik polynomials. Finding critical points has many other applications within the intersection theory of twisted cohomologies, including computation of integration-by-parts relations and differential equations for Feynman integrals \cite{Mizera:2017rqa,Mastrolia:2018uzb,Frellesvig:2019kgj,Frellesvig:2019uqt,Mizera:2019vvs,Weinzierl:2020xyy,Frellesvig:2020qot,Caron-Huot:2021xqj}. This motivates the need for an efficient algorithm for computing $\mathrm{Crit}(W_G)$ using homotopy continuation methods.

\begin{remark}
Using \eqref{eq:prop1-1} and \eqref{eq:prop1-2} the conditions for critical points can be written as
\be
\frac{\partial W_G}{\partial \alpha_e} = (d_G - \D/2) \frac{\U_{G\setminus e}}{\U_G} - d_G \frac{\F_{G \setminus e} + \m_e(\U_{G\setminus e} - 2 \U_G)}{\F_G} + \frac{\delta_e}{\alpha_e} = 0
\ee
for $e=1,2,\ldots,\E_G$.
Therefore we obtain a system of equations defined purely in terms of the combinatorics of the graph $G$.
\end{remark}

\subsection{\label{sec:chi-homotopy}Computational Results}

The critical points of $W_G$ can be computed with the help of \texttt{Landau.jl} and off-the-shelf software from numerical algebraic geometry, making the \texttt{Julia} code very concise:

\begin{minted}{julia}
edges = [[1,2],[2,3],[3,4],[4,5],[5,6],[6,1],[3,6]]
nodes = [1,2,4,5]

E = length(edges)

F, U, α, p, mm = getF(edges, nodes)
F, s, t, M, m = substitute4legs(F, p, mm)

@var u[1:E+2]
W = u[1] * log(U) + u[2] * log(F) + dot(u[3:E+2], log.(α))
dW = System(differentiate(subs(W, α[E] => 1), α[1:E-1]),
            parameters = [s; t; M; m; u])
            
Crit = monodromy_solve(dW)
crt = certify(dW, Crit)
println(ndistinct_certified(crt))
\end{minted}

The lines $1$--$7$ simply compute the Symanzik polynomials using \texttt{Landau.jl}, here in the example of $G=\mathtt{dbox}$ with generic masses. The potential function $W_G$ is given by $\mathtt{W}$, where $\mathtt{u[i]}$ are $\E_G{+}2$ generic parameters in front of the logarithms s in \eqref{eq:WG}. The following lines set up the critical point equations evaluated at $\alpha_{\E_G} =1$. These are solved using \texttt{HomotopyContinuation.jl} (v2.6.0) in line 14. In line 15, we use \texttt{certify} to get a rigorous proof that each of the computed solutions is an approximate solution in a suitable sense \cite{breiding2020certifying}. The number of distinct certified solutions, printed in line 16, is a lower bound on $\chi(X_G)$.

The above code is rather efficient, for example inputting $G = \mathtt{B}_{14}$, it finds $16383$ solutions in about $12$ minutes on the hardware used in Sec.~\ref{subsubsec:equalmass}. Certifying these solutions takes about 40 seconds. Let us apply the above code to the families of diagrams considered in this paper.

\begin{example}\label{ex:chiAB}
The results for $|\chi(X_{\mathtt{A}_{\E}})|$ and $|\chi(X_{\mathtt{B}_\E})|$ for $\E \leq 10$ are collected in Tab.~\ref{tab:chiAB}. For the banana diagrams $\mathtt{B}_{\E}$ we give the result for generic masses $\m_e$. The number of critical points matches the result $2^{\E}-1$ proven in \cite[Prop.~55]{Bitoun:2017nre}, in agreement with \cite{Kalmykov:2016lxx}.

\begin{table}[!t]
\centering
\begin{tabular}{c|c|c|c|c|c|c|c|c|c} 
$G$ & $\E=2$ & $3$ & $4$ & $5$ & $6$ & $7$ & $8$ & $9$ & $10$\\
\hline
$\mathtt{B}_\E$ $(\m_e)$ & $3$ & $7$ & $15$ & $31$ & $63$ & $127$ & $255$ & $511$ & $1023$ \\
$\mathtt{A}_\E$ $(\M_i,\m_e)$ & $3$ & $7$ & $15$ & $31$ & $63$ & $127$ & $255$ & $511$ & $1023$ \\
$\mathtt{A}_\E$ $(0, \m_e)$ & $2$ & $3$ & $11$ & $26$ & $57$ & $120$ & $247$ & $502$ & $1013$ \\
$\mathtt{A}_\E$ $(\M_i, 0)$ & $1$ & $4$ & $11$ & $26$ & $57$ & $120$ & $247$ & $502$ & $1013$ \\
$\mathtt{A}_\E$ $(0,0)$ & $1$ & $1$ & $3$ & $11$ & $33$ & $85$ & $199$ & $439$ & $933$ \\
\end{tabular}
\caption{Signed Euler characteristic $|\chi(X_G)|$ computed using the numerical code from Sec.~\ref{sec:chi-homotopy} for banana ($G = \mathtt{B}_\E$) and one-loop diagrams ($G = \mathtt{A}_\E$); see Ex.~\ref{ex:chiAB}.}\label{tab:chiAB}
\end{table}
In the case of the one-loop diagrams with $\E$ edges, $\mathtt{A}_\E$, we present the results for different ways of assigning masses, e.g., $(0,\m_e)$ means that all $\M_i = 0$ and $\m_e$ are left generic. We notice that in the fully generic case $(\M_i, \m_e)$, we seem to find $|\chi(M_{\mathtt{A}_\E})| = 2^{\E} -1$, and the case $(\M_i, 0)$ gives results consistent with $2^\E - \E - 1$. We conjecture these are valid for all $\E$. Previous work on counting the size of cohomology basis for one-loop integrals includes \cite{Boyling1968,AIF_2003__53_4_977_0}, though in different formalisms.
\end{example}

\begin{example}\label{ex:chiG}
The results for $|\chi(X_G)|$ for the other diagrams in Fig.~\ref{fig:diagrams} are presented in Tab.~\ref{tab:chiG} in the same notation as in Ex.~\ref{ex:chiAB}. We notice that for all the diagrams, the values of $|\chi(X_G)|$ are non-increasing going to the right of the table, as more and more special mass configurations are chosen. Individual entries of the table match previous results contained in \cite{Bitoun:2017nre,Klausen:2019hrg,Frellesvig:2020qot}.
\begin{table}[!t]
\centering
    \begin{tabular}{c|c|c|c|c|c|c|c}
        $G$ & $(\M_i,\m_e)$ & $(\M,\m)$ & $(0,\m_e)$ & $(0,\m)$ & $(\M_i,0)$ & $(\M,0)$ & $(0,0)$ \\
        \hline
        $\mathtt{par}$    & $19$  & $19$  & $13$  & $13$  & $4$  & $4$  & $1$ \\
        $\mathtt{acn}$    & $55$  & $55$  & $36$  & $25$  & $20$ & $20$ & $3$ \\
        $\mathtt{env}$    & $273$ & $273$ & $181$ & $181$ & $56$ & $56$ & $10$ \\
        $\mathtt{npltrb}$ & $116$ & $116$ & $77$  & $52$  & $28$ & $28$ & $5$ \\
        $\mathtt{tdetri}$ & $51$  & $51$  & $33$  & $33$  & $4$  & $4$  & $1$ \\
        $\mathtt{debox}$  & $43$  & $43$  & $31$  & $25$  & $11$ & $11$ & $3$ \\
        $\mathtt{tdebox}$ & $123$ & $123$ & $87$  & $87$  & $11$ & $11$ & $3$ \\
        $\mathtt{pltrb}$  & $81$  & $81$  & $61$  & $47$  & $16$ & $16$ & $4$ \\
        $\mathtt{dbox}$   & $227$ & $227$ & $159$ & $111$ & $75$ & $75$ & $12$ \\
        $\mathtt{pentb}$  & $543$ & $543$ & $430$ & $341$ & $228$ & $228$ & $62$ \\
    \end{tabular}
\caption{Signed Euler characteristic $|\chi(X_G)|$ computed using the numerical code from Sec.~\ref{sec:chi-homotopy} for diagrams $G$ from Fig.~\ref{fig:diagrams}; see Ex.~\ref{ex:chiG}.}\label{tab:chiG}
\end{table}

\end{example}

%% file: appendix.tex
\normalsize

In this appendix, we summarize the role of Feynman integrals in the theory of scattering amplitudes. Moreover, we present two different formulations of these integrals in some detail. The first representation, for which the integration domain is \emph{loop-momentum space}, can be constructed intuitively from a Feynman diagram $G$, taking the \emph{Feynman rules} for granted. The second formulation corresponds to the \emph{worldline formalism} discussed in Sec.~\ref{subsec:feynman}. It involves \emph{Schwinger parameters} and \emph{Symanzik polynomials}, which play a crucial role in this paper (see Sec.~\ref{sec:Symanzik}). The appendix' purpose is to make this paper more accessible for non-physicists, and to shed some light on the physical interpretation of the mathematical objects it studies. 

\paragraph*{Feynman diagrams and scattering amplitudes.}\,
The experimental set-up to keep in mind is a \emph{scattering process} or \emph{scattering experiment} in a \emph{particle accelerator}. For our purposes, the accelerator is a hollow sphere, and an experiment corresponds to sending elementary particles (e.g., photons, electrons, or muons) into it. After some \emph{particle interactions} inside the accelerator, some new particles exit.
The interactions that can happen depend on the physical theory governing the process. 

In total, there are $n$ ingoing and outgoing particles. These particles are labeled by their \emph{momentum vectors} $p_i = (p_i^{(0)}, p_i^{(1)}, \ldots, p_i^{(\D-1)}) \in \R^{1,\D-1}$, $i = 1, \ldots, n$. As in Sec.~\ref{subsec:feynman}, $\R^{1,\D-1}$ is the \emph{Minkowski momentum space}, endowed with the pairing 
$p \cdot q = p^{(0)}q^{(0)} - p^{(1)} q^{(1)} - \cdots - p^{(\D-1)} q^{(\D-1)}$, and we write $p^2 = p \cdot p$. In many real-world examples, the dimension $\D$ of $\R^{1,\D-1}$ is taken to be $4$. One dimension corresponds to time, and the other three are space dimensions. The momentum vectors capture the relevant physical information about the particles involved in the scattering experiment, such as their mass and velocity. \emph{Momentum conservation} imposes the relation $p_1 + \cdots + p_n = 0$. 

The \emph{scattering amplitude} $A(p_1, \ldots, p_n) : (\R^{1,\D-1})^n \rightarrow \C$ is a complex-valued function of the momenta, associated to a scattering process. Its modulus $|A(p_1, \ldots, p_n)|$ can be roughly thought of as a \emph{joint probability density function}, telling us what to expect for the outcome of the experiment. Such an amplitude function captures a great amount of physical information. Its evaluation and analytic properties are important active areas of research. For theories respecting Lorentz symmetry, an amplitude is invariant under the action of the Lorentz group. Therefore, it can be expressed as a function of the \emph{Lorentz invariants} $s_I = (\sum_{i \in I} p_i)^2$, $I \subset \{1, \ldots, n\}$ (also called \emph{Mandelstam invariants} when $1<|I| <n-1$), see Sec.~\ref{sec:Symanzik}. 

In \emph{perturbation theory}, one expresses the amplitude $A$ as a sum over all possible interaction patterns inside the particle collider. Let ${\cal G}$ be the set of these interaction patterns. We write 
\be \label{eq:perturb}
A = \sum_{G \in {\cal G}} {\cal I}_G. 
\ee 
This is usually an infinite sum, but as a rule of thumb, more complicated interaction patterns have a smaller contribution to the amplitude. Therefore, one hopes to approximate the amplitude by truncating the sum \eqref{eq:perturb} after considering sufficiently many ``simple'' scenarios $G$. As our notation suggests, the interaction patterns summed over in \eqref{eq:perturb} are encoded by \emph{Feynman diagrams}. Examples of Feynman diagrams are shown in Fig.~\ref{fig:diagrams} and Fig.~\ref{fig:bubble}. We state a definition below. The contribution ${\cal I}_G$ corresponding to the Feynman diagram $G$ is an integral, called a \emph{Feynman integral}. What this integral looks like is determined by the \emph{Feynman rules}. This is explained below. A standard strategy to study the analytic properties of the amplitude $A$, which is also applied in this paper, is to study those of the summands ${\cal I}_G$ in \eqref{eq:perturb}.

For our purposes, a \emph{Feynman diagram} is a connected, undirected but \emph{oriented} graph. That is, to each edge we assign an arbitrary orientation, but the resulting integral ${\cal I}_G$ does not depend on this choice. The graph has $n_G = n$ \emph{open edges}, corresponding to incoming and outgoing particles, and $\E_G$ \emph{internal edges}, corresponding to newly formed particles inside the accelerator. For instance, the Feynman diagram in Fig.~\ref{fig:bubble} has $n = 4$ open edges, labeled by $p_i$, and $\E_G = 2$ internal edges.

\paragraph*{Loop-momentum integrals.}\, We are now ready to describe the loop-momentum formulation of ${\cal I}_G$. Momentum conservation at each node of the diagram imposes that the sum of incoming momenta equals the sum of the outgoing momenta. This mirrors Kirchhoff's current law for electrical circuits. Assigning a momentum vector $q_e \in \R^{1,\D-1},  e = 1, \ldots, \E_G$ to each internal edge, we obtain a linear equation in $p_i, q_e$ for each node of the diagram. 

\begin{figure}[t]
    \centering
    \includegraphics{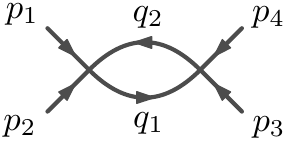}
    \caption{Bubble Feynman diagram used in Ex.~\ref{ex:bubble} and \ref{ex:pinch}.}
    \label{fig:bubble}
\end{figure}
\begin{example}[Bubble diagram]\label{ex:bubble}
Consider the bubble diagram in Fig.~\ref{fig:bubble}. For this graph, we have $n = 4$. Momentum conservation at the left and right vertex gives $p_1 + p_2 + q_2 = q_1$ and $p_3 + p_4 + q_1 = q_2$ respectively. 
\end{example}

These equations allow us to write the internal momenta $q_e$ in terms of the external momenta $p_i$ and $\L_G$ other independent parameters, called the \emph{loop momenta} $\ell_1, \ldots, \ell_{\L_G}$. Here $\L_G$ is the number of independent loops of $G$, or equivalently, the rank of the first homology group of $G$. If $\textup{V}_G$ is the number of vertices of $G$, we have $\L_G = \E_G - \textup{V}_G + 1$.

\begin{example}[Bubble diagram, continued] \label{ex:bubblecont}
The number of loops in the bubble Feynman diagram is $\L_G = 1$. Setting $q_2 = \ell$ we obtain $q_1 = p_1 + p_2 + \ell = - p_3 - p_4 + \ell$, where the second inequality is satisfied by the overall momentum conservation $p_1 + p_2 + p_3 + p_4 = 0$. 
\end{example}

We reiterate that momentum conservation fixes the internal momenta $q_e$ up to $\L_G$ degrees of freedom, called loop momenta. The Feynman integral ${\cal I}_G$ can be expressed as an integral over all possible loop momenta $\ell_1, \ldots, \ell_{\L_G}$. That is, the integration domain is $(\R^{1,\D-1})^{\L_G}$, where the $j$-th factor $\R^{1,\D-1} $ has coordinates $\ell_{j}^{(0)}, \ldots, \ell_{j}^{(\D-1)}$. The integrand is a product over all internal edges of the diagram $G$, in which the factor corresponding to the $e$-th edge is $i\hbar(q_e^2 - m_e^2 + i \epsilon)^{-1}$, with $m_e$ the mass of the particle propagating along edge $e$, a constant $\hbar>0$ known as the reduced Planck constant, and $i = \sqrt{-1}$. Here it is understood that $q_e$ is expressed as a linear combination of the external momenta $p_j$ and the loop momenta $\ell_j$. The integral reads
\be  \label{eq:loopmomentum}
{\cal I}_G = \frac{1}{(i\pi\hbar)^{\D\L_G/2}} \int_{(\R^{1,\D-1})^{\L_G}}  \prod_{e=1}^{\E_G} \frac{i\hbar}{q_e(\ell_1, \ldots, \ell_{\L_G}, p_1, \ldots, p_n)^2 - m_e^2 + i \epsilon} \, \d^{\D \L_G} \ell,
\ee
where $\d^{\D \L_G} \ell$ is short for $\d \ell_{1}^{(0)} \wedge \cdots \wedge \d \ell_{1}^{(\D-1)} \wedge \cdots \wedge \d \ell_{\L_G}^{(0)} \wedge \cdots \wedge \d \ell_{\L_G}^{(\D-1)}$ and the infinitesimal positive parameter $\epsilon$ is used to avoid singularities along the integration contour, also known as the \emph{Feynman $i\epsilon$ prescription}. The overall normalization is introduced for later convenience.

\paragraph*{Worldline formalism.}\, In order to rewrite the integral \eqref{eq:loopmomentum} as an integral over the positive orthant, we observe that 
\be
\frac{i\hbar}{q_e^2 - m_e^2 + i \epsilon} = \int_0^\infty e^{\frac{i}{\hbar}(q_e^2 - m_e^2 + i \epsilon)\alpha_e} \d \alpha_e.
\ee
The $\alpha_e$ are called \emph{Schwinger parameters}. 
Substituting this into \eqref{eq:loopmomentum} we obtain 
\be \label{eq:integr1} {\cal I}_G = \frac{1}{(i \pi\hbar)^{\D\L_G/2}} \int_{\R^{\E_G}_+}  \left (\int_{(\R^{1,\D-1})^{\L_G}}e^{\frac{i}{\hbar} \sum_{e = 1}^{\E_G} (q_e^2 - m_e^2 + i \epsilon)\alpha_e} \d^{\D \L_G} \ell \right) \d^{\E_G} \alpha, 
\ee
with $\d^{\E_G} \alpha = \d \alpha_1 \wedge \cdots \wedge \d \alpha_{\E_G}$. 
The advantage of this formulation is that now the inner integral is Gaussian, which allows us to integrate out the loop momenta analytically. There is a symmetric matrix $\mathbf{A}$, a column vector of momentum vectors $\mathbf{b}$ and a scalar $c$
such that 
\be \label{eq:ABc} \sum_{e = 1}^{\E_G} (q_e(\ell_1, \ldots, \ell_{\L_G},p_1, \cdots, p_n)^2 - m_e^2 + i \epsilon)\alpha_e =  \sum_{i,j = 1}^{\L_G} \mathbf{A}_{ij} (\ell_i \cdot \ell_j) + 2 \sum_{i = 1}^{\L_G} \mathbf{b}_{i} \cdot \ell_{i} + c. 
\ee
Here $\mathbf{A}$, $\mathbf{b}$ and $c$ have polynomial entries in the Schwinger parameters $\alpha_e$, the Lorentz invariants $s_I = \left( \sum_{i \in I} p_i \right)^2$, the internal masses $m_e$ and the parameter $\epsilon$. In fact, it is not hard to see that $\mathbf{A}$ only depends on the $\alpha_e$, and this dependence is linear. Replacing the exponents in \eqref{eq:integr1} by \eqref{eq:ABc}, the integral becomes 
\begin{align*}
{\cal I}_G &= \frac{1}{(i \pi\hbar)^{\D\L_G/2}} \int_{\R^{\E_G}_+}  \left (\int_{(\R^{1,\D-1})^{\L_G}}e^{\frac{i}{\hbar} (\sum_{i,j = 1}^{\L_G} \mathbf{A}_{ij} (\ell_i \cdot \ell_j) + 2 \sum_{i = 1}^{\L_G} \mathbf{b}_{i} \cdot \ell_{i} + c)} \d^{\D \L_G} \ell \right) \d^{\E_G} \alpha, \\
&= \int_{\R^{\E_G}_+} \frac{\d^{\E_G} \alpha}{(\det \mathbf{A})^{\D/2}} e^{ \frac{i}{\hbar} (-\mathbf{b}^\top \mathbf{A}^{-1} \mathbf{b} + c)}.
\end{align*} 
We now set $\U_G = \det \mathbf{A}$ and $\F_G = \U_G (-\mathbf{b}^\top \mathbf{A}^{-1} \mathbf{b} + c)$, where the Lorentz indices between $\mathbf{b}^\top$ and $\mathbf{b}$ are contracted in the second equation. This gives the integral in \eqref{eq:Feynman0} with $\N_G = 1$, also known as a \emph{scalar} Feynman integral. The polynomials $\U_G$ and $\F_G$ are known as \emph{Symanzik polynomials}. They can be obtained from the graph $G$ in a nice combinatorial way, see Sec.~\ref{sec:Symanzik}.

\begin{remark}Feynman integrals for more complicated scattering processes might give loop momentum dependent numerators in the integrand of \eqref{eq:loopmomentum}. Using a similar set of manipulations to those above, they translate to non-trivial polynomials $\N_G$. As discussed in Sec.~\ref{subsec:feynman}, additional regularization is needed to ensure convergence of $\I_G$. For more details we refer the reader to standard textbooks such as \cite{Smirnov:2012gma}.
\end{remark}

We conclude with an example that illustrates the worldline formulation for the bubble graph in Fig.~\ref{fig:bubble} and how the singularities of ${\cal I}_G$ depend on the singularities of $\F_G$, see Sec.~\ref{sec:landaueq}.

\begin{example}[Pinch singularity] \label{ex:pinch}
For the bubble diagram $G = \mathtt{B}_2$ from Fig.~\ref{fig:bubble}, the Symanzik polynomials are given by 
\[\U_{\mathtt{B}_2} = \alpha_1 + \alpha_2, \quad  \F_{\mathtt{B}_2} = s \alpha_1 \alpha_2 - (\m_1 \alpha_1 + \m_2 \alpha_2)(\alpha_1 + \alpha_2). \]
To illustrate how the Landau equations govern the singularities of the Feynman integral ${\cal I}_{\mathtt{B}_2}$ as a function of $s, \m_1, \m_2$, we will restrict to $\m = \m_1 = \m_2 = 1$. After regularization and performing the rewriting steps in Sec.~\ref{sec:landaueq}, this integral in the worldline formalism is given by 
\[ {\cal I}_{\mathtt{B}_2}(s) = \int_{\R_+} \frac{\textup{N}_{\mathtt{B}_2} \textup{R}_{\mathtt{B}_2}^{\textup{reg}}}{(\alpha_1 + 1)^{\D/2 - d_{\mathtt{B}_2}} (-\alpha_1^2 + (s-2) \alpha_1 -1)^{d_{\mathtt{B}_2}}} \d \alpha_1 ,\]
up to a constant factor. Recall that we set $\alpha_2 =1$. As explained in Ex.~\ref{ex:example4}, $(s:\m)$ belongs to the (equal-mass) Landau discriminant, if and only if $s = 4$. Indeed, for $s = 4$, $\F_{\mathtt{B}_2} = -\alpha_1^2 + (s-2) \alpha_1 -1 = 0$ has a singular solution, and $\U_{\mathtt{B}_2} \neq 0$ at this solution. To see how this value of $s$ causes trouble for the integration, let us evaluate the function ${\cal I}_{\mathtt{B}_2}(s)$ along a 1-real dimensional trajectory in the complex $s$-plane, as shown in the left panel of Fig.~ \ref{fig:pinch}. Note that the trajectory ends in $s = 4$. Along the way it crosses the real axis at $s^* > 4$. Before this happens, there are no poles of the integrand in $\R_+$. When $s$ passes the value $s^*$, ${\cal I}_{\mathtt{B}_2}(s)$ can be analytically continued by deforming the integration contour $\R_+$. This is illustrated in the right panel of Fig.~\ref{fig:pinch}. This deformation can be continued until $s$ reaches the value 4, at which point the relevant branch of ${\cal I}_{\mathtt{B}_2}(s)$ is necessarily singular.

\begin{figure}[!t]
    \centering
    \includegraphics[scale=1.1]{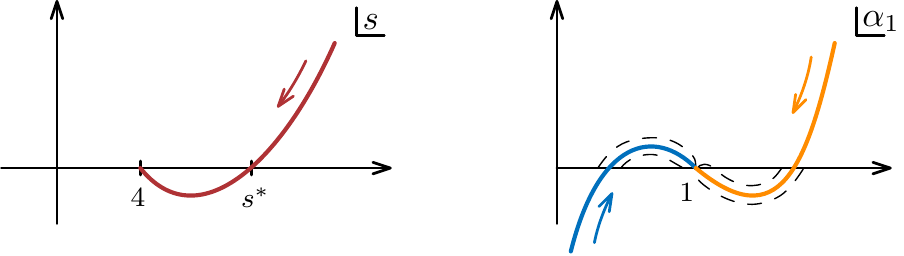}
    \caption{Illustration of Ex.~\ref{ex:pinch}. Left: trajectory of the parameter $s$ (in red) and the real axis. The end point of the trajectory is $s = 4$. Right: corresponding trajectories of the two solutions of $\F_{\mathtt{B}_2} = 0$ in the complex $\alpha_1$-plane, one in blue and one in orange. The black line represents the positive real axis. The dashed lines show a possible deformation of the integration contour as the blue and orange trajectories approach each other at $\alpha_1=1$.}
    \label{fig:pinch}
\end{figure}

\end{example}

%% file: article-JHEP.bbl
\providecommand{\href}[2]{#2}\begingroup\raggedright\begin{thebibliography}{10}

\bibitem{Smirnov:2012gma}
V.A.~Smirnov, \emph{{Analytic tools for Feynman integrals}}, vol.~250 (2012),
  \href{https://doi.org/10.1007/978-3-642-34886-0}{10.1007/978-3-642-34886-0}.

\bibitem{Caron-Huot:2016owq}
S.~Caron-Huot, L.J.~Dixon, A.~McLeod and M.~von Hippel, \emph{{Bootstrapping a
  Five-Loop Amplitude Using Steinmann Relations}},
  \href{https://doi.org/10.1103/PhysRevLett.117.241601}{\emph{Phys. Rev. Lett.}
  {\bfseries 117} (2016) 241601}
  [\href{https://arxiv.org/abs/1609.00669}{{\ttfamily 1609.00669}}].

\bibitem{Caron-Huot:2019vjl}
S.~Caron-Huot, L.J.~Dixon, F.~Dulat, M.~von Hippel, A.J.~McLeod and
  G.~Papathanasiou, \emph{{Six-Gluon amplitudes in planar $ \mathcal{N} $ = 4
  super-Yang-Mills theory at six and seven loops}},
  \href{https://doi.org/10.1007/JHEP08(2019)016}{\emph{JHEP} {\bfseries 08}
  (2019) 016} [\href{https://arxiv.org/abs/1903.10890}{{\ttfamily
  1903.10890}}].

\bibitem{Guerrieri:2020bto}
A.L.~Guerrieri, J.~Penedones and P.~Vieira, \emph{{S-matrix bootstrap for
  effective field theories: massless pions}},
  \href{https://doi.org/10.1007/JHEP06(2021)088}{\emph{JHEP} {\bfseries 06}
  (2021) 088} [\href{https://arxiv.org/abs/2011.02802}{{\ttfamily
  2011.02802}}].

\bibitem{Guerrieri:2021tak}
A.~Guerrieri and A.~Sever, \emph{{Rigorous bounds on the Analytic $S$-matrix}},
   \href{https://arxiv.org/abs/2106.10257}{{\ttfamily 2106.10257}}.

\bibitem{Adams:2006sv}
A.~Adams, N.~Arkani-Hamed, S.~Dubovsky, A.~Nicolis and R.~Rattazzi,
  \emph{{Causality, analyticity and an IR obstruction to UV completion}},
  \href{https://doi.org/10.1088/1126-6708/2006/10/014}{\emph{JHEP} {\bfseries
  10} (2006) 014} [\href{https://arxiv.org/abs/hep-th/0602178}{{\ttfamily
  hep-th/0602178}}].

\bibitem{Bellazzini:2020cot}
B.~Bellazzini, J.~Elias~Mir\'o, R.~Rattazzi, M.~Riembau and F.~Riva,
  \emph{{Positive Moments for Scattering Amplitudes}},
  \href{https://arxiv.org/abs/2011.00037}{{\ttfamily 2011.00037}}.

\bibitem{Tolley:2020gtv}
A.J.~Tolley, Z.-Y.~Wang and S.-Y.~Zhou, \emph{{New positivity bounds from full
  crossing symmetry}},
  \href{https://doi.org/10.1007/JHEP05(2021)255}{\emph{JHEP} {\bfseries 05}
  (2021) 255} [\href{https://arxiv.org/abs/2011.02400}{{\ttfamily
  2011.02400}}].

\bibitem{Caron-Huot:2020cmc}
S.~Caron-Huot and V.~Van~Duong, \emph{{Extremal Effective Field Theories}},
  \href{https://doi.org/10.1007/JHEP05(2021)280}{\emph{JHEP} {\bfseries 05}
  (2021) 280} [\href{https://arxiv.org/abs/2011.02957}{{\ttfamily
  2011.02957}}].

\bibitem{Prlina:2018ukf}
I.~Prlina, M.~Spradlin and S.~Stanojevic, \emph{{All-loop singularities of
  scattering amplitudes in massless planar theories}},
  \href{https://doi.org/10.1103/PhysRevLett.121.081601}{\emph{Phys. Rev. Lett.}
  {\bfseries 121} (2018) 081601}
  [\href{https://arxiv.org/abs/1805.11617}{{\ttfamily 1805.11617}}].

\bibitem{Drummond:2019cxm}
J.~Drummond, J.~Foster, O.~G\"urdogan and C.~Kalousios, \emph{{Algebraic
  singularities of scattering amplitudes from tropical geometry}},
  \href{https://doi.org/10.1007/JHEP04(2021)002}{\emph{JHEP} {\bfseries 04}
  (2021) 002} [\href{https://arxiv.org/abs/1912.08217}{{\ttfamily
  1912.08217}}].

\bibitem{Arkani-Hamed:2019rds}
N.~Arkani-Hamed, T.~Lam and M.~Spradlin, \emph{{Non-perturbative geometries for
  planar $ \mathcal{N} $ = 4 SYM amplitudes}},
  \href{https://doi.org/10.1007/JHEP03(2021)065}{\emph{JHEP} {\bfseries 03}
  (2021) 065} [\href{https://arxiv.org/abs/1912.08222}{{\ttfamily
  1912.08222}}].

\bibitem{Henke:2019hve}
N.~Henke and G.~Papathanasiou, \emph{{How tropical are seven- and
  eight-particle amplitudes?}},
  \href{https://doi.org/10.1007/JHEP08(2020)005}{\emph{JHEP} {\bfseries 08}
  (2020) 005} [\href{https://arxiv.org/abs/1912.08254}{{\ttfamily
  1912.08254}}].

\bibitem{Mago:2020kmp}
J.~Mago, A.~Schreiber, M.~Spradlin and A.~Volovich, \emph{{Symbol alphabets
  from plabic graphs}},
  \href{https://doi.org/10.1007/JHEP10(2020)128}{\emph{JHEP} {\bfseries 10}
  (2020) 128} [\href{https://arxiv.org/abs/2007.00646}{{\ttfamily
  2007.00646}}].

\bibitem{Golden:2013xva}
J.~Golden, A.B.~Goncharov, M.~Spradlin, C.~Vergu and A.~Volovich,
  \emph{{Motivic Amplitudes and Cluster Coordinates}},
  \href{https://doi.org/10.1007/JHEP01(2014)091}{\emph{JHEP} {\bfseries 01}
  (2014) 091} [\href{https://arxiv.org/abs/1305.1617}{{\ttfamily 1305.1617}}].

\bibitem{Chicherin:2020umh}
D.~Chicherin, J.M.~Henn and G.~Papathanasiou, \emph{{Cluster algebras for
  Feynman integrals}},
  \href{https://doi.org/10.1103/PhysRevLett.126.091603}{\emph{Phys. Rev. Lett.}
  {\bfseries 126} (2021) 091603}
  [\href{https://arxiv.org/abs/2012.12285}{{\ttfamily 2012.12285}}].

\bibitem{Eden:1966dnq}
R.J.~Eden, P.V.~Landshoff, D.I.~Olive and J.C.~Polkinghorne, \emph{{The
  analytic S-matrix}}, Cambridge Univ. Press, Cambridge (1966).

\bibitem{todorov2014analytic}
I.~Todorov, \emph{{Analytic Properties of Feynman Diagrams in Quantum Field
  Theory}}, International Series of Monographs in Natural Philosophy, Elsevier
  Science (2014),
  \href{https://doi.org/10.1016/C2013-0-02440-1}{10.1016/C2013-0-02440-1}.

\bibitem{timme2019mixed}
S.~Timme, \emph{Mixed precision path tracking for polynomial homotopy
  continuation},  \href{https://arxiv.org/abs/1902.02968}{{\ttfamily
  1902.02968}}.

\bibitem{telen2020robust}
S.~Telen, M.V.~Barel and J.~Verschelde, \emph{A robust numerical path tracking
  algorithm for polynomial homotopy continuation}, {\emph{SIAM Journal on
  Scientific Computing} {\bfseries 42} (2020) A3610}
  [\href{https://arxiv.org/abs/1909.04984}{{\ttfamily 1909.04984}}].

\bibitem{verschelde1999algorithm}
J.~Verschelde, \emph{{Algorithm 795: PHCpack: A general-purpose solver for
  polynomial systems by homotopy continuation}}, {\emph{ACM Transactions on
  Mathematical Software (TOMS)} {\bfseries 25} (1999) 251}.

\bibitem{bates2013numerically}
D.J.~Bates, A.J.~Sommese, J.D.~Hauenstein and C.W.~Wampler, \emph{{Numerically
  solving polynomial systems with Bertini}}, SIAM (2013).

\bibitem{10.1007/978-3-319-96418-8_54}
P.~Breiding and S.~Timme, \emph{{HomotopyContinuation.jl: A Package for
  Homotopy Continuation in Julia}},  in \emph{Mathematical Software -- ICMS
  2018}, J.H.~Davenport, M.~Kauers, G.~Labahn and J.~Urban, eds., (Cham),
  pp.~458--465, Springer International Publishing, 2018.

\bibitem{Bjorken:1959fd}
J.D.~Bjorken, \emph{{Experimental tests of Quantum electrodynamics and spectral
  representations of Green's functions in perturbation theory}}, Ph.D. thesis,
  Stanford U., 1959.

\bibitem{Landau:1959fi}
L.~Landau, \emph{{On analytic properties of vertex parts in quantum field
  theory}},
  \href{https://doi.org/10.1016/B978-0-08-010586-4.50103-6}{\emph{Nucl. Phys.}
  {\bfseries 13} (1960) 181}.

\bibitem{10.1143/PTP.22.128}
N.~Nakanishi, \emph{{Ordinary and Anomalous Thresholds in Perturbation
  Theory}}, \href{https://doi.org/10.1143/PTP.22.128}{\emph{Prog. Theor. Phys.}
  {\bfseries 22} (1959) 128}.

\bibitem{Mizera:2021ujs}
S.~Mizera, \emph{{Bounds on Crossing Symmetry}},
  \href{https://doi.org/10.1103/PhysRevD.103.L081701}{\emph{Phys. Rev. D}
  {\bfseries 103} (2021) 081701}
  [\href{https://arxiv.org/abs/2101.08266}{{\ttfamily 2101.08266}}].

\bibitem{Mizera:2021fap}
S.~Mizera, \emph{{Crossing symmetry in the planar limit}},
  \href{https://doi.org/10.1103/PhysRevD.104.045003}{\emph{Phys. Rev. D}
  {\bfseries 104} (2021) 045003}
  [\href{https://arxiv.org/abs/2104.12776}{{\ttfamily 2104.12776}}].

\bibitem{Brown:2009ta}
F.C.S.~Brown, \emph{{On the periods of some Feynman integrals}},
  \href{https://arxiv.org/abs/0910.0114}{{\ttfamily 0910.0114}}.

\bibitem{Bloch:2010gk}
S.~Bloch and D.~Kreimer, \emph{{Feynman amplitudes and Landau singularities for
  1-loop graphs}},
  \href{https://doi.org/10.4310/CNTP.2010.v4.n4.a4}{\emph{Commun. Num. Theor.
  Phys.} {\bfseries 4} (2010) 709}
  [\href{https://arxiv.org/abs/1007.0338}{{\ttfamily 1007.0338}}].

\bibitem{Abreu:2017ptx}
S.~Abreu, R.~Britto, C.~Duhr and E.~Gardi, \emph{{Cuts from residues: the
  one-loop case}}, \href{https://doi.org/10.1007/JHEP06(2017)114}{\emph{JHEP}
  {\bfseries 06} (2017) 114}
  [\href{https://arxiv.org/abs/1702.03163}{{\ttfamily 1702.03163}}].

\bibitem{Schultka:2019tfi}
K.~Schultka, \emph{{Microlocal analyticity of Feynman integrals}}, Ph.D.
  thesis, Humboldt U., Berlin, 2019.
\newblock 10.18452/20161.

\bibitem{Collins:2020euz}
J.~Collins, \emph{{A new and complete proof of the Landau condition for pinch
  singularities of Feynman graphs and other integrals}},
  \href{https://arxiv.org/abs/2007.04085}{{\ttfamily 2007.04085}}.

\bibitem{Berghoff:2020bug}
M.~Berghoff and D.~Kreimer, \emph{{Graph complexes and Feynman rules}},
  \href{https://arxiv.org/abs/2008.09540}{{\ttfamily 2008.09540}}.

\bibitem{Muhlbauer:2020kut}
M.~M\"uhlbauer, \emph{{Momentum Space Landau Equations Via Isotopy
  Techniques}},  \href{https://arxiv.org/abs/2011.10368}{{\ttfamily
  2011.10368}}.

\bibitem{HMSV}
H.S.~Hannesdottir, A.J.~McLeod, M.D.~Schwartz and C.~Vergu, \emph{{Implications
  of the Landau Equations for Iterated Integrals}}, .

\bibitem{Bendle:2020iim}
D.~Bendle, J.~Boehm, W.~Decker, A.~Georgoudis, F.-J.~Pfreundt, M.~Rahn et~al.,
  \emph{{Module Intersection for the Integration-by-Parts Reduction of
  Multi-Loop Feynman Integrals}},  in \emph{{MathemAmplitudes 2019:
  Intersection Theory and Feynman Integrals}}, 10, 2020
  [\href{https://arxiv.org/abs/2010.06895}{{\ttfamily 2010.06895}}].

\bibitem{gelfand2008discriminants}
I.M.~Gelfand, M.~Kapranov and A.~Zelevinsky, \emph{Discriminants, resultants,
  and multidimensional determinants}, Springer Science \& Business Media
  (2008).

\bibitem{M2}
D.R.~Grayson and M.E.~Stillman, ``Macaulay2, a software system for research in
  algebraic geometry.'' Available at \url{http://www.math.uiuc.edu/Macaulay2/}.

\bibitem{hauenstein2010witness}
J.D.~Hauenstein and A.J.~Sommese, \emph{Witness sets of projections},
  {\emph{Applied Mathematics and Computation} {\bfseries 217} (2010) 3349}.

\bibitem{duff2019solving}
T.~Duff, C.~Hill, A.~Jensen, K.~Lee, A.~Leykin and J.~Sommars, \emph{Solving
  polynomial systems via homotopy continuation and monodromy}, {\emph{IMA
  Journal of Numerical Analysis} {\bfseries 39} (2019) 1421}
  [\href{https://arxiv.org/abs/1609.08722}{{\ttfamily 1609.08722}}].

\bibitem{Mastrolia:2018uzb}
P.~Mastrolia and S.~Mizera, \emph{{Feynman Integrals and Intersection Theory}},
  \href{https://doi.org/10.1007/JHEP02(2019)139}{\emph{JHEP} {\bfseries 02}
  (2019) 139} [\href{https://arxiv.org/abs/1810.03818}{{\ttfamily
  1810.03818}}].

\bibitem{michalek2021invitation}
M.~Michalek and B.~Sturmfels, \emph{Invitation to nonlinear algebra}, vol.~211,
  American Mathematical Soc. (2021).

\bibitem{doi:10.1063/1.1724262}
D.B.~Fairlie, P.V.~Landshoff, J.~Nuttall and J.C.~Polkinghorne,
  \emph{{Singularities of the Second Type}},
  \href{https://doi.org/10.1063/1.1724262}{\emph{Journal of Mathematical
  Physics} {\bfseries 3} (1962) 594}.

\bibitem{cox2013ideals}
D.~Cox, J.~Little and D.~O'Shea, \emph{Ideals, varieties, and algorithms: an
  introduction to computational algebraic geometry and commutative algebra},
  Springer Science \& Business Media (2013).

\bibitem{MR3100243}
I.R.~Shafarevich, \emph{Basic algebraic geometry. 1}, Springer, Heidelberg,
  third~ed. (2013).

\bibitem{stewart1991perturbation}
G.W.~Stewart, \emph{{Perturbation theory for the singular value
  decomposition}}, {\emph{SVD and Signal Processing, II: Algorithms, Analysis
  and Applications} (1991) 99}.

\bibitem{hauenstein2018numerical}
J.D.~Hauenstein, J.I.~Rodriguez and F.~Sottile, \emph{{Numerical computation of
  Galois groups}}, {\emph{Foundations of Computational Mathematics} {\bfseries
  18} (2018) 867} [\href{https://arxiv.org/abs/1605.07806}{{\ttfamily
  1605.07806}}].

\bibitem{doi:10.1063/1.1664557}
C.~Risk, \emph{{Analyticity of the Envelope Diagrams}},
  \href{https://doi.org/10.1063/1.1664557}{\emph{J. Math. Phys.} {\bfseries 9}
  (1968) 2168}.

\bibitem{doi:10.1063/1.1703752}
R.J.~Eden, P.V.~Landshoff, J.C.~Polkinghorne and J.C.~Taylor, \emph{{Acnodes
  and Cusps on Landau Curves}},
  \href{https://doi.org/10.1063/1.1703752}{\emph{J. Math. Phys.} {\bfseries 2}
  (1961) 656}.

\bibitem{PhysRev.117.1159}
T.W.B.~Kibble, \emph{{Kinematics of General Scattering Processes and the
  Mandelstam Representation}},
  \href{https://doi.org/10.1103/PhysRev.117.1159}{\emph{Phys. Rev.} {\bfseries
  117} (1960) 1159}.

\bibitem{Coleman:1965xm}
S.~Coleman and R.~Norton, \emph{{Singularities in the physical region}},
  \href{https://doi.org/10.1007/BF02750472}{\emph{Nuovo Cim.} {\bfseries 38}
  (1965) 438}.

\bibitem{cox2011toric}
D.A.~Cox, J.B.~Little and H.K.~Schenck, \emph{Toric varieties}, vol.~124,
  American Mathematical Soc. (2011).

\bibitem{telen2020thesis}
S.~Telen, \emph{Solving systems of polynomial equations (doctoral dissertation,
  {KU} {L}euven, {L}euven, {B}elgium)}, {\emph{available at
  \url{https://simontelen.webnode.com/publications/}} (2020) }.

\bibitem{delaCruz:2019skx}
L.~de~la Cruz, \emph{{Feynman integrals as A-hypergeometric functions}},
  \href{https://doi.org/10.1007/JHEP12(2019)123}{\emph{JHEP} {\bfseries 12}
  (2019) 123} [\href{https://arxiv.org/abs/1907.00507}{{\ttfamily
  1907.00507}}].

\bibitem{Klausen:2019hrg}
R.P.~Klausen, \emph{{Hypergeometric Series Representations of Feynman Integrals
  by GKZ Hypergeometric Systems}},
  \href{https://doi.org/10.1007/JHEP04(2020)121}{\emph{JHEP} {\bfseries 04}
  (2020) 121} [\href{https://arxiv.org/abs/1910.08651}{{\ttfamily
  1910.08651}}].

\bibitem{Feng:2019bdx}
T.-F.~Feng, C.-H.~Chang, J.-B.~Chen and H.-B.~Zhang, \emph{{GKZ-hypergeometric
  systems for Feynman integrals}},
  \href{https://doi.org/10.1016/j.nuclphysb.2020.114952}{\emph{Nucl. Phys. B}
  {\bfseries 953} (2020) 114952}
  [\href{https://arxiv.org/abs/1912.01726}{{\ttfamily 1912.01726}}].

\bibitem{Tellander:2021xdz}
F.~Tellander and M.~Helmer, \emph{{Cohen-Macaulay Property of Feynman
  Integrals}},  \href{https://arxiv.org/abs/2108.01410}{{\ttfamily
  2108.01410}}.

\bibitem{duff2020polyhedral}
T.~Duff, S.~Telen, E.~Walker and T.~Yahl, \emph{{Polyhedral Homotopies in Cox
  Coordinates}},  \href{https://arxiv.org/abs/2012.04255}{{\ttfamily
  2012.04255}}.

\bibitem{bender2020toric}
M.R.~Bender and S.~Telen, \emph{{Toric eigenvalue methods for solving sparse
  polynomial systems}},  \href{https://arxiv.org/abs/2006.10654}{{\ttfamily
  2006.10654}}.

\bibitem{telen2020numerical}
S.~Telen, \emph{Numerical root finding via {C}ox rings}, {\emph{Journal of Pure
  and Applied Algebra} {\bfseries 224} (2020) 106367}
  [\href{https://arxiv.org/abs/1903.12002}{{\ttfamily 1903.12002}}].

\bibitem{10.1143/PTPS.18.1}
N.~Nakanishi, \emph{{Parametric Integral Formulas and Analytic Properties in
  Perturbation Theory}},
  \href{https://doi.org/10.1143/PTPS.18.1}{\emph{Progress of Theoretical
  Physics Supplement} {\bfseries 18} (1961) 1}.

\bibitem{Brown:2015fyf}
F.~Brown, \emph{{Feynman amplitudes, coaction principle, and cosmic Galois
  group}}, \href{https://doi.org/10.4310/CNTP.2017.v11.n3.a1}{\emph{Commun.
  Num. Theor. Phys.} {\bfseries 11} (2017) 453}
  [\href{https://arxiv.org/abs/1512.06409}{{\ttfamily 1512.06409}}].

\bibitem{AHHM}
N.~Arkani-Hamed, A.~Hillman and S.~Mizera, \emph{{Feynman polytopes and the
  tropical geometry of UV and IR divergences}},
  \href{https://doi.org/10.1103/PhysRevD.105.125013}{\emph{Phys. Rev. D}
  {\bfseries 105} (2022) 125013}
  [\href{https://arxiv.org/abs/2202.12296}{{\ttfamily 2202.12296}}].

\bibitem{Schultka:2018nrs}
K.~Schultka, \emph{{Toric geometry and regularization of Feynman integrals}},
  \href{https://arxiv.org/abs/1806.01086}{{\ttfamily 1806.01086}}.

\bibitem{Bloch:2005bh}
S.~Bloch, H.~Esnault and D.~Kreimer, \emph{{On Motives associated to graph
  polynomials}}, \href{https://doi.org/10.1007/s00220-006-0040-2}{\emph{Commun.
  Math. Phys.} {\bfseries 267} (2006) 181}
  [\href{https://arxiv.org/abs/math/0510011}{{\ttfamily math/0510011}}].

\bibitem{Pak:2010pt}
A.~Pak and A.~Smirnov, \emph{{Geometric approach to asymptotic expansion of
  Feynman integrals}},
  \href{https://doi.org/10.1140/epjc/s10052-011-1626-1}{\emph{Eur. Phys. J. C}
  {\bfseries 71} (2011) 1626}
  [\href{https://arxiv.org/abs/1011.4863}{{\ttfamily 1011.4863}}].

\bibitem{Ananthanarayan:2018tog}
B.~Ananthanarayan, A.~Pal, S.~Ramanan and R.~Sarkar, \emph{{Unveiling Regions
  in multi-scale Feynman Integrals using Singularities and Power Geometry}},
  \href{https://doi.org/10.1140/epjc/s10052-019-6533-x}{\emph{Eur. Phys. J. C}
  {\bfseries 79} (2019) 57} [\href{https://arxiv.org/abs/1810.06270}{{\ttfamily
  1810.06270}}].

\bibitem{Semenova:2018cwy}
T.Y.~Semenova, A.V.~Smirnov and V.A.~Smirnov, \emph{{On the status of expansion
  by regions}},
  \href{https://doi.org/10.1140/epjc/s10052-019-6653-3}{\emph{Eur. Phys. J. C}
  {\bfseries 79} (2019) 136}
  [\href{https://arxiv.org/abs/1809.04325}{{\ttfamily 1809.04325}}].

\bibitem{Panzer:2019yxl}
E.~Panzer, \emph{{Hepp's bound for Feynman graphs and matroids}},
  \href{https://arxiv.org/abs/1908.09820}{{\ttfamily 1908.09820}}.

\bibitem{Borinsky:2020rqs}
M.~Borinsky, \emph{{Tropical Monte Carlo quadrature for Feynman integrals}},
  \href{https://arxiv.org/abs/2008.12310}{{\ttfamily 2008.12310}}.

\bibitem{gawrilow2000polymake}
E.~Gawrilow and M.~Joswig, \emph{Polymake: a framework for analyzing convex
  polytopes},  in \emph{Polytopes—combinatorics and computation}, pp.~43--73,
  Springer, 2000.

\bibitem{Klemm:2019dbm}
A.~Klemm, C.~Nega and R.~Safari, \emph{{The $l$-loop Banana Amplitude from GKZ
  Systems and relative Calabi-Yau Periods}},
  \href{https://doi.org/10.1007/JHEP04(2020)088}{\emph{JHEP} {\bfseries 04}
  (2020) 088} [\href{https://arxiv.org/abs/1912.06201}{{\ttfamily
  1912.06201}}].

\bibitem{kaluba2020polymake}
M.~Kaluba, B.~Lorenz and S.~Timme, \emph{{Polymake.jl: A New Interface to
  polymake}},  in \emph{International Congress on Mathematical Software},
  pp.~377--385, Springer, 2020
  [\href{https://arxiv.org/abs/2003.11381}{{\ttfamily 2003.11381}}].

\bibitem{Postnikov06faces}
A.~Postnikov, V.~Reiner and L.~Williams, \emph{Faces of generalized
  permutohedra}, {\emph{Doc. Math.} {\bfseries 13} (2008) 207–273}
  [\href{https://arxiv.org/abs/math/0609184}{{\ttfamily math/0609184}}].

\bibitem{huh2013maximum}
J.~Huh, \emph{The maximum likelihood degree of a very affine variety},
  {\emph{Compositio Mathematica} {\bfseries 149} (2013) 1245}
  [\href{https://arxiv.org/abs/1207.0553}{{\ttfamily 1207.0553}}].

\bibitem{Sturmfels:2020mpv}
B.~Sturmfels and S.~Telen, \emph{{Likelihood Equations and Scattering
  Amplitudes}},  \href{https://arxiv.org/abs/2012.05041}{{\ttfamily
  2012.05041}}.

\bibitem{agostini2021likelihood}
D.~Agostini, T.~Brysiewicz, C.~Fevola, L.~K{\"u}hne, B.~Sturmfels and S.~Telen,
  \emph{{Likelihood Degenerations}},
  \href{https://arxiv.org/abs/2107.10518}{{\ttfamily 2107.10518}}.

\bibitem{Liu:2018brz}
Z.~Liu and X.~Zhao, \emph{{Bootstrapping solutions of scattering equations}},
  \href{https://doi.org/10.1007/JHEP02(2019)071}{\emph{JHEP} {\bfseries 02}
  (2019) 071} [\href{https://arxiv.org/abs/1810.00384}{{\ttfamily
  1810.00384}}].

\bibitem{AomotoKita}
K.~Aomoto and M.~Kita, \emph{Theory of Hypergeometric Functions}, Springer
  Monographs in Mathematics, Springer Japan (2011).

\bibitem{10.2969/jmsj/02720248}
K.~Aomoto, \emph{{On vanishing of cohomology attached to certain many valued
  meromorphic functions}},
  \href{https://doi.org/10.2969/jmsj/02720248}{\emph{Journal of the
  Mathematical Society of Japan} {\bfseries 27} (1975) 248 }.

\bibitem{Frellesvig:2019uqt}
H.~Frellesvig, F.~Gasparotto, M.K.~Mandal, P.~Mastrolia, L.~Mattiazzi and
  S.~Mizera, \emph{{Vector Space of Feynman Integrals and Multivariate
  Intersection Numbers}},
  \href{https://doi.org/10.1103/PhysRevLett.123.201602}{\emph{Phys. Rev. Lett.}
  {\bfseries 123} (2019) 201602}
  [\href{https://arxiv.org/abs/1907.02000}{{\ttfamily 1907.02000}}].

\bibitem{Mizera:2019vvs}
S.~Mizera and A.~Pokraka, \emph{{From Infinity to Four Dimensions: Higher
  Residue Pairings and Feynman Integrals}},
  \href{https://doi.org/10.1007/JHEP02(2020)159}{\emph{JHEP} {\bfseries 02}
  (2020) 159} [\href{https://arxiv.org/abs/1910.11852}{{\ttfamily
  1910.11852}}].

\bibitem{Lee:2013hzt}
R.N.~Lee and A.A.~Pomeransky, \emph{{Critical points and number of master
  integrals}}, \href{https://doi.org/10.1007/JHEP11(2013)165}{\emph{JHEP}
  {\bfseries 11} (2013) 165} [\href{https://arxiv.org/abs/1308.6676}{{\ttfamily
  1308.6676}}].

\bibitem{Bitoun:2017nre}
T.~Bitoun, C.~Bogner, R.P.~Klausen and E.~Panzer, \emph{{Feynman integral
  relations from parametric annihilators}},
  \href{https://doi.org/10.1007/s11005-018-1114-8}{\emph{Lett. Math. Phys.}
  {\bfseries 109} (2019) 497}
  [\href{https://arxiv.org/abs/1712.09215}{{\ttfamily 1712.09215}}].

\bibitem{Mizera:2017rqa}
S.~Mizera, \emph{{Scattering Amplitudes from Intersection Theory}},
  \href{https://doi.org/10.1103/PhysRevLett.120.141602}{\emph{Phys. Rev. Lett.}
  {\bfseries 120} (2018) 141602}
  [\href{https://arxiv.org/abs/1711.00469}{{\ttfamily 1711.00469}}].

\bibitem{Frellesvig:2019kgj}
H.~Frellesvig, F.~Gasparotto, S.~Laporta, M.K.~Mandal, P.~Mastrolia,
  L.~Mattiazzi et~al., \emph{{Decomposition of Feynman Integrals on the Maximal
  Cut by Intersection Numbers}},
  \href{https://doi.org/10.1007/JHEP05(2019)153}{\emph{JHEP} {\bfseries 05}
  (2019) 153} [\href{https://arxiv.org/abs/1901.11510}{{\ttfamily
  1901.11510}}].

\bibitem{Weinzierl:2020xyy}
S.~Weinzierl, \emph{{On the computation of intersection numbers for twisted
  cocycles}},  \href{https://arxiv.org/abs/2002.01930}{{\ttfamily 2002.01930}}.

\bibitem{Frellesvig:2020qot}
H.~Frellesvig, F.~Gasparotto, S.~Laporta, M.K.~Mandal, P.~Mastrolia,
  L.~Mattiazzi et~al., \emph{{Decomposition of Feynman Integrals by
  Multivariate Intersection Numbers}},
  \href{https://doi.org/10.1007/JHEP03(2021)027}{\emph{JHEP} {\bfseries 03}
  (2021) 027} [\href{https://arxiv.org/abs/2008.04823}{{\ttfamily
  2008.04823}}].

\bibitem{Caron-Huot:2021xqj}
S.~Caron-Huot and A.~Pokraka, \emph{{Duals of Feynman Integrals, I:
  Differential Equations}},  \href{https://arxiv.org/abs/2104.06898}{{\ttfamily
  2104.06898}}.

\bibitem{breiding2020certifying}
P.~Breiding, K.~Rose and S.~Timme, \emph{{Certifying zeros of polynomial
  systems using interval arithmetic}},
  \href{https://arxiv.org/abs/2011.05000}{{\ttfamily 2011.05000}}.

\bibitem{Kalmykov:2016lxx}
M.Y.~Kalmykov and B.A.~Kniehl, \emph{{Counting the number of master integrals
  for sunrise diagrams via the Mellin-Barnes representation}},
  \href{https://doi.org/10.1007/JHEP07(2017)031}{\emph{JHEP} {\bfseries 07}
  (2017) 031} [\href{https://arxiv.org/abs/1612.06637}{{\ttfamily
  1612.06637}}].

\bibitem{Boyling1968}
J.B.~Boyling, \emph{A homological approach to parametric feynman integrals},
  \href{https://doi.org/10.1007/BF02800115}{\emph{Il Nuovo Cimento A
  (1965-1970)} {\bfseries 53} (1968) 351}.

\bibitem{AIF_2003__53_4_977_0}
K.~Aomoto, \emph{{Gauss-Manin connections of Schl\"afli type for hypersphere
  arrangements}}, \href{https://doi.org/10.5802/aif.1970}{\emph{Annales de
  l'Institut Fourier} {\bfseries 53} (2003) 977}.

\bibitem{sturmfels1994newton}
B.~Sturmfels, \emph{On the {N}ewton polytope of the resultant}, {\emph{Journal
  of Algebraic Combinatorics} {\bfseries 3} (1994) 207}.

\end{thebibliography}\endgroup
